\documentclass[%
 reprint,
superscriptaddress,
nofootinbib,
 amsmath,amssymb,
 aps,
]{revtex4-2}
\usepackage{amsmath,amssymb,amsthm,bm,mathrsfs,mathtools,braket,physics}

\usepackage{tikz}
\usepackage{array}

\usepackage{ascmac}
\usepackage{xcolor}
\usepackage{stmaryrd}
\usepackage{enumitem}
\usepackage{orcidlink}
\usepackage{hyperref}
\usepackage[capitalise]{cleveref}
\usepackage{algorithmicx}
\usepackage{algpseudocode}
\usepackage[normalem]{ulem}

\numberwithin{equation}{section}

\newcommand{\LER}{\mathrm{LER}}
\newcommand{\MAX}{\mathrm{max}}
\newcommand{\IN}{\mathrm{in}}
\newcommand{\OUT}{\mathrm{out}}
\newcommand{\BF}{\mathrm{bf}}
\newcommand{\PF}{\mathrm{pf}}
\newcommand{\IND}{\mathrm{Ind}}
\newcommand{\Ber}{\mathrm{Bernoulli}}
\newcommand{\MW}{\mathrm{MW}}
\newcommand{\ML}{\mathrm{ML}}
\DeclareMathOperator*{\argmin}{arg\,min}
\DeclareMathOperator*{\argmax}{arg\,max}
\DeclareMathOperator*{\wt}{wt}
\DeclareMathOperator{\mathspan}{span}

\newcommand{\red}[1]{\textcolor{red}{#1}}

\theoremstyle{definition}
\newtheorem{theorem}{Theorem}[section]
\newtheorem{definition}[theorem]{Definition}
\newtheorem{lemma}[theorem]{Lemma}

\newtheorem{example}[theorem]{Example}

\makeatletter
\newcommand{\equalcontrib}{%
  \frontmatter@footnote{These authors contributed equally to this work.}%
}
\makeatother

\newcommand{\Square}[3][]{
  \draw[#1] #2 rectangle ++(#3,#3);
}

\newcommand{\Rect}[4][]{
  \draw[#1] #2 rectangle ++(#3,#4);
}

\newcommand{\SmallCircle}[5][]{
\pgfmathsetmacro{\GrayPct}{#4*100}
  \filldraw[
    draw=black,
    fill=black!\GrayPct,
    fill opacity=1.0,
    draw opacity=1.0,
    #1
  ] ({#2},{#3}) circle ({#5});
}

\newcommand{\DressedPlaquette}[4]{
  \Square[draw=#4]{(#1, #2)}{#3}

  \pgfmathsetmacro{\pr}{0.075}

  \SmallCircle[draw=#4]{#1}{#2}{1}{\pr}
  \SmallCircle[draw=#4]{#1+#3}{#2}{1}{\pr}
  \SmallCircle[draw=#4]{#1+#3}{#2+#3}{1}{\pr}
  \SmallCircle[draw=#4]{#1}{#2+#3}{1}{\pr}
}

\newcommand{\LPlaquette}[5]{
  \pgfmathsetmacro{\pr}{0.075}
  \pgfmathsetmacro{\prr}{#4}
  \Rect[draw=#5]{(#1+\prr, #2)}{#3-\prr}{#3}
  \draw[color=#5] (#1 - \prr, #2) -- (#1 - \prr, #2 + #3);
}

\newcommand{\RPlaquette}[5]{
  \pgfmathsetmacro{\pr}{0.075}
  \pgfmathsetmacro{\prr}{#4}
  \Rect[draw=#5]{(#1, #2)}{#3-\prr}{#3}
  \draw[color=#5] (#1 + #3 + \prr, #2) -- (#1 + #3 + \prr, #2 + #3);
}

\definecolor{colorA}{HTML}{E6194B} % Red
\definecolor{colorB}{HTML}{3CB44B} % Green
\definecolor{colorC}{HTML}{FFE119} % Yellow
\definecolor{colorD}{HTML}{4363D8} % Blue
\definecolor{colorE}{HTML}{F58231} % Orange
\definecolor{colorF}{HTML}{911EB4} % Purple

\begin{document}

\title{Arbitrary-Distance Quantum Error Correction \\ with Gauss's Law for \texorpdfstring{$\mathbb Z_2$}{Z_2} Lattice Gauge Theory}%

\author{Neel S.~Modi\,\orcidlink{0009-0002-8308-9882}\equalcontrib}
\email{neel\_modi@berkeley.edu}
\affiliation{Physics Division, Lawrence Berkeley National Laboratory, Berkeley, CA 94720, USA}
\affiliation{Leinweber Institute for Theoretical Physics and Department of Physics,\\ University of California, Berkeley, Berkeley, CA 94720, USA}
\author{Lento~Nagano\,\orcidlink{0000-0003-4515-2115}\equalcontrib}
\email{lento.nagano@keio.jp}
\affiliation{Graduate School of Science and Technology, Keio University,\\
Yokohama, Kanagawa 223-8522, Japan}
\affiliation{International Center for Elementary Particle Physics (ICEPP),\\ The University of Tokyo, 7-3-1 Hongo, Bunkyo-ku, Tokyo 113-0033, Japan}

\author{Masazumi~Honda\,\orcidlink{0000-0001-6935-5609}}
\email{masazumi.honda@riken.jp}
\affiliation{RIKEN Center for Interdisciplinary Theoretical and Mathematical Sciences (iTHEMS), RIKEN, 2-1 Hirosawa, Wako, Saitama 351-0198 Japan}
\affiliation{Graduate School of Science and Engineering, Saitama University, 255 Shimo-Okubo, Sakura-ku, Saitama 338-8570, Japan}
\author{Nobuyuki~Yoshioka\,\orcidlink{0000-0001-6094-8635}}
\email{ny.nobuyoshioka@gmail.com}
\affiliation{International Center for Elementary Particle Physics (ICEPP),\\ The University of Tokyo, 7-3-1 Hongo, Bunkyo-ku, Tokyo 113-0033, Japan}

\author{Christian W.~Bauer\,\orcidlink{0000-0001-9820-5810}}
\email{cwbauer@lbl.gov}
\affiliation{Physics Division, Lawrence Berkeley National Laboratory, Berkeley, CA 94720, USA}
\affiliation{Leinweber Institute for Theoretical Physics and Department of Physics,\\ University of California, Berkeley, Berkeley, CA 94720, USA}

\date{\today}
\begin{abstract}
    It has previously been shown~\cite{Rajput:2021trn} that $\mathbb Z_2$ Gauss's law constraints can be used to build efficient quantum error-correcting codes (QECCs) that are robust against arbitrary single-qubit errors. In this work, we generalize the construction to be robust against arbitrary $t$-qubit errors, where $t$ is any positive integer. This includes a derivation of the optimal Gauss's law code within the considered family by minimizing the number of physical qubits required for a given code distance. Finally, we compare our codes against other efficient QECCs on metrics such as the number of physical qubits, the locality of the encoded Hamiltonian, and the logical error rate in the code capacity setting.
    Compared to using a domain-agnostic code for every lattice degree of freedom, we find that the Gauss's law code primarily excels at reducing the locality of the encoded Hamiltonian. Moreover, the physical qubit overhead is also reduced for $t \le 3$ (distance $d \le 7$).
\end{abstract}

\maketitle

\tableofcontents

\section{Introduction}\label{sec:intro}

Gauge theory describes fundamental interactions between elementary particles.
Studying the dynamics of a given gauge theory is relevant to understanding how non-perturbative phenomena such as confinement and string breaking emerge.
Lattice gauge theory (LGT) regulates the continuum gauge theory described by quantum field theory by putting it on the discretized lattice space(time), and allows us to investigate non-perturbative physics numerically~\cite{Wilson:1974sk,Kogut:1975,Kogut:1979wt}.
While the conventional Monte Carlo approach has provided illuminating results, it suffers from the sign problem~\cite{Troyer:2004ge} when dealing with chemical potentials, topological terms, and real-time dynamics.
To overcome this problem, we have recently seen a growth of interest in Hamiltonian simulations of LGT using quantum computing (reviews of the many different approaches, both analogue and digital, can be found in Refs.~\cite{Wiese:2013uua,Zohar:2015hwa,Dalmonte:2016alw,Aidelsburger:2021mia,Banuls:2019bmf,Zohar:2021nyc,Klco:2021lap,Bauer:2022hpo}).

In the Hamiltonian setting, one typically works in Weyl ($A_0 = 0$) gauge.
In this gauge, Gauss's law is not enforced through an equation of motion, but is instead enforced through constraints on the Hilbert space. 
We can impose these constraints by completely or partly solving them at the kinematical level and/or fixing the gauge redundancy, and rewriting the Hamiltonian in terms of the remaining degrees of freedom~\cite{Kaplan:2018vnj,Zohar:2019ygc,Raychowdhury:2018osk,Davoudi:2020yln,DAndrea:2023qnr,Grabowska:2024emw,Drell:1978hr,Bauer:2021gek,Haase:2020kaj}.
This approach automatically guarantees gauge-invariant dynamics and can reduce the number of qubits. However, it may introduce longer-range or more complicated interactions in some formulations, which may cause increasing gate costs~\cite{Hamer:1997dx,Martinez:2016yna,Zohar:2019ygc,DAndrea:2023qnr, Grabowska:2024emw,Grabowska:2022uos}.
For $\mathbb{Z}_2$ gauge theories coupled to matter, this strategy can be used to eliminate the fermionic matter entirely and recast the model as a pure spin system~\cite{Zohar:2018cwb,Zohar:2019ygc,Borla:2019chl,Borla:2020upy}.
The alternative is to retain the local, gauge-redundant Hamiltonian and either preserve Gauss's law at the algorithmic level~\cite{Lamm:2019bik,Mazzola:2021hma,Stryker:2021asy,Alexandrou:2025vaj,Sekiyama:2026rgt}, suppress gauge-violating errors during the evolution~\cite{Zohar:2011cw,Lamm:2020jwv,Tran:2020azk,Halimeh:2020ecg,Davoudi:2022uzo}, or mitigate their effects in measured observables through symmetry verification and postselection~\cite{Stryker:2018efp,Ballini:2024qmr,Bonet-Monroig:2018mpi}.
In this approach, one requires additional qubits to represent the enlarged Hilbert space, and one needs to protect the gauge symmetry from algorithmic or hardware errors.
In particular, taming the hardware error in a scalable way is also necessary in general to have reliable results in large-scale simulations.

Quantum error correction (QEC) is a general collection of techniques to protect quantum data (in the form of \emph{logical} states) by equipping it with redundancies (typically by using \emph{physical} states in the larger Hilbert space)~\cite{Shor:1995hbe,Steane:1995vv,Calderbank:1995dw,Knill:1996ny,nielsen00,gottesman2009introductionquantumerrorcorrection,Terhal:2013vbm}.
These techniques have recently moved from principle to practice: a superconducting surface-code memory has been operated below the error-correction threshold, with the logical error rate suppressed as the code distance is increased~\cite{GoogleQuantumAIandCollaborators:2024efv}, which makes the scaling of code distance a concrete rather than an asymptotic concern.

In the stabilizer formalism~\cite{Gottesman:1997zz}, the space of logical states can be specified by the generators (\emph{stabilizers}) of an Abelian group (stabilizer group) $\mathcal{S}$, via the requirement $g\ket{\psi} = \ket{\psi}$ for all $g\in \mathcal{S}$.
Rajput, Roggero, and Wiebe constructed prototypical examples of LGT-specific stabilizer codes~\cite{Rajput:2021trn} by regarding Gauss’s law constraints in $\mathbb{Z}_2$ LGT as stabilizers and incorporating them with repetition code stabilizers.
We refer to this as the RRW construction or RRW code.
It not only detects gauge-violating errors but can also be used to correct them using standard QEC protocols.
On top of that, it was shown that the RRW code uses fewer physical qubits compared with naively encoding every lattice variable independently with a domain-agnostic code ($\llbracket 5,1,3\rrbracket$ code)~\cite{Rajput:2021trn}.
Furthermore, Ref.~\cite{Spagnoli:2024mib} proposed how to perform fault-tolerant real-time evolution in the RRW code.
Ref.~\cite{Pato:2026wow} performed memory and dynamics simulations to compare the performance of the RRW code and generic quantum error-correcting codes (QECCs) against bit-flip noise.
Other important results include extensions to $\mathbb{Z}_N$ LGT~\cite{Spagnoli:2026qni,Turco:2026cte} and the lowest-level truncation of $\mathrm{SU}(2)$ LGT~\cite{Yao:2025cxs}.
(See also~\cite{Faist:2019ahr, Klco:2021jxl, Kong:2021wau, delPino:2022zzx, Chen:2022dox, Gustafson:2023swx, Wauters:2024shc, Carrozza:2024smc, Lacambra:2026twp, rothlin2026error,Carena:2024dzu} for other studies utilizing or formulating Gauss's law for QEC.)
While these developments have enabled a better understanding of the relationship between Gauss's law and QEC, the RRW construction in $\mathbb{Z}_2$ LGT~\cite{Rajput:2021trn, Spagnoli:2024mib} focused on code distance $d=3$, meaning it can only correct up to weight-one errors.
For quantum error correction to be useful, it is important to be able to suppress higher-weight errors, as well as to obtain the error threshold under which we can control the logical error rate.

\begin{figure*}
    \centering
\includegraphics[width=0.9\linewidth]{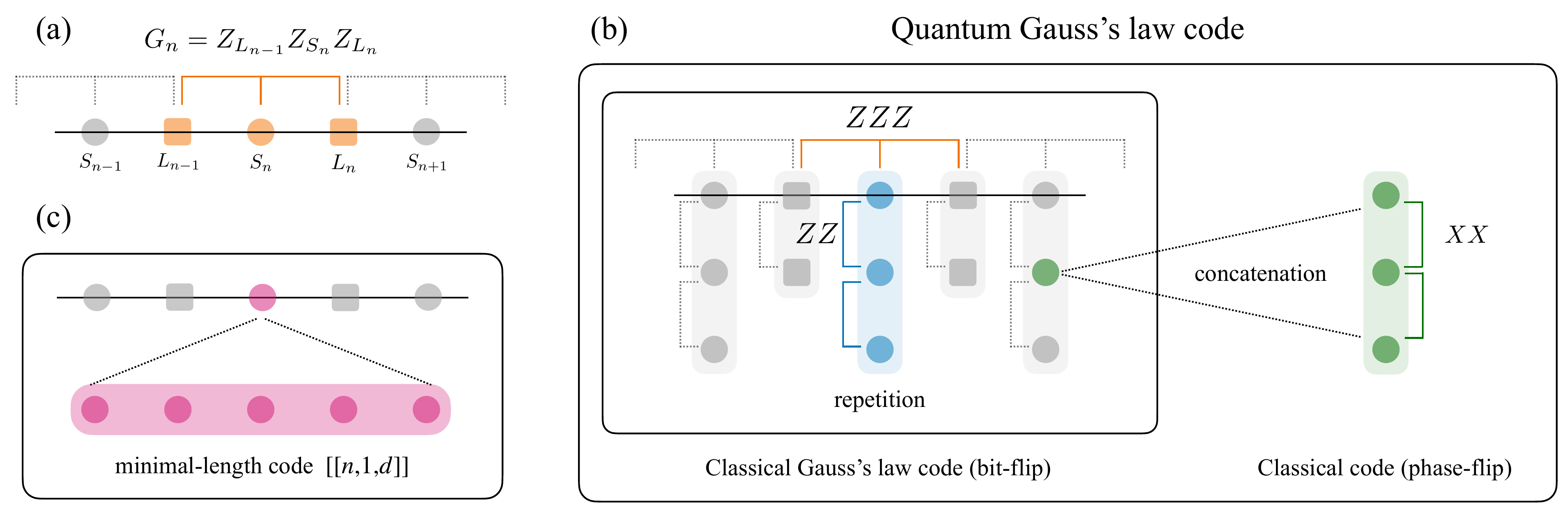}
    \caption{Summary of this paper. We depict the 1D example for simplicity, but we also give 2D results in the main text. (a) 1D lattice for $\mathbb{Z}_2$ LGT. The square symbol is a link variable (gauge field) and the circle is a site variable (matter field). Each Gauss's law constraint is a parity check between a site and two adjacent links, as highlighted in orange. (b) construction of the classical and quantum Gauss's law code. Classical Gauss's law duplicates sites and links with Gauss's law to build $Z$-type stabilizers, and is concatenated with $X$-type classical codes to yield a quantum Gauss's law code with prescribed code distance. (c) Generic (domain-agnostic) QEC encoding used as a baseline for comparison. We encode each lattice site and link by using the $\llbracket n,1,d\rrbracket$ stabilizer code with minimal $n$.}
    \label{fig:summary}
\end{figure*}
\begin{table*}[]
    \centering
    \caption{Summary of our results}
    \label{tab:results-summary}
    \begin{tabular}{c|c}
        Description & Location   \\ \hline \hline
    construction of classical Gauss's law (GL) code with arbitrary $d$ & Def.~\ref{def:gaussclassical}
    \\
    alternative construction of classical GL code & Lem.~\ref{lem:gaussclassicalalternative}
    \\
    construction of quantum GL code & Def.~\ref{def:gaussquantum} 
    \\ \hline
    GL code with optimal encoding rate $k/n$ (1D) & Thm.~\ref{thm:optimal1D} 
    \\
    GL code with optimal encoding rate $k/n$ (2D) & Thm.~\ref{thm:optimal2D} 
    \\ \hline
    qubits/locality of GL and minimal-length code (1D) & Tab.~\ref{tab:qubit-comparison-1D}, Tab.~\ref{tab:locality-comparison-1D}
    \\
    qubits/locality of GL and minimal-length code (2D) & Tab.~\ref{tab:qubit-comparison-2D}, Tab.~\ref{tab:locality-comparison-2D} 
    \\
    code capacity memory simulation of GL and minimal length code (1D) & \cref{fig:uncorrelated-XZ,fig:memory-experiment-depolarizing} 
    \end{tabular}
\end{table*}

Our main results are summarized in \cref{tab:results-summary} and \cref{fig:summary}.
We construct an explicit family of Gauss's law codes with an \emph{arbitrary} code distance for $\mathbb{Z}_2$ LGT coupled to matter, thereby extending the previous RRW construction ($d=3$) (see Def.~\ref{def:gaussclassical} and Def.~\ref{def:gaussquantum}).
As shown in \cref{fig:summary}(b), the construction first duplicates each lattice site and/or link, and then builds a $Z$-type stabilizer by combining Gauss's law with duplications.
The total number of physical qubits thus depends on the number of duplications.
We obtain the explicit construction (i.e., duplication pattern) which minimizes the number of physical qubits for a given number of logical qubits $k$ and code distance $d$ (see Thm.~\ref{thm:optimal1D} for the 1D case and Thm.~\ref{thm:optimal2D} for the 2D case).
In other words, our codes have an optimal encoding rate $k/n$ among the family of Gauss's law codes.
To assess the performance of Gauss's law codes, the number of physical qubits and the number of qubits involved in each Hamiltonian term are evaluated for concrete examples of $\mathbb Z_2$ LGTs, which are then compared with baseline (domain-agnostic) QECCs (see~\cref{fig:summary}~(c)).
The comparison results are given in ~\cref{tab:qubit-comparison-1D,tab:locality-comparison-1D} for the 1D case, and \cref{tab:qubit-comparison-2D,tab:locality-comparison-2D} for the 2D case.
As a first step toward threshold studies (see \cref{fig:uncorrelated-XZ,fig:memory-experiment-depolarizing}), we also confirm that we can suppress the logical error rate by increasing the code distance under code capacity setting.

This paper is organized as follows.
Section~\ref{sec:prelim} explains preliminaries needed for this work, including QEC (\cref{sec:QECC}) and Hamiltonian LGTs with $\mathbb Z_2$ gauge symmetry (\cref{sec:LGT}).
Readers familiar with these topics can skip this section.
Useful notations for $\mathbb{Z}_2$ LGTs are summarized in \cref{tab:notation-z2-lgt}.
Next, our theoretical results are given in \cref{sec:gauss}.
After reviewing the RRW results ($d=3$)~\cite{Rajput:2021trn} in \cref{sec:RRW}, the arbitrary distance Gauss's law code is constructed in~\cref{subsec:gauss-law-code-arbitrary-distance}.
Furthermore, the optimal Gauss's law code constructions for 1D and 2D cases are given in~\ref{sec:optimal}. 
\Cref{sec:compare} shows the comparison of code properties with baseline codes (minimal-length codes).
We compare the number of physical qubits and Hamiltonian locality in~\cref{sec:parameter-comparison}, and the code capacity memory simulation is given in~\cref{sec:LER-comparison}.
Finally, we summarize our results and provide an outlook in \cref{sec:summary}.
Additional details related to our results are provided in the appendices.

\section{Preliminary}\label{sec:prelim}

\subsection{Primer on Error-Correcting Codes}\label{sec:QECC} 

Error correction refers to a broad class of techniques that are used to protect data from being corrupted by noisy environments.
Included in these techniques are: (i) \textit{error-correcting codes} that encode logical information with redundancies; (ii) \textit{recovery operations} that aim to purify data; and (iii) strategies for \textit{scaling} to arbitrarily-high robustness (c.f. fault tolerance and thresholds).
It was discovered in the 1990s that many techniques for protecting classical data are neatly generalizable to the quantum setting ~\cite{calderbankshor1996, steane1996, Gottesman:1997zz}.
The similarities are powerful enough that it is useful for our purposes to separately review both the classical and quantum theory of error correction.
Indeed, many of the constructions introduced in this paper are formulated purely in the classical language, with the quantization being a mathematical formality.

This section is intended as a pedagogical introduction for the reader who is less familiar with these subjects.
We have chosen to restrict the discussion solely to error-correcting codes, as this is the key concept needed to understand the substance of this paper.
The only other technicality strictly needed to make sense of our results is the notion of \textit{syndrome decoding}, which is discussed in~\cref{app:syndrome-decoding}.
For the reader who is interested in learning more about error correction in the quantum setting, we highly recommend the famous notes by Gottesman \cite{gottesman2009introductionquantumerrorcorrection} and Preskill \cite{preskill1998}.

\subsubsection{Classical Linear Codes}\label{sec:classical}

Error-correcting codes were developed in the classical setting to facilitate the storing and transmission of data in the presence of noisy environments.
The idea is to convert a $k$-bit message~$m$ into a longer $n$-bit \textit{codeword}~$c(m)$, prior to exposure to noise.
That way, when noise occurs, it happens directly on the codeword $c(m)$.
If the mathematical structure of $c(m)$ is chosen cleverly, and the amount of noise experienced is reasonably small, then it can be possible to recover the exact original message $m$, in spite of the errors incurred.

Note that it is the bits of the codeword---not the message---that are directly subject to noise.
Therefore, many meaningful properties of an error-correcting code depend entirely on the set of codewords, remaining agnostic as to which codeword represents which underlying message.
For this reason, a classical code is typically defined purely by a selection of $n$-bit codewords $C$.
In the setting of binary codes, we have $C \subseteq \{0,1\}^n$.

A useful class of error-correcting codes is provided by the (binary) linear codes.
This class represents the setting in which the set of codewords $C$ is a linear subspace of the binary vector space $\mathbb F_2^n$.
When we use the term ``classical code" in this paper, we will always refer to binary linear codes.

\begin{definition}[Classical Code]\label{def:classical}
    Let $n\ge 1$ be an integer.
    Then an $n$-bit \textit{classical code} is a linear subspace $C \subseteq \mathbb F_2^n$. Elements $c\in C$ are called \emph{codewords}.
\end{definition}

\begin{example}[Classical Code]\label{ex:classical}
    Let $n = 5$, and define $C \subseteq \mathbb F_2^5$ by
    \begin{equation}
        C = \{00000, 10100, 01011, 11111\}.
    \end{equation}
    This is a linear subspace, because $C$ can also be viewed as the span of two basis vectors,
    \begin{equation}\label{eq:exclassicalspan}
        C = \mathspan\{10100, 01011\}.
    \end{equation}
    Therefore, $C$ is a $5$-bit classical code.
\end{example}

The basic function of a classical code is to be able to detect errors that may occur on codewords.
This is achieved by checking whether some received bit string is a codeword or not.
However, some errors will not be detectable.
For instance, if the bits flipped by an error are exactly those needed to transform the codeword into an entirely different codeword, then the received bit string will be recognized as a codeword, and the error will go undetected.
Definition~\ref{def:logical} summarizes the notion of undetectable errors.

\begin{definition}[Undetectable Error]\label{def:logical}
    Let $C \subseteq \mathbb F_2^n$ be a classical code, and let $c \in C$ be a codeword.
    A vector $e \in \mathbb F_2^n$ is called an \textit{undetectable error when it acts on $c$} if the following conditions are both satisfied:
    \begin{enumerate}[label=(\roman*)]
        \item $e \neq 0$; and
        \item $c + e \in C$.
    \end{enumerate}
\end{definition}

A basic fact about undetectable errors is that they are \textit{always} undetectable, irrespective of which codeword they act on.
In fact, any nonzero codeword is an undetectable error, and vice versa.
The nonzero requirement simply ensures that we're genuinely talking about an ``error," meaning a strict \textit{change} in the codeword.
This is summarized by Lemma~\ref{lem:logical}, but it follows trivially from linearity.

\begin{lemma}\label{lem:logical}
    Let $C \subseteq \mathbb F_2^n$ be a classical code, and let $c \in C$ be arbitrary.
    Then $e$ is an undetectable error when it acts on $c$ if and only if $e$ is a nonzero codeword.
\end{lemma}

Thus, we can simply write ``$e$ is an undetectable error," without referencing on which codeword it acts.

\begin{example}[Undetectable Error]\label{ex:logical}
    Continue from example \ref{ex:classical}.
    Take, for example, the codeword $c = 10100$.
    An example of an undetectable error on this codeword is obtained by flipping every single bit, which transforms $c$ into $c' = 01011$.
    This error is captured by the vector $e = 11111$, so that $c + e = c'$.
    And, as indicated by Lemma~\ref{lem:logical}, $e$ is indeed a codeword.
    
    Moreover, $e$ is an undetectable error when it acts on any codeword.
    For instance, flipping all the bits of $00000$ by $e$ produces $11111$, which is again a codeword.
\end{example}

Having discussed what exactly an undetectable error means, we can now discuss the basic properties of a classical code~$C$.
The most common properties usually considered are:
\begin{itemize}
    \item The total number of bits used by the code, also called the number of \textit{physical bits} (denoted $n$, as before).
    \item The number of \textit{logical bits} of information that can be stored in the code (denoted $k$).
    \item The minimum number of bits that are flipped by any undetectable error, also called the \textit{distance} of the code (denoted $d$).
\end{itemize}

The number of logical bits $k$ can be thought of as the maximum length of a message $m$ that can be faithfully encoded by $C$.
Mathematically, this is simply given by
\begin{equation}
    k \equiv \dim C.
\end{equation}

The code distance $d$ is a measure of how difficult it is to produce an error that is invisible to the code $C$.
Not all undetectable errors are equal in this regard; a common way to distinguish them is through their \textit{weight}.
The weight of an error $e$ is denoted $\wt(e)$, and it simply counts the number of $1$s in the binary representation of $e$.
Therefore, the distance of $C$ can be mathematically expressed as
\begin{equation}\label{eq:d}
    d = \min_{c \in C\setminus\{0\}} \wt(c).
\end{equation}

When a code $C$ has $n$ physical bits, $k$ logical bits, and code distance $d$, it is common to write the descriptor $C = [n,k,d]$ as a shorthand.
Importantly, specifying $n$, $k$, and $d$ alone is not enough to fully specify the code $C$.
Nevertheless, the descriptor notation provides a quick and easy way to build a mental model of the performance and efficiency of the code, without getting into details.
The performance roughly corresponds to the distance $d$, and the efficiency roughly corresponds to the ratio $r = k / n$, also called the \textit{encoding rate}.

\begin{example}[Properties of Classical Code]\label{ex:properties}
    Continue from Example \ref{ex:classical}.
    The number of physical bits is clearly $n = 5$.
    The number of logical bits is $k = 2$, since Eq.~\eqref{eq:exclassicalspan} indicates that there are two basis vectors.
    The distance is $d = 2$, because although we discussed the undetectable error $e = 11111$ in Example~\ref{ex:logical}, this is not the undetectable error with smallest weight---that would be $10100$, with weight $2$.
    Based on the above observations, this code can be written with the shorthand $C = [5,2,2]$.
\end{example}

In order to detect (and eventually correct) errors in a classical code $C$, it is necessary to be able to quickly determine whether a given $n$-bit string $v$ is a codeword or not.
This knowledge could be obtained from looking at the full list of codewords, but that is not practical for larger codes.
A more efficient approach is to provide a set of linear constraints on the physical bits such that $C$ is exactly the linear subspace that satisfies all of those constraints.
If all linear constraints are satisfied, then $v \in C$; if even a single linear constraint is not satisfied, then~$v \notin C$.
In this setting, the linear constraints are called \textit{parity checks} (or simply ``checks" for short).

If $C\subseteq \mathbb F_2^n$ is $k$-dimensional, then one needs exactly $(n-k)$ independent parity checks.
One can always arrange for these checks to be homogeneous, so that they take the form
\begin{equation}\label{eq:PCM}
    Hc = 0,
\end{equation}
where $H$ is an $(n-k) \times n$ binary matrix, called the \textit{parity-check matrix}, such that the code $C$ coincides exactly with the set of vectors $c \in \mathbb F_2^n$ that satisfy Eq.~\eqref{eq:PCM}.
Conversely, any $(n-k) \times n$ binary parity-check matrix $H$ can itself be used to define the code $C$.
Note that the same code $C$ may admit several parity-check matrices, which are thus considered equivalent.

\begin{definition}[Parity-Check Matrix]\label{def:parity}
    Let $C = [n,k,d]$ be a classical code.
Then $H \in \mathbb F_2^{(n-k)\times n}$ is called a \textit{(full-rank)~\footnote{It is possible that a parity check matrix $H\in\mathbb{F}_2^{m\times n}$ is not full-rank. In this case,  $k=n-\text{rank}(H)$.} parity-check matrix} for $C$ if
    \begin{equation}
        H c = 0 \Leftrightarrow c \in C.
    \end{equation}
    Conversely, any full-rank $H \in \mathbb F_2^{(n-k)\times n}$ defines a classical code $C$ by
    \begin{equation}
        C \equiv \{c \in \mathbb F_2^n \mid Hc = 0\}.
    \end{equation}
\end{definition}

\begin{example}[Parity-Check Matrix]\label{ex:parity}
    Take $C = [5,2,2]$ from example \ref{ex:properties}, and write out the bit string for each codeword as $c = c_1 c_2 c_3 c_4 c_5$.
    The linear constraints
    \begin{align}\label{eq:exparityconstraints}
        c_1 &= c_3,\\
        c_2 &= c_4 = c_5,
    \end{align}
    are satisfied for every codeword.
    Moreover, every $5$-bit string that satisfies these constraints is also a codeword.
    There are $3$ independent constraints in \eqref{eq:exparityconstraints}, and they can be made homogeneous as follows:
    \begin{align}\label{eq:exparityhomogeneous}
        c_1 + c_3 &= 0,\\
        c_2 + c_4 &= 0,\\
        c_4 + c_5 &= 0,
    \end{align}
    where addition is performed in $\mathbb F_2$.
    The matrix form of \eqref{eq:exparityhomogeneous} is precisely $Hc = 0$, where
    \begin{equation}\label{eq:exparityH}
        H =
        \begin{pmatrix}
            1 & 0 & 1 & 0 & 0\\
            0 & 1 & 0 & 1 & 0\\
            0 & 0 & 0 & 1 & 1
        \end{pmatrix}
    \end{equation}
    is the parity-check matrix.
    Note that this matrix is not unique; for instance, we could have replaced the check $c_4 + c_5 = 0$ with $c_2 + c_5 = 0$, which would change the third row of $H$.

    To test whether $v = 10111$ is a codeword, simply compute $Hv = 010$.
    Since this is nonzero, $v$ cannot be a codeword.
    The quantity $Hv$ is also sometimes denoted $\sigma(v)$ and is called the \textit{(error) syndrome} of $v$.
    In addition to providing insight on whether an error has occurred on a received word, the syndrome is often useful for deducing (or \textit{decoding}) which specific error occurred, thereby enabling recovery of the original encoded data.
    
    Notably, this ``syndrome decoding" process is largely a separate undertaking from the code construction itself, and we do not discuss it in the main text to avoid distracting from the central story of this paper.
    We leave the task of defining error syndromes and decoding strategies to~\cref{app:syndrome-decoding}, which contains all technical details required to mathematically specify our code capacity simulations.
\end{example}

\subsubsection{Quantum Codes in the Stabilizer Formalism}
\label{sec:quantum_code_stabilizer}

Analogously to the classical case, quantum codes aim to encode a $k$-qubit message state $\ket{\psi_m}$ into an $n$-qubit code state $\ket{\psi_c}$.
It is then only the code state $\ket{\psi_c}$ that is exposed to noise, meaning that the ability to protect against errors is primarily a function of the allowed set of code states.

We will focus on the special case where the set of code states forms a sub-Hilbert space of all possible $n$-qubit states.
This is called the \textit{code space}, $\mathcal H_C \subseteq \mathbb C^{2^n}$.
The sub-Hilbert space property is the quantum analogue of the linear subspace property for classical codes.

\begin{definition}[Quantum Code]\label{def:quantum}
    Let $n \ge 1$ be an integer.
    Then an $n$-qubit \textit{quantum code} is a sub-Hilbert space $\mathcal H_C \subseteq \mathbb C^{2^n}$.
    Elements of $\mathcal H_C$ are called code states.
\end{definition}

\begin{example}[Quantum Code]\label{ex:quantum}
    Let $n=5$, and consider the set $\mathcal H_C$ of all quantum states of the form
    \begin{equation}\label{eq:exquantumstates}
        \alpha_1 \ket{00000} + \alpha_2 \ket{10100} + \alpha_3 \ket{01011} + \alpha_4 \ket{11111},
    \end{equation}
    where each $\alpha_j$ is a complex number.
    This is a sub-Hilbert space of $\mathbb C^{2^5}$, and therefore $\mathcal H_C$ qualifies as a $5$-qubit quantum code.\footnote{The fact that normalized quantum states in the code space satisfy $\sum_j |\alpha_j|^2 = 1$ is independent from the definition of the code space itself.}
    In fact, this is the quantum code whose code states are superpositions of codewords from the classical code defined in Example~\ref{ex:classical}.
    Not all quantum codes have the simple interpretation of representing superpositions of classical codewords, but using an example with this property helps make the intuition behind the definitions more clear.
\end{example}

Having discussed the notion of quantum codes in general, we now turn to the special case of \textit{stabilizer codes}, which are central to the modern theory of QECCs.

\begin{definition}[Stabilizer Code]\label{def:stabilizer} 
Let $\mathcal P_n=\{\pm1,\pm i\}\cdot \{I,X,Y,Z\}^{\otimes n}$ be the $n$-qubit Pauli group.
    Suppose a subgroup~$\mathcal S$ of the Pauli group $\mathcal P_n$ satisfies the following properties:
    \begin{enumerate}[label=(\roman*)]
        \item $\mathcal S$ is Abelian;
        \item The negative identity $-I \in \mathcal P_n$ is \textit{not} an element of $\mathcal S$.
    \end{enumerate}
    Then quantum code $\mathcal H_C \subseteq \mathbb C^{2^n}$ defined by
    \begin{equation}
        \mathcal H_C = \{\ket{\psi} \in \mathbb C^{2^n} \mid S\ket{\psi} = \ket{\psi},\quad\forall S \in \mathcal S\}.
    \end{equation}
    is called the \textit{stabilizer code} with stabilizer $\mathcal S$.
    We also say that $\mathcal H_C$ is the \textit{code space} of $\mathcal S$.
\end{definition}

Even though the physical Hilbert space $\mathbb C^{2^n}$ has dimension which grows exponentially in $n$, the code space can be defined efficiently by using stabilizer codes.
In particular, Definition~\ref{def:stabilizer} reduces the problem to providing an Abelian subgroup $\mathcal S$ of $\mathcal P_n$.
Such a subgroup $\mathcal S$ can always be defined by a set of at most $n$ generators.
Each stabilizer generator can be viewed as cutting the Hilbert space dimension in half (see Lemma \ref{lem:numgenerators}), implying that these are high-dimensional linear constraints.

\begin{example}[Stabilizer Code]\label{ex:stabilizer}
    The quantum code $\mathcal H_C$ from example \ref{ex:quantum} is actually a stabilizer code.
    To understand this, consider the subgroup $\mathcal S \subseteq \mathcal P_5$ generated by the operators
    \begin{align}
        S_1 &= \texttt{Z I Z I I},\\
        S_2 &= \texttt{I Z I Z I},\\
        S_3 &= \texttt{I I I Z Z}.
    \end{align}
    This clearly defines an Abelian group that does not contain $-I$.
    For this example, we have selected $S_1$, $S_2$, and $S_3$ as $Z$-type operators that mimic the rows $h_1$, $h_2$, and $h_3$ of the parity-check matrix \eqref{eq:exparityH} from example \ref{ex:parity}:
    \begin{align}
        h_1 &= 10100,\\
        h_2 &= 01010,\\
        h_3 &= 00011.
    \end{align}
    By doing so, we are guaranteed that
    \begin{equation}
        S_i \ket{x} = (-1)^{h_i \cdot x} \ket{x},
    \end{equation}
    for any $5$-bit string $x \in \mathbb F_2^5$ in the computational basis.
    Then the conditions $S_i \ket{x} = \ket{x}$ are precisely equivalent to the linear constraints $h_i \cdot x = 0$, which are exactly the parity checks for the classical code $C$ from Example~\ref{ex:classical}.
    Since $\mathcal H_C$ consists of superpositions of the classical codewords, as noted in Example~\ref{ex:quantum}, this immediately implies by linearity that $\mathcal H_C$ is the code space for $\mathcal S$, making it a stabilizer code.
\end{example}

A highly useful concept in the setting of stabilizer codes is that of \textit{logical Pauli operators}, which does not have a direct classical analogue.
In fact, this is precisely the state-independent quantum notion of an ``undetectable" error, as clarified by Lemma~\ref{lem:logical-pauli-operator}.
The reader who is less familiar with quantum error correction may consider proving it as an exercise.

\begin{definition}[Logical Pauli Operators]\label{def:logical-pauli-operator}
    Let $\mathcal H_C$ be a stabilizer code with stabilizer $\mathcal S$.
    Then an $n$-qubit Pauli operator $E \in \mathcal P_n$ is called a \textit{logical Pauli operator} if the following conditions are satisfied:
    \begin{enumerate}[label=(\roman*)]
        \item $E$ commutes with every element of the stabilizer group $\mathcal S$.
        \item $E$ is \textit{not} contained in the stabilizer group $\mathcal S$.
        \item $E$ is \textit{not} a global phase operator, $I$, $-I$, $iI$, or $-iI$.
    \end{enumerate}
    As a shorthand, we write that $E$ is a logical Pauli operator if and only if
    \begin{equation}
        E \in \mathcal N(\mathcal S)\setminus\left(\mathcal{S} \mathcal{I}\right),
    \end{equation}
    where
    \begin{equation}
        \mathcal N(\mathcal S) \equiv \{P \in \mathcal P_n \mid P S P^\dagger \in \mathcal S,\quad \forall S \in \mathcal S\}
    \end{equation}
    is the \textit{normalizer} of $\mathcal S$ in $\mathcal P_n$, and
    \begin{equation}
        \mathcal I \equiv \{I, iI, -I, -iI\} \subseteq \mathcal P_n
    \end{equation}
    is the set of global phase operations; and
    \begin{equation}
        \mathcal{S} \mathcal{I}= \{i^a S \mid a \in \{0,1,2,3\}, S\in \mathcal{S}\}\,.
    \end{equation}
\end{definition}

\begin{lemma}\label{lem:logical-pauli-operator}
    Let $\mathcal H_C$ be a stabilizer code with stabilizer $\mathcal S$, and suppose $E \in \mathcal P_n$ is a Pauli operator.
    Then $E$ is \textit{undetectable} (meaning that $E \ket{\psi} \propto \ket{\psi}$ for some $\ket{\psi} \in \mathcal H_C$) if and only if $E$ is a logical Pauli operator.
    Moreover, this is equivalent to the condition that $E$ commutes with the entire stabilizer $\mathcal S$, while not being related by global phase to any element of $\mathcal S$.
\end{lemma}

\begin{example}[Logical Pauli Operators]\label{ex:logical-pauli-operator}
    Take $\mathcal H_C$ to be the quantum code from Example~\ref{ex:quantum}, whose stabilizer $\mathcal S$ was determined in Example~\ref{ex:stabilizer}.
    An example of a logical operator on $\mathcal H_C$ is given by $E = \texttt{X X X X X}$.
    A quick check shows that $E$ commutes with all stabilizer generators, $S_1$, $S_2$, and $S_3$.
    In fact, this operation has precisely the same interpretation as the undetectable error discussed in Example~\ref{ex:logical} for the classical code---both operations apply a bit-flip to all $5$ bits.
    To make mathematical sense of why the detectability of classical bit-flip errors is the same as that for $X$-type errors in the quantum setting, note the relation
    \begin{equation}
        S_i E = (-1)^{h_i \cdot e} ES_i,
    \end{equation}
    which holds for any $X$-type Pauli error $E$, provided that $e \in \mathbb F_2^5$ is chosen to be the corresponding classical bit-flip error.
    
    At least for bit-flip errors, this clarifies why commutation with the stabilizer group is related to the undetectability of the error.
    As seen from Lemma~\ref{lem:logical-pauli-operator}, it is in fact the general rule that a Pauli operator will be undetectable on \textit{every} code state if and only if it commutes with the entire stabilizer.
    Analogously to the classical case, this motivates the notion of \textit{stabilizer syndrome} for Pauli operators.
    If $E \in \mathcal P_n$ is any Pauli operator, then its syndrome is defined to be the answers to each boolean question of the form, ``Does $E$ commute with the stabilizer generator $S_i \in \mathcal S$?"
    As in the classical case, the syndrome clearly helps to determine the detectability of an error, but it can also be used to posit informed guesses on how an error should be corrected, by using syndrome decoders.
    We delegate the discussion of technicalities related to syndrome decoding to~\cref{app:syndrome-decoding}.
\end{example}

Stabilizer codes admit properties that are analogous to the properties of classical codes.
The most basic properties are the following:
\begin{itemize}
    \item The total number of qubits used by the code, also called the number of \textit{physical qubits} (denoted $n$).
    \item The number of \textit{logical qubits} of information that can be stored in the code (denoted $k$).
    \item The minimum number of qubits acted upon by any logical Pauli operator, also called the \textit{distance} of the code (denoted $d$).
\end{itemize}

The number of logical qubits $k$ can be thought of as the maximum number of qubits in a message state $\ket{\psi_m}$ that can be faithfully encoded by the code.
Mathematically, it is given by
\begin{equation}
    k = \log_2 (\dim \mathcal H_C).
\end{equation}
In practical terms, however, one needs to know how $k$ can be determined from the generators of a stabilizer group $\mathcal S$, since that will be the actual information available when dealing with a code.
For this, we refer to Lemma~\ref{lem:numgenerators}, which immediately implies that $k = n - r$, where $r$ is the number of stabilizer generators.

\begin{lemma}\label{lem:numgenerators}
    Let $\mathcal S$ be an $n$-qubit stabilizer group with code space $\mathcal H_C$.
    If $r\ge 1$ is an integer, and $\{S_1, S_2, \dots, S_r\} \subseteq \mathcal S$ is a complete set of independent generators for $\mathcal S$, then $\dim \mathcal H_C = 2^{n-r}$.
\end{lemma}

As in the classical case, the code distance $d$ is a measure of how difficult it is to produce an error that is invisible to the code---but this time, the errors can be Pauli operators.
In this setting, the \textit{weight} of an error $E$ is denoted $\wt(E)$, and it counts the number of physical qubits on which $E$ acts as non-identity.
This is the analogue of the classical weight function that counts the number of nonzero bits.
The distance of the stabilizer code is then the minimum weight of any logical Pauli operator.
In other words, this can mathematically be expressed as
\begin{equation}\label{eq:stabd}
    d = \min_{E \in \mathcal N(\mathcal S)\setminus(\mathcal S \mathcal I)} \wt(E).
\end{equation}

As in the classical case, it is common to summarize the properties of a stabilizer code by using the descriptor $\llbracket n,k,d\rrbracket$, where the double-bracket notation indicates that this is a stabilizer code, rather than a classical code.

\begin{example}[Properties of Stabilizer Code]\label{ex:propertiesquantum}
    Take $\mathcal H_C$ and $\mathcal S$ from Examples~\ref{ex:quantum} and~\ref{ex:stabilizer}, which has $n=5$ physical qubits. There are $r = 3$ stabilizer generators, so the number of logical qubits can also be computed as $k = n - r = 2$. 
    This agrees with the fact that $\dim \mathcal H_C = 4$, corresponding to two logical qubits.
     
    For the code distance, we need the minimum-weight Pauli error $E$ that is undetectable by the code.
    The logical operator $\texttt{X X X X X}$ from Example~\ref{ex:logical-pauli-operator} has weight $5$, but it is not the smallest weight error.
    As an example, consider the operator $E = \texttt{Z I I I I}$.
    Clearly, this operator commutes with the entire stabilizer $\mathcal S$, and it's not difficult to see that $E$ is not in $\mathcal S$ and is clearly not a global phase.
    Therefore, $E$ is a logical operator.
    Lemma \ref{lem:logical-pauli-operator} suggests that there must exist a code state $\ket{\psi}$ which is mapped to a physically different code state by $E$.
    Such an example is not hard to find; for instance,
    \begin{equation}
        E\left(\ket{00000} + \ket{10100}\right) = \ket{00000} - \ket{10100}.
    \end{equation}
    The weight of this operator is $\wt(E) = 1$, which clearly cannot be reduced without literally making it proportional to the identity element of the Pauli group (and thus not an error).
    Therefore, the code distance is $d = 1$.

    Putting it all together, the code $\mathcal H_C$ is a $\llbracket 5,2,1\rrbracket$ stabilizer code.
    Note that this descriptor does not exactly match the descriptor $[5,2,2]$ for the classical code $C$ that is imitated by $\mathcal H_C$.
    It matches on the number of physical bits and logical bits, but not on the code distance.
    
    The number of physical qubits matches the number of physical bits because each qubit of the quantum code is imitating a corresponding bit of the classical code.
    The number of logical qubits matches the number of logical bits because Lemma~\ref{lem:numgenerators} shows that each stabilizer generator reduces the Hilbert space dimension by a factor of two.
    But the code distance does not match, because the minimum-weight Pauli error was a $Z$-error, which has no classical analogue---the classical code was only detecting bit-flip errors, which are $X$-errors in this setting.
\end{example}

The failure of the quantum code distance to match the classical code distance in Example~\ref{ex:properties} highlights an important point: quantum codes need to take into account the possibility of bit-flip errors \textit{and} phase-flip errors.
With some ingenuity, this is achievable in the stabilizer framework, as discussed in any standard text on the subject, e.g.~\cite{nielsen00,preskill_chap7}.
For the purposes of this paper, we will deal with this point by introducing a construction known as the \textit{concatenation} of two codes.

\subsubsection{Concatenation of Codes}

There are many ways to construct stabilizer codes that can simultaneously handle bit-flip and phase-flip errors.
Any standard reference will discuss this point, which includes the construction of CSS codes, perfect codes, and concatenation~\cite{knill1996concatenated}.
For the purposes of this work, we will need the concatenation construction.
In fact, we will only need a very special case of concatenation, formulated so that it produces a quantum code from a pair of classical codes.

As alluded to by Example~\ref{ex:stabilizer}, it is possible to produce a stabilizer code---entirely from $Z$-type operators---that mimics a classical parity-check matrix.
But the problem with doing so, as discussed in Example~\ref{ex:properties}, is that the code distance drops to $1$, because although bit-flip errors can be detected exactly as in the classical code, the quantum code remains indifferent to phase-flip errors.
Conversely, if we were to construct a stabilizer code entirely from $X$-type operators, it would be able to detect phase-flip errors, but it would become oblivious to bit-flip errors.

Concatenation solves this problem by instead using two parity-check matrices, $H^{(1)}$ and $H^{(2)}$, for classical codes $[n_\OUT, k_\OUT, d_\OUT]$ and $[n_\IN, k_\IN, d_\IN]$, respectively.
Converting $H^{(1)}$ into $Z$-type stabilizer generators, and considering $n_\OUT$ blocks of $n_\IN$ qubits, it becomes possible to detect $X$-type errors that are incurred over multiple blocks of the code.
However, $Z$-type errors remain undetected. 
By converting $H^{(2)}$ into $X$-type stabilizer generators, considering $n_\OUT$ blocks of $n_\IN$ qubitsallows to detect $Z$-type errors.

To define the construction precisely, the simplest method is to write down the stabilizer generators.
We will focus on the case where $k_\IN = 1$, since that is all we will need in this paper.
Our optimal code constructions in~\cref{sec:gauss} are largely unaffected if other concatenation strategies are used.

We first introduce a precursor definition (mostly for notation):
\begin{definition}\label{def:SOh}
    Let $h$ be an $n$-bit string, and let $O$ be an $n'$-qubit operator.
    Then define $S^O(h)$ to be the $n\times n'$-qubit operator
    \begin{equation}
        S^O(h) \equiv \bigotimes_{i = 1}^n O(i),
    \end{equation}
    where $O(i)$ is the $n'$-qubit operator
    \begin{equation}
        O(i) \equiv
        \begin{cases}
            I^{\otimes n'}, & h_i = 0,\\
            O, & h_i = 1.
        \end{cases}
    \end{equation}
\end{definition}

\begin{definition}[Concatenation]\label{def:concatenation}
    Let $C_\OUT = [n_\OUT,k_\OUT,d_\OUT]$ and $C_\IN = [n_\IN,1,d_\IN]$ be two classical codes, with parity-check matrices $H^{(\OUT)}$ and $H^{(\IN)}$, respectively.
    Write $h^{(\OUT)}_i$ and $h^{(\IN)}_i$ for the $i^{\mathrm{th}}$ rows of $H^{(\OUT)}$ and $H^{(\IN)}$, respectively.
    For $1 \le i \le n_\OUT$ and $1 \le j \le n_\IN-1$ (remember that $k_\IN = 1$), define the $n_\OUT \times n_\IN$-qubit \textit{inner block} operators
    \begin{equation}
        S_{i,j} \equiv I^{\otimes (i-1)n_\IN} \otimes S^X\left(h^{(\IN)}_j\right) \otimes I^{\otimes (n_\OUT-i)n_\IN},
    \end{equation}
    where $X$ denotes the single-qubit Pauli-$X$ operator.
    Now let $\overline{Z}$ denote any $Z$-type Pauli operator that is a logical error for $C_\IN$.
    For $1 \le i \le n_\OUT-k_\OUT$, define the $n_\OUT \times n_\IN$-qubit \textit{outer block} operators
    \begin{equation}
        S_i \equiv S^{\overline{Z}}\left(h^{(\OUT)}_i\right).
    \end{equation}
    Then the \textit{concatenation} of $C_\OUT$ and $C_\IN$ is the stabilizer code generated by all inner block and outer block operators defined above.
    We say that $C_\OUT$ is the \textit{outer code}, and $C_\IN$ is the \textit{inner code} for this construction.
    The concatenation of $C_\OUT$ and $C_\IN$ is denoted $C_\OUT \circ C_\IN$.
\end{definition}

Note that this definition depends on the choice of parity-check matrices used to define the classical codes $C_{\IN}$ and $C_{\OUT}$, and we mentioned earlier that a given code does not necessarily have a unique parity check matrix. However, the concatenation is independent of the choice of parity check matrices:

\begin{lemma}[Independent of Parity-Check Matrices]\label{lem:concatenationindependence}
    Let $C_\OUT = [n_\OUT,k_\OUT,d_\OUT]$ and $C_\IN = [n_\IN,1,d_\IN]$ be two classical codes.
    Then the concatenation $C \equiv C_\OUT\circ C_\IN$ is independent of the parity-check matrices chosen for $C_\OUT$ and $C_\IN$.
\end{lemma}

Lemma~\ref{lem:concatenationindependence} serves as the final step to conclude our definition, ensuring that concatenated codes are well-defined even when parity-check matrices are not explicitly provided.
We can now proceed to Example~\ref{ex:concatenation}, which shows explicitly how to build the stabilizer generators for a concatenated code.

\begin{example}[Concatenation]\label{ex:concatenation}
    Take the outer code~$C_\OUT$ to be defined by the parity-check matrix
    \begin{equation}\label{eq:exconcatenationH1}
        H^{(\OUT)} =
        \begin{pmatrix}
            1 & 0 & 1 & 0 & 0\\
            0 & 1 & 0 & 1 & 0\\
            0 & 0 & 0 & 1 & 1
        \end{pmatrix},
    \end{equation}
    as used in previous examples.
    Clearly, $C_\OUT = [5,2,2]$ as we have discussed in previous examples.
    Take the inner code $C_\IN$ to be defined by the parity-check matrix
    \begin{equation}\label{eq:exconcatenationH2}
        H^{(\IN)} =
        \begin{pmatrix}
            1 & 1 & 0\\
            0 & 1 & 1
        \end{pmatrix}.
    \end{equation}
    It can be shown that $C_\IN = [3,1,3]$, which encodes a single logical bit, and therefore can be used in the concatenation construction from Definition~\ref{def:concatenation}.

    First, we construct operators $S^X\left(h^{(\IN)}_j\right)$ from the rows of \eqref{eq:exconcatenationH2} as
    \begin{align}\label{eq:exconcatenationSX}
        S^X\left(h^{(\IN)}_1\right) &= \texttt{X X I},\\
        S^X\left(h^{(\IN)}_2\right) &= \texttt{I X X},
    \end{align}
    which leads to the inner block generators
    \begin{align}
        S_{1,1} &= \texttt{X X I} \mid \texttt{I I I} \mid \texttt{I I I} \mid \texttt{I I I} \mid \texttt{I I I},\\
        S_{1,2} &= \texttt{I X X} \mid \texttt{I I I} \mid \texttt{I I I} \mid \texttt{I I I} \mid \texttt{I I I},\\
        S_{2,1} &= \texttt{I I I} \mid \texttt{X X I} \mid \texttt{I I I} \mid \texttt{I I I} \mid \texttt{I I I},\\
        S_{2,2} &= \texttt{I I I} \mid \texttt{I X X} \mid \texttt{I I I} \mid \texttt{I I I} \mid \texttt{I I I},\\
        S_{3,1} &= \texttt{I I I} \mid \texttt{I I I} \mid \texttt{X X I} \mid \texttt{I I I} \mid \texttt{I I I},\\
        S_{3,2} &= \texttt{I I I} \mid \texttt{I I I} \mid \texttt{I X X} \mid \texttt{I I I} \mid \texttt{I I I},\\
        S_{4,1} &= \texttt{I I I} \mid \texttt{I I I} \mid \texttt{I I I} \mid \texttt{X X I} \mid \texttt{I I I},\\
        S_{4,2} &= \texttt{I I I} \mid \texttt{I I I} \mid \texttt{I I I} \mid \texttt{I X X} \mid \texttt{I I I},\\
        S_{5,1} &= \texttt{I I I} \mid \texttt{I I I} \mid \texttt{I I I} \mid \texttt{I I I} \mid \texttt{X X I},\\
        S_{5,2} &= \texttt{I I I} \mid \texttt{I I I} \mid \texttt{I I I} \mid \texttt{I I I} \mid \texttt{I X X},
    \end{align}
    where the divider lines $\mid$ are shown to make it easier to distinguish the different blocks of qubits.

    Next, for the outer block generators, we need to choose a $Z$-type Pauli operator $\overline{Z}$ that serves as a logical error for the inner code $C_\IN$, where $C_\IN$ has been modified to work in the $X$-basis via the operators Eq.~\eqref{eq:exconcatenationSX}.
    An example of a logical $Z$-type error is provided by $\overline{Z} = \texttt{Z Z Z}$, which clearly commutes with both $S^X\left(h^{(\IN)}_1\right)$ and $S^X\left(h^{(\IN)}_2\right)$, but is not contained in the group generated by those operators.

    By using the notation from definition \ref{def:SOh}, the operators $S^{\overline{Z}}$ are given by
    \begin{align}
        S^{\overline{Z}}\left(h^{(\OUT)}_1\right) &= \overline{\texttt{Z}} \mid \overline{\texttt{I}} \mid \overline{\texttt{Z}} \mid \overline{\texttt{I}} \mid \overline{\texttt{I}},\\
        S^{\overline{Z}}\left(h^{(\OUT)}_2\right) &= \overline{\texttt{I}} \mid \overline{\texttt{Z}} \mid \overline{\texttt{I}} \mid \overline{\texttt{Z}} \mid \overline{\texttt{I}},\\
        S^{\overline{Z}}\left(h^{(\OUT)}_3\right) &= \overline{\texttt{I}} \mid \overline{\texttt{I}} \mid \overline{\texttt{I}} \mid \overline{\texttt{Z}} \mid \overline{\texttt{Z}},
    \end{align}
    where we define $\overline{\texttt{I}}$ as a shorthand for the $3$-qubit identity, $\texttt{I I I}$.
    Plugging in $\overline{Z}$, we finally obtain the outer block generators
    \begin{align}
        S_1 &= \texttt{Z Z Z} \mid \texttt{I I I} \mid \texttt{Z Z Z} \mid \texttt{I I I} \mid \texttt{I I I},\\
        S_2 &= \texttt{I I I} \mid \texttt{Z Z Z} \mid \texttt{I I I} \mid \texttt{Z Z Z} \mid \texttt{I I I},\\
        S_3 &= \texttt{I I I} \mid \texttt{I I I} \mid \texttt{I I I} \mid \texttt{Z Z Z} \mid \texttt{Z Z Z}.
    \end{align}

    The concatenated code is now fully defined by taking the inner block generators $S_{i,j}$ together with the outer block generators $S_i$.
    If any $2$-qubit $Z$-type error occurs---whether supported on a single block, or separated between multiple blocks---then it \textit{must} anti-commute with one of the $X$-type generators $S_{i,j}$, and therefore it will be detected by the concatenated code.
    Similarly, if any single-qubit $X$-type error occurs, it \textit{must} anti-commute with one of the $Z$-type generators $S_i$, and therefore it will also be detected by the concatenated code.
    It is not hard to see that these properties follow from the classical abilities of the inner code $C_\IN = [3,1,3]$ and outer code $C_\OUT = [5,2,2]$ to detect errors of weights $2$ and $1$, respectively.
    Extending this observation to the smallest \textit{undetectable} (i.e., logical) errors, we obtain lemma \ref{lem:concatenationdXdZ}.
\end{example}

\begin{lemma}\label{lem:concatenationdXdZ}
    Let $C_\OUT = [n_\OUT, k_\OUT, d_\OUT]$ and $C_\IN = [n_\IN, 1, d_\IN]$ be two classical codes.
    Write $d_X$ and $d_Z$ for the minimum weights of $X$-type and $Z$-type Pauli operators that are logical errors for the concatenated code $C_\OUT \circ C_\IN$.
    Then $d_X = d_\OUT$ and $d_Z = d_\IN$.
\end{lemma}

On the surface, Lemma~\ref{lem:concatenationdXdZ} might not seem immediately helpful to understand the actual distance of a concatenated code.
After all, distance refers to the smallest weight of \textit{any} logical Pauli error, which can even include mixed errors that are made from all three Pauli matrices, $X$, $Y$, and $Z$.
This is where Lemma~\ref{lem:concatenationproperties} comes into play.
In particular, it clarifies that knowledge of $d_X$ and $d_Z$ immediately implies knowledge of the full code distance.

\begin{lemma}[Properties of Concatenated Code]\label{lem:concatenationproperties}
    Let $C_\OUT = [n_\OUT, k_\OUT, d_\OUT]$ and $C_\IN = [n_\IN, 1, d_\IN]$ be two classical codes.
    Then the concatenation $C_\OUT \circ C_\IN$ is a stabilizer code with descriptor
    \begin{equation}
        C_\OUT \circ C_\IN = \llbracket n_\OUT \times n_\IN, k_\OUT, \min(d_\OUT,d_\IN)\rrbracket.
    \end{equation}
\end{lemma}

\begin{example}[Properties of Concatenated Code]\label{ex:concatenationproperties}
    Continue from Example~\ref{ex:concatenation}.
    It's immediately clear that the number of physical qubits in the concatenation is $n = n_\OUT \times n_\IN = 15$.

    To understand the number of logical qubits, note that we have $10$ inner block generators and $3$ outer block generators, for a total of $r = 13$ stabilizer generators.
    Then the number of logical qubits is given by $k = n - r = 2$, which indeed does match the number of logical bits $k_\OUT = 2$ of the classical code $C_\OUT$.

    Finally, regarding the code distance, Lemma~\ref{lem:concatenationproperties} suggests that the distance should be $2$.
    This is equivalent to saying that any single-qubit Pauli error should be detectable, and there is a $2$-qubit Pauli error that is logical.
    Note that this should be true when considering the \textit{entire} Pauli group---not just $X$-type or $Z$-type errors.
    Indeed, by glancing at the generators from Example~\ref{ex:concatenation}, it's easy to see that any single-qubit $X$, $Y$, or $Z$ error anticommutes with at least \textit{one} of the generators.
    Additionally, there are plenty of $2$-qubit Pauli errors that are invisible to the code; an example is given by
    \begin{equation}
        E = \texttt{X I I} \mid \texttt{I I I} \mid \texttt{I X I} \mid \texttt{I I I} \mid \texttt{I I I},
    \end{equation}
    which clearly commutes with all inner and outer block generators. 
    Therefore, the code distance is $d = 2$.

    Putting it all together, the descriptor of the concatenated code is
    \begin{equation}
        C_\OUT \circ C_\IN = \llbracket 15,2,2\rrbracket,
    \end{equation}
    as suggested by Lemma~\ref{lem:concatenationproperties}.
\end{example}

A useful benefit to keep in mind when thinking about concatenated codes is that, while there is a fixed code distance $d = \min(d_X, d_Z)$, it is possible in practice to detect many errors with weights higher than $d$.
For instance, if low-weight errors occur on every single block of the code, then the inner generators will still detect them.
The exact code distance becomes the dominant measure of performance only when we deal with very specific error channels (e.g., local stochastic Pauli noise) and consider the limit of very low per-qubit error rates.
This point will be revisited in~\cref{sec:compare}.

\subsection{$\mathbb Z_2$ Lattice Gauge Theory with and without Gauss's law constraints}\label{sec:LGT}

This work focuses on constructing error-correcting codes for $\mathbb Z_2$ lattice gauge theories (LGTs) coupled to matter.
While there are several inequivalent LGTs that can be defined with the gauge group $\mathbb Z_2$, many of them share exactly the same Hilbert space at the kinematical level, once a lattice geometry has been fixed.
Therefore, the construction of error-correcting codes needs essentially zero input from the dynamics of these theories.

That being said, the dynamics can indeed impact the relative difficulty when it comes to actually implementing a full quantum simulation protocol for a given LGT.
For instance, it is completely possible that a code $C_1$ serves better than another code $C_2$ when simulating Hamiltonian $H$; while at the same time, $C_2$ serves better than $C_1$ for simulating a different Hamiltonian $H'$.
The difference can even depend on the specific simulation strategy chosen for each Hamiltonian, or indeed the low-level hardware specifications.

In this work, we do not perform a detailed comparison of circuits for quantum simulation; instead, we consider the \textit{locality of the encoded Hamiltonian} as a simple proxy for the difficulty of simulation.
Comparisons using this proxy are performed in~\cref{sec:compare}, but they require us to introduce explicit Hamiltonians in this section.
Since the code constructions depend only on the kinematics and not the dynamics, we have organized this section based on the kinematical Hilbert space, and we discuss the prototypical Hamiltonian choice for each Hilbert space therein.
All notation introduced in this section is summarized in~\cref{tab:notation-z2-lgt}.

Before discussing the Hilbert spaces we need, it is worth outlining the general setup for all LGTs considered in this work.
Unless otherwise specified, we work on a periodic, hypercubic lattice
\begin{equation}\label{eq:def-cubic-lattice}
    \Lambda \equiv \left([0,N_1) \times [0,N_2) \times \cdots \times [0,N_D)\right) \cap \mathbb Z^D,
\end{equation}
in $D$ spatial dimensions.
Some LGTs may use staggered fermions for their matter content, which requires that the number of lattice sites $N_\mu$ along the $\mu^{\mathrm{th}}$ axis is an \textit{even integer} for $\mu \in \{1, 2, \dots, D\}$.
The set of lattice sites will be denoted $\mathcal S \cong \Lambda$, the set of links will be denoted $\mathcal L$, and the set of plaquettes will be denoted $\mathcal P$.
For any plaquette $P \in \mathcal P$, we write $\partial P$ to denote the set of links $L \in \mathcal L$ comprising its boundary.

There are two distinguished kinematical descriptions for the $\mathbb Z_2$ LGT that we will consider.
We will often refer to these as the ``standard" description and the ``gauge-fixed" description.
In the standard description, gauge fields are supported on links, matter is supported on lattice sites, and Gauss's law is enforced by hand (but not kinematically).
In the gauge-fixed description, the Hilbert space is reduced by solving Gauss's law at the kinematical level, thereby removing all matter content and leaving the gauge fields as the only dynamical degrees of freedom.
This kinematical reduction allows any Hamiltonian in the standard description to be uniquely projected down to a Hamiltonian in the gauge-fixed formulation, with identical dynamics.

\subsubsection{Standard Description}
\label{subsec:standard-form}

Lattice sites are labeled by $S_{\mathbf{n}}$, where $\mathbf{n} = (n_1, n_2, \dots, n_D) \in \Lambda$ is a lattice coordinate. 
Links are labeled by $L_{\mathbf{n},\mu}$, where $\mathbf{n} \in \Lambda$ and $\mu \in \{1,2,\dots, D\}$ denotes an available spatial direction.
By abuse of notation, we use $S_{\mathbf{n}}$ to simultaneously denote the site $S_{\mathbf{n}} \in \mathcal S$, as well as the corresponding site variable operator, and its $\mathbb Z_2$-eigenvalue (represented as a bit $S_{\mathbf{n}} \in \mathbb F_2$) in any given state within its own eigenbasis.
Similarly, we use~$L_{\mathbf{n},\mu}$ to simultaneously denote the link $L_{\mathbf{n},\mu} \in \mathcal L$, the corresponding link variable operator, and its $\mathbb Z_2$-eigenvalue (represented as a bit $L_{\mathbf{n},\mu} \in \mathbb F_2$) in any given state within its own eigenbasis.
Graphically, this labeling of lattice data is shown for the $2$D spatial lattice in \cref{fig:standard}.

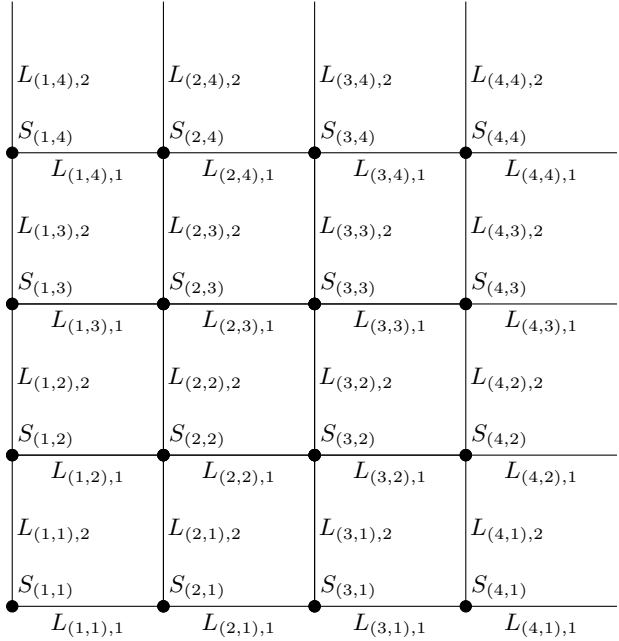
\begin{figure}[htp]
    \centering
    \begin{tikzpicture}
        \pgfmathsetmacro{\a}{2.0}
        \pgfmathsetmacro{\offsetLx}{-0.55}
        \pgfmathsetmacro{\offsetLy}{0.25}
        \pgfmathsetmacro{\offsetSx}{0.45}
        \pgfmathsetmacro{\offsetSy}{0.25}
        \foreach \i in {0,1,2} {
            \foreach \j in {0,1,2} {
                \DressedPlaquette{\i*\a}{\j*\a}{\a}{black}
            }
        }
        \foreach \i in {0,1,2,3} {
            \draw (\i*\a,3*\a) -- (\i*\a,4*\a);
            \draw (3*\a,\i*\a) -- (4*\a,\i*\a);
            \foreach \j in {0,1,2,3} {
                \pgfmathtruncatemacro{\indexhor}{\j+8*\i+1}
                \pgfmathtruncatemacro{\indexver}{\indexhor+4}
                \pgfmathtruncatemacro{\indexsite}{\j+4*\i+1}
                \pgfmathtruncatemacro{\nx}{\j+1}
                \pgfmathtruncatemacro{\ny}{\i+1}
                \node at (\a*\j+\a*0.5,\a*\i-\offsetLy) {$L_{(\nx,\ny),1}$};
                \node at (\a*\j-\offsetLx,\a*\i+\a*0.5) {$L_{(\nx,\ny),2}$};
                \node at (\a*\j+\offsetSx,\a*\i+\offsetSy) {$S_{(\nx,\ny)}$};
            }
        }
    \end{tikzpicture}
    \caption{Example $4 \times 4$ lattice in two spatial dimensions, with periodic boundary conditions. The unique set of sites and links are shown, with an example indexing. The Hilbert space $\mathcal H$ for this lattice is spanned by every possible assignment of bits to the site and link variables, disregarding Gauss's law constraints.}\label{fig:standard}
\end{figure}

The Hilbert space $\mathcal H$ used in the standard description is then spanned by the simultaneous eigenbasis of all $S_{\mathbf{n}}$ and~$L_{\mathbf{n},\mu}$ operators.
In other words, $\mathcal H$ consists of a single qubit on every lattice site and link, representing the $\mathbb Z_2$-charge or $\mathbb Z_2$-flux on that site or link, respectively.

Gauge invariant states in this description can be defined via the Gauss law operators 
\begin{equation}\label{eq:gauss-law-operator}
    G_{\mathbf{n}} \equiv Z_{S_{\mathbf{n}}} \times \prod_{L\in\partial S^*_{\mathbf{n}}} Z_L,\quad \mathbf{n}\in\Lambda,\footnote{A popular convention when dealing with fermionic matter involves the staggered background charge $\varepsilon_{\mathbf{n}} \in \{-1,1\}$ in the definition of $G_{\mathbf{n}}$; in that case, the $\ket{0}$ and $\ket{1}$ states typically refer to fermion occupancies, rather than excitation level above the Dirac sea. We will choose our computational basis so that no staggered phases appear, irrespective of whether the matter is bosonic or fermionic.}
\end{equation}
where we use the notation
\begin{equation}\label{eq:productnotation}
    \prod_{L\in\partial S_{\mathbf{n}}^*} Z_L \equiv \left(\prod_{\mu=1}^D Z_{L_{\mathbf{n},\mu}}\right) \times \left(\prod_{\mu=1}^D Z_{L_{\mathbf{n}-\hat{\mu},\mu}}\right),
\end{equation}
to denote the product over all links $L$ that are incident to the lattice site $\mathbf{n}$, and $\hat{\mu}$ denotes the unit lattice vector along the $\mu$-direction.
When we subscript a Pauli operator with a lattice variable, we mean that the Pauli operator acts solely on the qubit corresponding to that lattice variable.

The gauge-invariant subspace $\mathcal H_G \subseteq \mathcal H$ is then the simultaneous $+1$-eigenspace of all Gauss's law operators~$G_{\mathbf{n}}$;
\begin{equation}\label{eq:def-gauge-inv-hilbert-space}
    \mathcal{H}_G := 
    \{ \ket{\psi} \in \mathcal{H} \mid G_{\mathbf{n}}\ket{\psi} = \ket{\psi},\, \forall\mathbf{n}\in\Lambda\}\,.
\end{equation}
Indeed, states in this subspace remain invariant under arbitrary gauge transformations.

There are several inequivalent $\mathbb Z_2$ LGTs that can be defined in the standard description, thereby carrying the same kinematical Hilbert space $\mathcal H$, the same Gauss's law operators $G_{\mathbf{n}}$, and the same resulting gauge-invariant subspace $\mathcal H_G$.
These include theories with hard-core bosonic matter \cite{HornYankielowicz1979GaugeMatterScreening, Horn1980PhaseStructureZ2Matter, HornKatznelson1980VariationalZ2, BanksSinclair1981Z2Matter}, and truncations of the $\mathrm{U}(1)$ gauge field coupled to staggered fermions \cite{Kogut:1975,Susskind:1977}.
A prototypical Hamiltonian for $\mathbb Z_2$ LGT coupled to
hard-core bosonic matter \cite{HornYankielowicz1979GaugeMatterScreening} takes the electric basis form
\begin{align}\label{eq:Hstandard}
    H = &\,-\frac{g^2}{4} \sum_{L \in \mathcal L} Z_{L} - \frac{m}{2} \sum_{S \in \mathcal S} Z_{S}-\lambda \sum_{P \in \mathcal P} \prod_{L \in \partial P} X_{L}\nonumber\\
    &\,-\frac{\kappa}{2}\sum_{\mathbf{n}\in\Lambda}\sum_{\mu=1}^D X_{S_{\mathbf{n}}} X_{L_{\mathbf{n},\mu}} X_{S_{\mathbf{n}+\hat{\mu}}}\,,
\end{align}
where the plaquette term with coupling constant $\lambda$ should be dropped in the case of spatial dimension $D = 1$.
If one wishes to couple to staggered fermions instead, an alternative $\mathbb Z_2$ LGT can be defined by a Hamiltonian of the form
\begin{align}\label{eq:Hfermion}
    H_{\rm fer} = &\,-\frac{g^2}{4} \sum_{L \in \mathcal L} Z_{L} +m \sum_{\mathbf{n}\in\Lambda} \epsilon_{S_{\mathbf{n}}} \chi^\dag_{S_{\mathbf{n}}} \chi_{S_{\mathbf{n}}} -\lambda \sum_{P \in \mathcal P} \prod_{L \in \partial P} X_{L}\nonumber\\
    &\, +i\frac{\kappa}{2}
    \sum_{\mathbf{n}\in\Lambda}\sum_{\mu=1}^D 
    \left( \chi^\dag_{S_{\mathbf{n}}}  \eta^\mu_{S_{\mathbf{n}}} X_{L_{\mathbf{n},\mu}} \chi_{S_{\mathbf{n}+\hat{\mu}}}
    -{\rm h.c.} \right)\,,
\end{align}
where $\chi_S $ and $\chi_S^\dag$ are fermion operators at site $S$ satisfying the standard canonical anti-commutation relations, and $\epsilon_{S_{\mathbf{n}}}$ and $\eta^\mu_{S_{\mathbf{n}}}$ are site-dependent signatures defined as
\begin{align}
\epsilon_{S_{\mathbf{n}}} = (-1)^{\sum_{\nu=1}^D n_\nu } ,\quad
\eta^\mu_{S_{\mathbf{n}}} = (-1)^{\sum_{\nu=1}^{\mu-1} n_\nu } .
\end{align}
It is known that the fermionic Fock space can be mapped to qubits, where the fermion operators are expressed as spin operators via appropriate transformations.
The most well-known such transformation is the Jordan-Wigner transformation \cite{Jordan:1928wi} which maps the Hamiltonian \eqref{eq:Hfermion} into a local spin Hamiltonian for $D=1$ and a non-local spin Hamiltonian for $D\geq 2$.
For the $1$D application, see for instance~\cite{Mildenberger_2025}, which can be adapted to our convention as the Hamiltonian
\begin{align}\label{eq:HfermionJW}
    H_{\rm fer} = &\,-\frac{g^2}{4} \sum_{n=1}^N Z_{L_n} - \frac{m}{2} \sum_{n=1}^N (-1)^n Z_{S_n}\nonumber\\ 
     &\,-\frac{\kappa}{4}
    \sum_{n=1}^N
    \left( X_{S_n} X_{L_n} X_{S_{n+1}} + Y_{S_n} X_{L_n} Y_{S_{n+1}} \right)\,,
\end{align}
where we used the fact that $\mathbf{n} \to n$ is a single integer in $1$D.
The Hamiltonians \eqref{eq:Hstandard} and \eqref{eq:HfermionJW} commute with all Gauss's law operators $G_{\mathbf{n}}$, implying that Gauss's law is satisfied during the dynamics, even though there are states in $\mathcal H$ that are not gauge invariant, and therefore do not satisfy Gauss's law.

For simplicity, in $2$D and higher dimensions, this paper exclusively uses \eqref{eq:Hstandard} as our $\mathbb Z_2$ LGT formulation.\footnote{In $1$D, we may occasionally use \eqref{eq:HfermionJW} for illustrative purposes.}
Since our constructions are entirely kinematical and therefore apply automatically to all $\mathbb Z_2$ LGTs on the same lattice, they are identical for all equivalent lattice Hilbert spaces, provided that the Gauss's law conventions are chosen consistently with \eqref{eq:gauss-law-operator}.

\subsubsection{Gauge-Fixed Description}

An alternative kinematical description is obtained by solving the Gauss's law constraints explicitly.
This leads to a Hilbert space spanned entirely by the link variables $L_{\mathbf{n},\mu} \in \mathcal L$.
In this formulation, the kinematical Hilbert space is isomorphic to the gauge-invariant Hilbert space $\mathcal H_G$ (see \cref{eq:def-gauge-inv-hilbert-space}), and simply consists of one qubit per link, representing the electric flux.

Any Hamiltonian in the standard description can be converted to the gauge-fixed description by demanding that each Gauss's law operator $G_{\mathbf{n}}$ acts as identity on the gauge-fixed Hilbert space, and replacing all site operators appropriately.
For example, for the hard-core boson theory~\eqref{eq:Hstandard}, we can substitute
\begin{align}
    X_S &\to I_S,\\
    Z_S &\to \prod_{L \in \partial S^*} Z_L,
\end{align}
into \eqref{eq:Hstandard} to obtain the gauge-fixed Hamiltonian
\begin{align}\label{eq:Hgaugefixed}
    H_{\text{gf}} = &\, -\frac{g^2}{4} \sum_{L \in \mathcal L} Z_L - \frac{m}{2} \sum_{\mathbf{n}\in\Lambda} \prod_{L\in\partial S^*} Z_L\nonumber\\
    &\,-\lambda \sum_{P \in \mathcal P} \prod_{L \in \partial P} X_{L} -\frac{\kappa}{2} \sum_{L \in \mathcal L} X_L\,,
\end{align}
where we reuse the product notation given as~\cref{eq:productnotation}.
As before, if we are dealing with the case of one spatial dimension, then the plaquette term is discarded.

While the gauge-fixed formulation \eqref{eq:Hgaugefixed} requires fewer qubits than the standard formulation, the locality of the Hamiltonian can be substantially higher.
For this reason, the Hamiltonian of the standard formulation \eqref{eq:Hstandard} may be a more feasible target for implementation on existing hardware, particularly in circumstances where the ability to perform multi-qubit gates is more of a bottleneck than the total number of available qubits.
\begin{table}
\centering
\caption{List of notations for cubic lattice and $\mathbb{Z}_2$ LGT in general dimension}
\label{tab:notation-z2-lgt}
\begin{tabular}{c|c}
Symbol & Name
\\
\hline
    $\Lambda$ (\cref{eq:def-cubic-lattice}) & lattice
    \\
    $D$ & space dimension
    \\
    $\mathbf{n}\in \Lambda$ & coordinates
    \\
    \hline
    $S_\mathbf{n} \in \mathcal{S}$ & sites
    \\
    $L_\mathbf{n,\mu} \in \mathcal{L}$ & links
    \\
    $P \in \mathcal{P}$ & plaquettes
    \\
    \hline
    $G_\mathbf{n}$ (\cref{eq:gauss-law-operator}) & Gauss's law
    \\
    $H$ (\cref{eq:Hstandard}) & hard-core boson Hamiltonian
    \\
    $H_{\text{fer}}$ (\cref{eq:Hfermion}) & staggered fermion Hamiltonian
    \\
    $H_{\text{gf}}$ (\cref{eq:Hgaugefixed}) & gauge-fixed bosonic Hamiltonian
    \\
    \hline
    $\mathcal{H}$ & full Hilbert space
    \\
    $\mathcal{H}_G$ (\cref{eq:def-gauge-inv-hilbert-space}) & gauge-invariant Hilbert space
\end{tabular}
\end{table}

\section{Arbitrary-Distance Gauss's Law Codes}\label{sec:gauss}

\subsection{Motivation: The RRW Code}\label{sec:RRW}

In our standard formulation of the $\mathbb Z_2$ LGT defined in \cref{sec:LGT}, a violation of Gauss's law manifests itself as a bit-flip error.
For instance, consider the $1$D lattice, where sites $S_n$ can be indexed by a single integer $1 \le n \le N$, and links $L_n$ no longer need to specify a direction $\mu$ in space.
A piece of the $1$D lattice is shown in \cref{fig:gauss-law-1D}.

The Gauss's law operator at site $S_n$ takes the form
\begin{equation}
    G_{n} \equiv Z_{L_{n-1}} Z_{S_{n}} Z_{L_{n}},
\end{equation}
which is diagonalized in the computational basis, and acts as identity on any gauge-invariant state.
Therefore, at the level of dynamics constrained by gauge invariance, this permits the locally-defined states
\begin{align}
    \ket{0}_{L_{n-1}}\ket{0}_{S_{n}}\ket{0}_{L_{n}},\\
    \ket{1}_{L_{n-1}}\ket{1}_{S_{n}}\ket{0}_{L_{n}},\\
    \ket{0}_{L_{n-1}}\ket{1}_{S_{n}}\ket{1}_{L_{n}},\\
    \ket{1}_{L_{n-1}}\ket{0}_{S_{n}}\ket{1}_{L_{n}},
\end{align}
but forbids the locally-defined states
\begin{align}
    \ket{1}_{L_{n-1}}\ket{0}_{S_{n}}\ket{0}_{L_{n}},\\
    \ket{0}_{L_{n-1}}\ket{1}_{S_{n}}\ket{0}_{L_{n}},\\
    \ket{0}_{L_{n-1}}\ket{0}_{S_{n}}\ket{1}_{L_{n}},\\
    \ket{1}_{L_{n-1}}\ket{1}_{S_{n}}\ket{1}_{L_{n}},
\end{align}
where every state in the latter group can be obtained via single bit-flip from some state(s) in the former group.

\begin{figure}[htp]
    \centering
    \begin{tikzpicture}
        \pgfmathsetmacro{\pr}{0.075}
        \pgfmathsetmacro{\offset}{0.3}
        \pgfmathsetmacro{\a}{2.0}
        \SmallCircle{0}{0}{1}{\pr}
        \draw (-\a,0) -- (\a,0);
        \node at (-0.5*\a,-\offset) {$L_{n-1}$};
        \node at (0,\offset) {$S_{n}$};
        \node at (0.5*\a,-\offset) {$L_{n}$};
    \end{tikzpicture}
    \caption{Illustration of variables participating in a single Gauss's law check for a lattice in $1$ spatial dimension.}\label{fig:gauss-law-1D}
\end{figure}
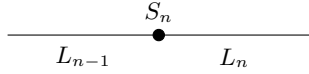

This is a \textit{classical} restriction, based solely on the values of the individual bits, with no regard to the phase information.
In other words, it merely restricts which computational basis states can be reached during gauge-invariant dynamics, without making reference to any superpositions of states.
To clarify this point, convert any computational basis state into a ``lattice bit string" by writing the eigenvalues of all link and site variables in order from left to right:
\begin{equation}
    b \equiv S_1 L_1 S_2 L_2 \cdots L_{n-1} S_n L_n \cdots S_N L_N \in \mathbb F_2^{2N}.
\end{equation}
The computational basis state $\ket{b}$ is always an eigenstate of the Gauss's law operator $G_n$, and it satisfies the Gauss's law constraint if and only if
\begin{equation}\label{eq:Gb}
    G_n \ket{b} = \ket{b} \iff L_{n-1} + S_n + L_n = 0,
\end{equation}
where addition is performed in $\mathbb F_2$.

Since the Gauss law operator commutes with the Hamiltonian, any computational basis state that contributes during exact real-time evolution governed by the Hamiltonian \eqref{eq:Hstandard} must necessarily satisfy the $N$ linear constraints given by \eqref{eq:Gb}.
As long as the initial state is gauge invariant, the dynamics is constrained guaranteed to only explore the gauge-invariant ``code" space.
For a periodic lattice containing $N$ sites in $1$D, this code can either be viewed as a classical $[2N, N, 3]$ code, or a $\llbracket 2N, N, 1\rrbracket$ quantum code.

To understand this more clearly, note that a periodic $1$D lattice with $N$ sites contains $2N$ degrees of freedom (DOFs), considering both sites and links.
Depending on whether we think in classical or quantum terms, these DOFs are either bits or qubits.
Gauss's law then provides a single linear constraint per vertex (for $N$ total constraints), which are parity checks at the classical level, or stabilizer generators at the quantum level.
This means there are exactly $2N - N = N$ logical DOFs left over when all constraints are applied.
Of course, we know exactly what those $N$ logical DOFs physically represent: they are simply the $N$ link variables, whose configuration uniquely determines the fermion content at all lattice sites.

The classical code distance is the minimum number of bit-flips needed to modify one codeword into a different codeword.
This is $3$, because in addition to flipping a prescribed link $L$, one must also flip the lattice sites at both endpoints of $L$ to preserve Gauss's law.
As discussed in~\cref{sec:quantum_code_stabilizer}, viewed as a quantum code the code distance is only 1, since only bit-flip ($X$) errors can be detected, while phase-flip ($Z$) errors can not. 

By concatenating this classical bit-flip code with a separate $[3,1,3]$ code used to handle phase-flips, we obtain a full
\begin{equation}
    [2N, N, 3] \circ [3,1,3] = \llbracket 6N,N,3 \rrbracket
\end{equation}
stabilizer code---this is precisely the \emph{RRW code}, originally introduced in~\cite{Rajput:2021trn}.

To study what errors can and can not be detected, we consider bit-flip and phase-flip errors that can occur on the lattice. 
In order for an error to be undetectable, it \textit{must} either:
\begin{itemize}
    \item Flip all $3$ phases within an inner block; or
    \item Flip a bit on a sequence of consecutive blocks that constitutes a flux string (pair of matter excitations and link excitation between them) on the lattice.
\end{itemize}
Both of these possibilities are illustrated in \cref{fig:errors-RRW-1D}, and it is easy to see that the corresponding Pauli errors will commute with the stabilizer generators of the concatenated code.
For additional intuition, \cref{fig:errors-RRW-1D} also shows an example of a higher-weight error that can not only be detected, but even corrected based on typical QEC protocols from the literature~(e.g.~\cite{Gottesman:1997zz, Knill:1996ny}).
This point will be discussed in more detail in~\cref{sec:compare}.

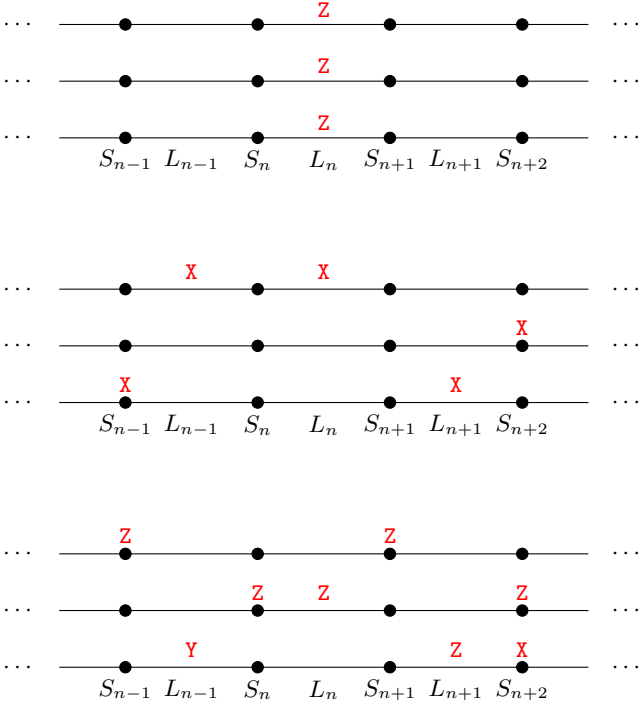
\begin{figure}[htp]
    \begin{minipage}{\linewidth}
        \begin{tikzpicture}
            \pgfmathsetmacro{\pr}{0.075}
            \pgfmathsetmacro{\offset}{0.2}
            \pgfmathsetmacro{\offsetmath}{0.3}
            \pgfmathsetmacro{\dx}{1.75}
            \pgfmathsetmacro{\dy}{0.75}
            \pgfmathsetmacro{\c}{0.1}
            \foreach \j in {0,1,2} {
                \foreach \i in {0,1,2,3} {
                    \SmallCircle{\i*\dx}{\j*\dy}{1}{\pr}
                }
                \draw (-0.5*\dx,\j*\dy) -- (3.5*\dx,\j*\dy);
                \node at (-0.8*\dx, \j*\dy) {$\cdots$};
                \node at (3.8*\dx,\j*\dy) {$\cdots$};
            }
            \node at (1.5*\dx,0*\dy+\offset) {\red{\texttt{Z}}};
            \node at (1.5*\dx,1*\dy+\offset) {\red{\texttt{Z}}};
            \node at (1.5*\dx,2*\dy+\offset) {\red{\texttt{Z}}};

            \node at (0.0*\dx,0*\dy-\offsetmath) {$S_{n-1}$};
            \node at (1.0*\dx,0*\dy-\offsetmath) {$S_n$};
            \node at (2.0*\dx,0*\dy-\offsetmath) {$S_{n+1}$};
            \node at (3.0*\dx,0*\dy-\offsetmath) {$S_{n+2}$};
            \node at (0.5*\dx,0*\dy-\offsetmath) {$L_{n-1}$};
            \node at (1.5*\dx,0*\dy-\offsetmath) {$L_n$};
            \node at (2.5*\dx,0*\dy-\offsetmath) {$L_{n+1}$};
        \end{tikzpicture}
    \end{minipage}

    \vspace{3em}

    \begin{minipage}{\linewidth}
        \begin{tikzpicture}
            \pgfmathsetmacro{\pr}{0.075}
            \pgfmathsetmacro{\offset}{0.2}
            \pgfmathsetmacro{\offsetmath}{0.3}
            \pgfmathsetmacro{\dx}{1.75}
            \pgfmathsetmacro{\dy}{0.75}
            \pgfmathsetmacro{\c}{0.1}
            \foreach \j in {0,1,2} {
                \foreach \i in {0,1,2,3} {
                    \SmallCircle{\i*\dx}{\j*\dy}{1}{\pr}
                }
                \draw (-0.5*\dx,\j*\dy) -- (3.5*\dx,\j*\dy);
                \node at (-0.8*\dx, \j*\dy) {$\cdots$};
                \node at (3.8*\dx,\j*\dy) {$\cdots$};
            }
            \node at (0.0*\dx,0*\dy+0.5*\pr+\offset) {\red{\texttt{X}}};
            \node at (0.5*\dx,2*\dy+0.5*\pr+\offset) {\red{\texttt{X}}};
            \node at (1.5*\dx,2*\dy+0.5*\pr+\offset) {\red{\texttt{X}}};
            \node at (2.5*\dx,0*\dy+0.5*\pr+\offset) {\red{\texttt{X}}};
            \node at (3.0*\dx,1*\dy+0.5*\pr+\offset) {\red{\texttt{X}}};

            \node at (0.0*\dx,0*\dy-\offsetmath) {$S_{n-1}$};
            \node at (1.0*\dx,0*\dy-\offsetmath) {$S_n$};
            \node at (2.0*\dx,0*\dy-\offsetmath) {$S_{n+1}$};
            \node at (3.0*\dx,0*\dy-\offsetmath) {$S_{n+2}$};
            \node at (0.5*\dx,0*\dy-\offsetmath) {$L_{n-1}$};
            \node at (1.5*\dx,0*\dy-\offsetmath) {$L_n$};
            \node at (2.5*\dx,0*\dy-\offsetmath) {$L_{n+1}$};
        \end{tikzpicture}
    \end{minipage}
    
    \vspace{3em}

    \begin{minipage}{\linewidth}
        \begin{tikzpicture}
            \pgfmathsetmacro{\pr}{0.075}
            \pgfmathsetmacro{\offset}{0.2}
            \pgfmathsetmacro{\offsetmath}{0.3}
            \pgfmathsetmacro{\dx}{1.75}
            \pgfmathsetmacro{\dy}{0.75}
            \pgfmathsetmacro{\c}{0.1}
            \foreach \j in {0,1,2} {
                \foreach \i in {0,1,2,3} {
                    \SmallCircle{\i*\dx}{\j*\dy}{1}{\pr}
                }
                \draw (-0.5*\dx,\j*\dy) -- (3.5*\dx,\j*\dy);
                \node at (-0.8*\dx, \j*\dy) {$\cdots$};
                \node at (3.8*\dx,\j*\dy) {$\cdots$};
            }
            \node at (0.0*\dx,2*\dy+0.5*\pr+\offset) {\red{\texttt{Z}}};
            \node at (0.5*\dx,0*\dy+0.5*\pr+\offset) {\red{\texttt{Y}}};
            \node at (1.0*\dx,1*\dy+0.5*\pr+\offset) {\red{\texttt{Z}}};
            \node at (1.5*\dx,1*\dy+0.5*\pr+\offset) {\red{\texttt{Z}}};
            \node at (2.0*\dx,2*\dy+0.5*\pr+\offset) {\red{\texttt{Z}}};
            \node at (2.5*\dx,0*\dy+0.5*\pr+\offset) {\red{\texttt{Z}}};
            \node at (3.0*\dx,1*\dy+0.5*\pr+\offset) {\red{\texttt{Z}}};
            \node at (3.0*\dx,0*\dy+0.5*\pr+\offset) {\red{\texttt{X}}};

            \node at (0.0*\dx,0*\dy-\offsetmath) {$S_{n-1}$};
            \node at (1.0*\dx,0*\dy-\offsetmath) {$S_n$};
            \node at (2.0*\dx,0*\dy-\offsetmath) {$S_{n+1}$};
            \node at (3.0*\dx,0*\dy-\offsetmath) {$S_{n+2}$};
            \node at (0.5*\dx,0*\dy-\offsetmath) {$L_{n-1}$};
            \node at (1.5*\dx,0*\dy-\offsetmath) {$L_n$};
            \node at (2.5*\dx,0*\dy-\offsetmath) {$L_{n+1}$};
        \end{tikzpicture}
    \end{minipage}
    \caption{Different types of Pauli errors in the RRW code, shown on top of the physical qubits representing a portion of the $1$D lattice. Each link and site is represented with $3$ physical qubits functioning as a $[3,1,3]$ phase-flip code. Top: logical (undetectable) $Z$-type error on link $L_n$. Middle: logical (undetectable) $X$-type error that creates a flux string originating at site $S_{n-1}$, spanning $3$ links $L_{n-1}$, $L_n$, and $L_{n+1}$, and terminating at site $S_{n+2}$. Bottom: higher-weight \textit{detectable and correctable} error using all Pauli matrices, $X$, $Y$, and $Z$.}\label{fig:errors-RRW-1D}
\end{figure}

Note that for errors to be undetectable, several Pauli errors have to happen in close proximity. 
Errors with sufficiently low density of affected qubits are easily handled by this code.
This is in fact a feature of many quantum error correcting codes, which are built out of local blocks of physical qubits comprising a set of logical qubits.
For errors to be undetectable, several errors have to occur in the same block. 
One difference of the Gauss' law codes discussed here is that there are no separated blocks that can be defined, since Gauss' law always acts on neighbors on both ``sides''. 
This will be discussed more later.

As this construction will be extended to higher distances, it will be interesting to investigate how this locality continues to hold.

\subsection{Extension to Arbitrary Distance}
\label{subsec:gauss-law-code-arbitrary-distance}
Now we come to our extension of the RRW code to higher distances.
The construction proceeds in two steps.
First, we construct the classical Gauss's law code, capable of reaching arbitrary distance, that handles bit-flip errors in an analogous way to the original RRW construction.
Finally, to build a fully capable quantum code, we concatenate the classical Gauss's law code with an arbitrarily chosen secondary code that can handle phase-flip errors.

The majority of the effort goes into extending the classical Gauss's law code.
For this, we have to understand more deeply why the RRW code defined in \cref{sec:RRW} has distance $3$.
The standard formulation of the $\mathbb Z_2$ LGT makes reference to site variables $S_{\mathbf{n}}$ and link variables $L_{\mathbf{n},\mu}$.
At the classical level, a configuration of the lattice assigns a single bit to each variable, and each Gauss's law constraint takes the form of the linear constraint
\begin{equation}\label{eq:classicalgauss}
    S_{\mathbf{n}} + \sum_{\mu=1}^D\left(L_{\mathbf{n},\mu} + L_{\mathbf{n}-\hat{\mu},\mu}\right) = 0,\quad \mathbf{n}\in\Lambda.
\end{equation}
The set of configurations on the lattice that satisfy these constraints are geometrically represented by disjoint collections of open strings and closed strings.
An open string consists of a pair of matter excitations joined by a continuous, non-intersecting path of excited electric flux along the links of the lattice. Closed strings are similar, but they close into a loop, and therefore do not possess matter excitations.
These two possibilities are shown in \cref{fig:open-closed-string}.

\begin{figure}[htp]
    \centering
    \begin{tikzpicture}
        \pgfmathsetmacro{\a}{2.0}
        \pgfmathsetmacro{\offsetx}{0.25}
        \pgfmathsetmacro{\offsety}{0.25}
        \foreach \i in {0,1,2,3} {
            \foreach \j in {0,1,2,3} {
                \DressedPlaquette{\i*\a}{\j*\a}{\a}{black}
            }
        }

        \pgfmathsetmacro{\pr}{0.08}
        \SmallCircle[color=blue]{0*\a}{4*\a}{1}{\pr}
        \draw[color=blue, line width=2pt] (0*\a,4*\a) -- (0.96*\a, 4*\a);
        \draw[color=blue, line width=2pt] (1*\a,3.96*\a) -- (1*\a, 3.04*\a);
        \draw[color=blue, line width=2pt] (1.04*\a,3*\a) -- (2*\a, 3*\a);
        \SmallCircle[color=blue]{2*\a}{3*\a}{1}{\pr}

        \draw[color=green!70!black, line width=2pt] (2.04*\a,2*\a) -- (2.96*\a,2*\a);
        \draw[color=green!70!black, line width=2pt] (3*\a,1.96*\a) -- (3*\a,1.04*\a);
        \draw[color=green!70!black, line width=2pt] (3.04*\a,1*\a) -- (3.96*\a,1*\a);
        \draw[color=green!70!black, line width=2pt] (4*\a,0.96*\a) -- (4*\a,0.04*\a);
        \draw[color=green!70!black, line width=2pt] (3.96*\a,0*\a) -- (3.04*\a,0*\a);
        \draw[color=green!70!black, line width=2pt] (2.96*\a,0*\a) -- (2.04*\a,0*\a);
        \draw[color=green!70!black, line width=2pt] (2*\a,0.04*\a) -- (2*\a,0.96*\a);
        \draw[color=green!70!black, line width=2pt] (2*\a,1.04*\a) -- (2*\a,1.96*\a);

        \node at (0*\a-\offsetx,4*\a+\offsety) {$1$};
        \node at (0.5*\a,4*\a+\offsety) {$1$};
        \node at (1*\a+\offsetx,4*\a+\offsety) {$0$};
        \node at (1*\a+\offsetx,3.5*\a) {$1$};
        \node at (1*\a-\offsetx,3*\a-\offsety) {$0$};
        \node at (1.5*\a,3*\a-\offsety) {$1$};
        \node at (2*\a+\offsetx,3*\a-\offsety) {$1$};

        \node at (2*\a-\offsetx,2*\a+\offsety) {$0$};
        \node at (2.5*\a,2*\a+\offsety) {$1$};
        \node at (3*\a+\offsetx,2*\a+\offsety) {$0$};
        \node at (3*\a+\offsetx,1.5*\a) {$1$};
        \node at (3*\a-\offsetx,1*\a-\offsety) {$0$};
        \node at (3.5*\a,1*\a-\offsety) {$1$};
        \node at (4*\a+\offsetx,1*\a+\offsety) {$0$};
        \node at (4*\a+\offsetx,0.5*\a) {$1$};
        \node at (4*\a+\offsetx,0*\a-\offsety) {$0$};
        \node at (3.5*\a,0*\a-\offsety) {$1$};
        \node at (3*\a,0*\a-\offsety) {$0$};
        \node at (2.5*\a,0*\a-\offsety) {$1$};
        \node at (2*\a-\offsetx,0*\a-\offsety) {$0$};
        \node at (2*\a-\offsetx,0.5*\a) {$1$};
        \node at (2*\a-\offsetx,1.1*\a) {$0$};
        \node at (2*\a-\offsetx,1.5*\a) {$1$};
    \end{tikzpicture}
    \caption{Examples of flux strings that correspond to classical lattice configurations that satisfy Gauss's law. The bits for each link and site along the string paths are shown. Blue: open string, contains matter excitations at the endpoints. Green: closed string, contains no matter excitations.}\label{fig:open-closed-string}
\end{figure}
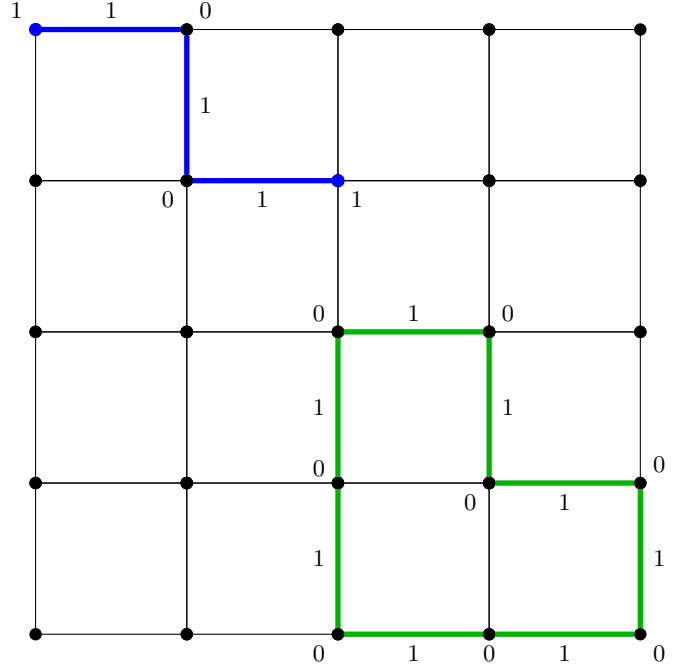

As used by RRW, the key observation is that classical lattice configurations satisfying \eqref{eq:classicalgauss} can be interpreted as the codewords defining a classical code within the larger 
space of all possible assignments of bits to the lattice variables.
The set of all possible assignments of bits is the binary vector space $\mathbb F_2^{(D+1)N}$, since there are $DN$ links and $N$ sites, and therefore $(D+1)N$ total bits that need to be assigned in order to define a classical lattice configuration.
The set of codewords $C$ is isomorphic to the linear subspace $C \cong \mathbb F_2^{DN} \subseteq \mathbb F_2^{(D+1)N}$.
The reduction in dimension by $N$ is precisely the result of the $N$ linear constraints of the form \eqref{eq:classicalgauss}.
The distance of this code is automatically $3$, because the smallest gauge invariant excitation that can be created is a pair of matter excitations on adjacent sites, separated by a single link of flux.\footnote{Technically, if the lattice has a periodic dimension spanned by just $2$ links, there can also be a non-contractible flux string that winds around in this very short direction. To keep our discussion simple, we will neglect this possibility by assuming that the number of lattice sites in each direction is large enough.}

An extension to higher distances can be motivated from this physical observation.
Namely, we duplicate the data on each site and link, so that the number of bit-flips needed to excite a flux string is increased.
For instance, if a single copy is made of every site and link variable, then it now requires $6$ bit-flips to excite the smallest flux string, instead of the $3$ bits that were needed in the previous case.
This example is shown in \cref{fig:flux-string-double}.

\begin{figure}[htp]
    \centering
    \begin{tikzpicture}
        \pgfmathsetmacro{\a}{2.0}
        \pgfmathsetmacro{\halfsep}{0.08}
        \pgfmathsetmacro{\pr}{0.075}
        \pgfmathsetmacro{\offsetx}{0.25}
        \pgfmathsetmacro{\offsety}{0.3}

        \draw (-0.5*\a,0*\a-\halfsep) -- (1.5*\a,0*\a-\halfsep);
        \draw (-0.5*\a,0*\a+\halfsep) -- (1.5*\a,0*\a+\halfsep);
        \draw (0*\a-\halfsep,0.5*\a) -- (0*\a-\halfsep,-0.5*\a);
        \draw (0*\a+\halfsep,0.5*\a) -- (0*\a+\halfsep,-0.5*\a);
        \draw (1*\a-\halfsep,0.5*\a) -- (1*\a-\halfsep,-0.5*\a);
        \draw (1*\a+\halfsep,0.5*\a) -- (1*\a+\halfsep,-0.5*\a);

        \SmallCircle[color=blue]{0*\a}{0*\a-\halfsep}{1}{\pr}
        \SmallCircle[color=blue]{0*\a}{0*\a+\halfsep}{1}{\pr}
        \SmallCircle[color=blue]{1*\a}{0*\a-\halfsep}{1}{\pr}
        \SmallCircle[color=blue]{1*\a}{0*\a+\halfsep}{1}{\pr}
        \draw[color=blue, line width=2pt] (0*\a,0*\a+\halfsep) -- (1*\a,0*\a+\halfsep);
        \draw[color=blue, line width=2pt] (0*\a,0*\a-\halfsep) -- (1*\a,0*\a-\halfsep);

        \node at (-\offsetx,\offsety) {$1$};
        \node at (0.5*\a,\offsety) {$1$};
        \node at (1*\a+\offsetx,\offsety) {$1$};
        \node at (-\offsetx,-\offsety) {$1$};
        \node at (0.5*\a,-\offsety) {$1$};
        \node at (1*\a+\offsetx,-\offsety) {$1$};
    \end{tikzpicture}
    \caption{Excitation of the smallest flux string, if there is an identical copy of every link and site variable on the lattice. This doubles the distance of the code from $3$ to $6$.}\label{fig:flux-string-double}
\end{figure}
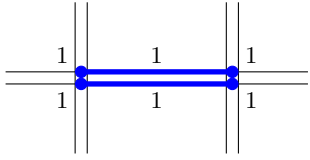

The duplication increases the number of bits that we need to keep track of in a classical configuration.
A useful notation is to write $S^{(i)}_{\mathbf{n}}$ for the $i^{\mathrm{th}}$ copy of the site variable $S_{\mathbf{n}}$, and $L^{(j)}_{\mathbf{n},\mu}$ for the $j^{\mathrm{th}}$ copy of the link variable $L_{\mathbf{n},\mu}$.
In the most general case, the number of repetitions of each site and link variable can be chosen independently across the lattice.
To describe this general case, we will find it useful to introduce integers $p_{\mathbf{n}} \in \mathbb N$ and $q_{\mathbf{n},\mu} \in \mathbb N$ that specify the number of times the site variable $S_{\mathbf{n}}$ and link variable $L_{\mathbf{n},\mu}$ have been repeated, respectively.

To understand the resulting code, we need to understand the Gauss's law constraints in terms of the duplicated variables.
For instance, see \cref{fig:gauss-vertex} for an example vertex where every $p_{\mathbf{n}} = 2$ and every $q_{\mathbf{n},\mu} = 2$.

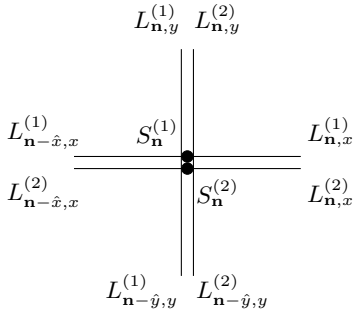
\begin{figure}[htp]
    \centering
    \begin{tikzpicture}
        \pgfmathsetmacro{\a}{2.0}
        \pgfmathsetmacro{\halfsep}{0.08}
        \pgfmathsetmacro{\pr}{0.075}
        \pgfmathsetmacro{\offsetx}{0.4}
        \pgfmathsetmacro{\offsety}{0.4}

        \draw (-0.75*\a,0*\a-\halfsep) -- (0.75*\a,0*\a-\halfsep);
        \draw (-0.75*\a,0*\a+\halfsep) -- (0.75*\a,0*\a+\halfsep);
        \draw (0*\a-\halfsep,0.75*\a) -- (0*\a-\halfsep,-0.75*\a);
        \draw (0*\a+\halfsep,0.75*\a) -- (0*\a+\halfsep,-0.75*\a);

        \SmallCircle{0*\a}{0*\a-\halfsep}{1}{\pr}
        \SmallCircle{0*\a}{0*\a+\halfsep}{1}{\pr}

        \node at (-\offsetx,\offsety) {$S^{(1)}_{\mathbf{n}}$};
        \node at (\offsetx,-\offsety) {$S^{(2)}_{\mathbf{n}}$};
        \node at (0.75*\a+\offsetx,\offsety) {$L^{(1)}_{\mathbf{n},x}$};
        \node at (0.75*\a+\offsetx,-\offsety) {$L^{(2)}_{\mathbf{n},x}$};
        \node at (-0.75*\a-\offsetx,\offsety) {$L^{(1)}_{\mathbf{n}-\hat{x},x}$};
        \node at (-0.75*\a-\offsetx,-\offsety) {$L^{(2)}_{\mathbf{n}-\hat{x},x}$};
        \node at (-\offsety,0.75*\a+\offsetx) {$L^{(1)}_{\mathbf{n},y}$};
        \node at (\offsety,0.75*\a+\offsetx) {$L^{(2)}_{\mathbf{n},y}$};
        \node at (-1.5*\offsety,-0.75*\a-0.5*\offsetx) {$L^{(1)}_{\mathbf{n}-\hat{y},y}$};
        \node at (1.5*\offsety,-0.75*\a-0.5*\offsetx) {$L^{(2)}_{\mathbf{n}-\hat{y},y}$};
    \end{tikzpicture}
    \caption{Example vertex showing labeling of duplicate lattice variables. Several copies of Gauss's law constraints are satisfied between all the duplicated variables.}\label{fig:gauss-vertex}
\end{figure}

It turns out that there are two mathematically equivalent ways to express the linear constraints enforced at a vertex $\mathbf{n}\in\Lambda$ by Gauss's law, under the assumption that the classical lattice data has simply been copied into each respective duplicated variable.
The first way to express this is to imagine selecting a representative value for the lattice variable from each link and site, and then demanding that a Gauss's law constraint holds between them.
Enforcing this for every possible choice of representatives gives a constraint of the form
\begin{equation}\label{eq:gaussrepresentative}
    S^{(a)}_{\mathbf{n}} + \sum_{\mu=1}^D \left(L^{(b_\mu)}_{\mathbf{n},\mu} + L^{(c_\mu)}_{\mathbf{n}-\hat{\mu},\mu}\right) = 0,
\end{equation}
where $1\le a\le p_{\mathbf{n}}$ is the index that selects any of the possible representatives of $S_{\mathbf{n}}$, and the sequences of indices $1\le b_{\mu} \le q_{\mathbf{n},\mu}$ and $1\le c_{\mu}\le q_{\mathbf{n}-\hat{\mu},\mu}$ select any of the possible representatives of the surrounding link variables.
The number of constraints that can be written down of the form \eqref{eq:gaussrepresentative} is given by the combinatorial factor
\begin{equation}
    \# \text{ constraints} = p_{\mathbf{n}} \times \prod_{\mu=1}^D q_{\mathbf{n},\mu} q_{\mathbf{n}-\hat{\mu},\mu},
\end{equation}
but it turns out that the number of \textit{independent} linear constraints within this collection is only
\begin{equation}\label{eq:gaussrepresentativeindependent}
    \# \text{ indep. constraints} = -2D + p_{\mathbf{n}} + \sum_{\mu=1}^D \left(q_{\mathbf{n},\mu} + q_{\mathbf{n}-\hat{\mu},\mu}\right).
\end{equation}

To understand why \eqref{eq:gaussrepresentativeindependent} holds true, we come to the second way in which the linear constraints can be expressed at any vertex $\mathbf{n}\in\Lambda$.
Namely, in homogeneous form,
\begin{align}\label{eq:gaussrepetitionconstraints}
    S^{(1)}_{\mathbf{n}} + \sum_{\mu=1}^D \left(L^{(1)}_{\mathbf{n},\mu} + L^{(1)}_{\mathbf{n}-\hat{\mu},\mu}\right) &= 0,\\
    S^{(a)}_{\mathbf{n}} + S^{(a')}_{\mathbf{n}} &= 0,\\
    L^{(b)}_{\mathbf{n},\mu} + L^{(b')}_{\mathbf{n},\mu} &= 0,\\
    L^{(c)}_{\mathbf{n}-\hat{\mu},\mu} + L^{(c')}_{\mathbf{n}-\hat{\mu},\mu} &= 0,
\end{align}
where it's important to keep in mind that addition is still defined over $\mathbb F_2$.
The linear constraints \eqref{eq:gaussrepetitionconstraints} illustrate that there is really only a single Gauss's law constraint, applied to one set of representative variables (for which we have arbitrarily chosen the index $1$).
The rest of the constraints are simply the statement that every representative of a lattice variable equals every other representative.
But this is basically just a classical repetition code on the duplicate variables.
Setting $M$ binary variables equal to each other requires $M-1$ constraints.
Therefore, the number of independent constraints is obtained by summing
\begin{align}
    \# & \text{ indep. constraints} \nonumber\\
    & \qquad =  p_{\mathbf{n}} - 1 + \sum_{\mu=1}^D\left(q_{\mathbf{n},\mu} - 1\right)\nonumber\\
    &\qquad\qquad+ \sum_{\mu=1}^D\left(q_{\mathbf{n}-\hat{\mu},\mu} - 1\right)+ 1\,,\nonumber\\
    & \qquad = \, -2D + p_{\mathbf{n}}+ \sum_{\mu=1}^D\left(q_{\mathbf{n},\mu} + q_{\mathbf{n}-\hat{\mu},\mu}\right)
\end{align}
where the final $+1$ comes from including the single representative constraint for Gauss's law.
Indeed, this exactly simplifies to \eqref{eq:gaussrepresentativeindependent}, for any lattice site $\mathbf{n}$.

Knowing the number of independent constraints at any vertex $\mathbf{n} \in \Lambda$ puts us in the perfect position to count the number of logical bits encoded by this code.
This is an important calculation to perform as a consistency check to make sure that our duplication procedure did not inadvertently add any logical DOFs that were not already present in the original $\mathbb Z_2$ LGT.
To count the number of logical DOFs, we need to know $r$, the number of rows in the full-rank parity-check matrix $H$ for this code.
Na\"ively, we could sum the number of independent linear constraints per site according to \eqref{eq:gaussrepresentativeindependent}, over all sites $\mathbf{n}\in\Lambda$,
\begin{align}
    r_{\text{na\"ive}} &= \sum_{\mathbf{n}\in\Lambda} \left(-2D + p_{\mathbf{n}} + \sum_{\mu=1}^D\left(q_{\mathbf{n},\mu} + q_{\mathbf{n}-\hat{\mu},\mu}\right)\right)\\
    & = -2DN + \sum_{\mathbf{n}\in\Lambda} p_{\mathbf{n}} + 2\sum_{\mathbf{n}\in\Lambda} \sum_{\mu=1}^D q_{\mathbf{n},\mu}.
\end{align}

But note that not all constraints are independent.
According to \eqref{eq:gaussrepetitionconstraints}, all representatives of any link adjacent to a site $\mathbf{n}\in\Lambda$ are set equal to each other by the vertex constraints at that site.
When collecting together this set of constraints for every lattice site, we double count those link constraints---once for each endpoint of the link.
This is the only double counting present, as a single representative of Gauss's law is independent at every lattice site, and the constraints equating all copies of a given site variable are also independent at different sites.
Therefore, the amount of over counting is precisely given by the number of constraints that equate the copies of each link variable, i.e.,
\begin{align}
    \# \text{ over counted} &= \sum_{\mathbf{n}\in\Lambda} \sum_{\mu=1}^D \left(q_{\mathbf{n},\mu}-1\right)\\
    & = -DN + \sum_{n\in\Lambda} \sum_{\mu=1}^D q_{\mathbf{n},\mu},
\end{align}
when summed over all lattice sites.

The number of independent parity checks is therefore
\begin{align}
    r &= r_{\text{na\"ive}} - \# \text{ over counted}\\
    & = -DN + \sum_{\mathbf{n}\in\Lambda} p_{\mathbf{n}} + \sum_{\mathbf{n}\in\Lambda} \sum_{\mu=1}^D q_{\mathbf{n},\mu}.
\end{align}
Finally, note that
\begin{equation}
    n = \sum_{\mathbf{n}\in\Lambda} p_{\mathbf{n}} + \sum_{\mathbf{n}\in\Lambda} \sum_{\mu=1}^D q_{\mathbf{n},\mu}
\end{equation}
is simply the total number of physical bits used by this code, since it adds up the number of duplicates used by every link and site variable.
The number of independent logical DOFs used by this code is then exactly
\begin{align}
    k &= n - r\\
    & = DN,
\end{align}
which indeed agrees with our expectation that the independent kinematical DOFs are in bijective correspondence with the original $DN$ link variables in the standard formulation of the $\mathbb Z_2$ LGT.

Having performed this consistency check, we can now summarize the classical construction of this code with definition \ref{def:gaussclassical} and lemma \ref{lem:gaussclassicalalternative}.

\begin{definition}[Classical Gauss's Law Code]\label{def:gaussclassical}
    Assume we are given positive integers $p_{\mathbf{n}}, q_{\mathbf{n},\mu} \in \mathbb N$ for every $\mathbf{n} \in \Lambda$ and $\mu \in \{1,\dots,D\}$.
    Let $n$ denote the total number of physical bits, defined by
    \begin{equation}
        n = \sum_{\mathbf{n}\in\Lambda} p_{\mathbf{n}} + \sum_{\mathbf{n}\in\Lambda} \sum_{\mu=1}^D q_{\mathbf{n},\mu}.
    \end{equation}
    Suppose that we identify the binary vector space $\mathbb F_2^{n}$ with assignments of bits to duplicated lattice variables
    \begin{equation}
        S^{(i)}_{\mathbf{n}}, L^{(j)}_{\mathbf{n},\mu} \in \mathbb F_2,
    \end{equation}
    where $1\le i\le p_{\mathbf{n}}$ and $1\le j\le q_{\mathbf{n},\mu}$.
    Then the \textit{classical Gauss's law code} built from the integers $p_{\mathbf{n}}$ and $q_{\mathbf{n},\mu}$ is the linear subspace $C \subseteq \mathbb F_2^n$ satisfying the constraints
    \begin{equation}
        S^{(a)}_{\mathbf{n}} + \sum_{\mu=1}^D \left(L^{(b_\mu)}_{\mathbf{n},\mu} + L^{(c_\mu)}_{\mathbf{n}-\hat{\mu},\mu}\right) = 0,
    \end{equation}
    for every integer $1 \le a \le p_{\mathbf{n}}$, and every pair of integer sequences $b_{\mu}$ and $c_{\mu}$ whose values satisfy $1\le b_{\mu} \le q_{\mathbf{n},\mu}$ and $1\le c_{\mu} \le q_{\mathbf{n}-\hat{\mu},\mu}$.
\end{definition}

\begin{lemma}[Alternative Characterization]\label{lem:gaussclassicalalternative}
    Assume we are given positive integers $p_{\mathbf{n}}, q_{\mathbf{n},\mu} \in \mathbb N$ for every $\mathbf{n} \in \Lambda$ and $\mu \in \{1,\dots,D\}$.
    Let $n$ denote the total number of physical bits, defined by
    \begin{equation}
        n = \sum_{\mathbf{n}\in\Lambda} p_{\mathbf{n}} + \sum_{\mathbf{n}\in\Lambda} \sum_{\mu=1}^D q_{\mathbf{n},\mu}.
    \end{equation}
    Suppose that we identify the binary vector space $\mathbb F_2^{n}$ with assignments of bits to duplicated lattice variables
    \begin{equation}
        S^{(i)}_{\mathbf{n}}, L^{(j)}_{\mathbf{n},\mu} \in \mathbb F_2,
    \end{equation}
    where $1\le i\le p_{\mathbf{n}}$ and $1\le j\le q_{\mathbf{n},\mu}$.
    Let $C \subseteq \mathbb F_2^n$ denote the linear subspace satisfying the constraints
    \begin{align}
        S^{(1)}_{\mathbf{n}} + \sum_{\mu=1}^D \left(L^{(1)}_{\mathbf{n},\mu} + L^{(1)}_{\mathbf{n}-\hat{\mu},\mu}\right) &= 0,\\
        S^{(a)}_{\mathbf{n}} + S^{(a')}_{\mathbf{n}} &= 0,\\
        L^{(b)}_{\mathbf{n},\mu} + L^{(b')}_{\mathbf{n},\mu} &= 0,\\
        L^{(c)}_{\mathbf{n}-\hat{\mu},\mu} + L^{(c')}_{\mathbf{n}-\hat{\mu},\mu} &= 0,
    \end{align}
    for every pair of integers $1\le a,a'\le p_{\mathbf{n}}$, every pair of integers $1\le b,b'\le q_{\mathbf{n},\mu}$, and every pair of integers $1\le c,c'\le q_{\mathbf{n}-\hat{\mu},\mu}$.
    Then $C$ is exactly the classical Gauss's law code defined by the collection of integers $p_{\mathbf{n}}$ and $q_{\mathbf{n},\mu}$.
\end{lemma}

We currently know the number of physical and logical bits used by the classical Gauss's law code defined from a collection of integers $p_{\mathbf{n}}$ and $q_{\mathbf{n},\mu}$.
But we have not yet discussed the code distance in this general framework, apart from the initial motivation that led us to introduce duplications of lattice variables in the first place.

As we saw in \cref{fig:flux-string-double}, the distance is related to gauge-invariant excitations that correspond to flux strings on the lattice.
To create any such excitation, all the bits of every participating lattice variable must be flipped.
The specific lattice variables involved will depend on whether we have an open string or closed string, and the specific shape of that string.

An open string $\alpha$ corresponds to a continuous, non-intersecting path of links between two lattice sites $\mathbf{n},\mathbf{n}'\in\Lambda$ with $\mathbf{n}\neq\mathbf{n'}$.
For a path spanning $\ell$ links, we can write $L_{\mathbf{n_1},\mu_1}, L_{\mathbf{n_2},\mu_2}, \dots L_{\mathbf{n_{\ell}},\mu_{\ell}}$ for the links on this path.
In this case, the weight of the logical error corresponding to exciting this open string is given by
\begin{equation}
    \wt(\alpha) = p_{\mathbf{n}} + q_{\mathbf{n_1},\mu_1} + q_{\mathbf{n_2},\mu_2} + \cdots + q_{\mathbf{n_{\ell}},\mu_{\ell}} + p_{\mathbf{n'}}.
\end{equation}
This is the number of bit-flips needed to flip every copy of each lattice variable involved in the string $\alpha$; namely, the two lattice sites at the endpoints, and all links along the path.

A closed string $\beta$ corresponds to a continuous, non-intersecting closed path of links.
For a closed path spanning $\ell$ links, we can again write $L_{\mathbf{n_1},\mu_1}, L_{\mathbf{n_2},\mu_2}, \dots L_{\mathbf{n_{\ell}},\mu_{\ell}}$ for the links along this path.
The weight in this case is given by
\begin{equation}
    \wt(\beta) = q_{\mathbf{n_1},\mu_1} + q_{\mathbf{n_2},\mu_2} + \cdots + q_{\mathbf{n_{\ell}},\mu_{\ell}},
\end{equation}
which reflects the fact that there are no matter excitations on a closed string.

The classical code distance is simply the minimum weight of any flux string, taken over all open strings and closed strings on the lattice.
This is summarized by Lemma~\ref{lem:gaussdistance}.

\begin{lemma}[Classical Gauss's Law Code Distance]\label{lem:gaussdistance}
    Let $C$ be the classical Gauss's law code defined by the collection of integers $p_{\mathbf{n}},q_{\mathbf{n},\mu} \in \mathbb N$.
    Let
    \begin{equation}
        w_o \equiv \min_{\alpha} \wt(\alpha)
    \end{equation}
    denote the minimum weight of any open string $\alpha$, and let
    \begin{equation}
        w_c \equiv \min_{\beta} \wt(\beta)
    \end{equation}
    denote the minimum weight of any closed string $\beta$.
    Then the code distance of $C$ is exactly
    \begin{equation}
        d = \min(w_o, w_c).
    \end{equation}
\end{lemma}

Note that for the case where all $p_{\mathbf{n}} = q_{\mathbf{n},\mu} = p$ equal, one finds $w_o = 3 p$, $w_c = 4p$, giving $d = 3p$.
For other choices of $p_{\mathbf{n}}$ and $q_{\mathbf{n},\mu}$, however, other values of $d$ can be achieved. 
For example, using $q_{\mathbf{n},\mu} = 1$, while $p_{\rm even} = 1$, $p_{\rm odd} = 3$ gives $w_o = w_c = d = 4$, as can be seen in \cref{fig:longer-string}. 
From this example, one also notices the important fact that the minimum must be taken over flux strings of \textit{every} possible length on the lattice.
The smallest string measured according to lattice units will not necessarily have the smallest weight.
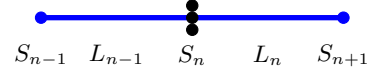
\begin{figure}[htp]
    \centering
    \begin{tikzpicture}
        \pgfmathsetmacro{\a}{2.0}
        \pgfmathsetmacro{\halfsep}{0.08}
        \pgfmathsetmacro{\pr}{0.075}
        \pgfmathsetmacro{\offset}{0.5}

        \SmallCircle[color=blue]{0*\a}{0*\a}{1}{\pr}
        \SmallCircle[color=blue]{2*\a}{0*\a}{1}{\pr}
        \draw[color=blue, line width=2pt] (0*\a,0*\a) -- (2*\a,0*\a);
        \SmallCircle[color=black]{1*\a}{0*\a}{1}{\pr}
        \SmallCircle[color=black]{1*\a}{0*\a+2*\halfsep}{1}{\pr}
        \SmallCircle[color=black]{1*\a}{0*\a-2*\halfsep}{1}{\pr}

        \node at (0*\a,0*\a-\offset) {$S_{n-1}$};
        \node at (0.5*\a,0*\a-\offset) {$L_{n-1}$};
        \node at (1*\a,0*\a-\offset) {$S_{n}$};
        \node at (1.5*\a,0*\a-\offset) {$L_{n}$};
        \node at (2*\a,0*\a-\offset) {$S_{n+1}$};
    \end{tikzpicture}
    \caption{Example where a longer string has lower weight. The number of copies of each lattice variable are given by $(p_{n-1},q_{n-1},p_{n},q_{n},p_{n+1}) = (1,1,3,1,1)$. The length-$2$ string shown in blue has weight $4$, while the length-$1$ strings that can be made here would both have weight $5$.}\label{fig:longer-string}
\end{figure} 
Understanding how to choose the values $p_{\mathbf{n}}$ and $q_{\mathbf{n},\mu}$ to obtain a given value $d$ in the most efficient way will be discussed in~\cref{sec:optimal}.

Having discussed our construction of classical Gauss's law codes to arbitrary distance, all that remains is to form a full quantum code by concatenating with an appropriately chosen inner code.
In \cref{sec:prelim}, we opted for a basic definition of concatenation (see Definition~\ref{def:concatenation}) which functions purely as an operation between two classical codes, where one code handles bit-flips and the other handles phase-flips.
Therefore, to keep the discussion both accessible and concise, we write Definition~\ref{def:gaussquantum} under the assumption that the concatenation is performed using another classical code.

\begin{definition}[Quantum Gauss's Law Code]\label{def:gaussquantum}
    Let $C_\OUT = [n_\OUT,k_\OUT,d_\OUT]$ be a classical Gauss's law code, and let $C_\IN = [n_\IN,1,d_\IN]$ be an arbitrary classical code that encodes a single logical bit.
    Then the concatenation
    \begin{equation}
        C_\OUT \circ C_\IN = \llbracket n_\OUT n_\IN, k_\OUT, \min(d_\OUT,d_\IN) \rrbracket
    \end{equation}
    is called the \textit{quantum Gauss's law code} built from $C_\OUT$ and $C_\IN$.
\end{definition}

It is worth noting that the construction of the full quantum code works with any general concatenation strategy.
This does lead to a few notable differences between quantum Gauss's law codes defined by the basic concatenation strategy (Definition \ref{def:gaussquantum}) and those obtained from other choices of concatenation.
For instance, one may choose a quantum code
$C_\IN=\llbracket n_\IN,k_\IN=n_\OUT,d_\IN\rrbracket$ as an inner code, identify its logical qubits with the physical qubits of $C_\OUT$, and replace the Pauli operators in the stabilizers of $C_\OUT$ with logical Pauli operators of $C_\IN$. This can lead to a more compact encoding, i.e., with a higher encoding rate $k/n$. The construction for the code distance three was given in Ref.~\cite{Spagnoli:2024mib}.
However, the stabilizers after concatenation may have weights that grow with the lattice size $N$, making their measurement more costly.
On the other hand, Definition~\ref{def:gaussquantum} \textit{always} produces a quantum \textit{low-density parity check} (LDPC) code, meaning that the weight of stabilizer generators and the number of stabilizers acting on each qubit are bounded independently of $N$.

Therefore, the class of quantum Gauss's law codes defined by concatenation in the form of Definition~\ref{def:gaussquantum} is a reasonable target for further investigation.
Our optimality theorems in \cref{sec:optimal} for the encoding rate---specifically Theorems~\ref{thm:optimal1D} and \ref{thm:optimal2D}---depend mostly on our construction of the classical Gauss's law code, as elaborated further in \cref{app:proofs}.
Therefore, these results suffice to provide an insight into what we can expect from any concatenated construction for a quantum Gauss's law code.

\subsection{Optimal Encoding Rate}\label{sec:optimal}

Now we come to the question of the optimal encoding rate $k/n$ for quantum Gauss's law codes defined by our previous class of constructions (Definitions \ref{def:gaussclassical} and \ref{def:gaussquantum}).
Equivalently, we are interested in determining which specific choice of multiplicities $p_{\mathbf{n}}$ and $q_{\mathbf{n},\mu}$ is necessary to minimize the total number of physical qubits $n$, subject to a desired code distance $d$ and fixed lattice size $N$.
The answer can get quite messy to write down, due to edge cases that depend on the parity of lattice sites, and so on---but the underlying logic is always the same.
In order to keep our discussion clean and concise, and to ensure that the logic is easy to see, we will make a few simplifying assumptions:
\begin{enumerate}[label=(\roman*)]
    \item We assume that the number of lattice sites $N_{\mu}$ is even along every axis $\mu \in \{1,2,\dots,D\}$.
    \item We assume that the number of lattice sites $N_{\mu}$ is larger than the code distance $d$ along every axis $\mu$.
\end{enumerate}

Assumption (i) is a standard requirement for the $\mathbb Z_2$ LGT coupled to staggered fermions; and while it is not strictly necessary for a more generic theory, it will help to write a relatively simple expression for the optimal choice of $p_{\mathbf{n}}$ and $q_{\mathbf{n},\mu}$, without getting into distracting technicalities.
Assumption (ii) is useful just for the sake of ensuring that a non-contractible string doesn't contribute to the optimization procedure---as long as every $N_{\mu} \ge d$, it is guaranteed that every non-contractible string already has weight $\ge d$, and therefore there are no constraints placed on $p_{\mathbf{n}}$ and $q_{\mathbf{n},\mu}$ by such strings. 
For typical lattices that are used to extract interesting physics, this constraint is always satisfied.

We discuss only the $1$D and $2$D cases in explicit detail.
The $1$D case is special because there are no plaquettes, and therefore all logical bit-flip errors arise from open strings.
The $2$D case introduces the full complexity of the problem, requiring consideration of both open strings and closed strings.
The extension to higher dimensions follows straightforwardly once the $2$D case is understood, and we do not write it explicitly here.
For the technical proofs of our optimality theorems, see \cref{app:proofs}.

\subsubsection{$1$D Case}\label{sec:optimal1D}

The intuition for the optimally efficient code construction in $1$D is as follows: the smallest possible gauge-invariant excitation is formed from exciting a pair of matter sites separated by $1$ link.
This excitation requires \textit{two} lattice sites to be flipped, but only \textit{one} link to be flipped.
Therefore, there is a sense in which the code distance is affected more by the difficulty of flipping lattice sites, in comparison to that of flipping links.
Of course, the situation is slightly complicated due to the fact that longer strings also need to be considered, as in \cref{fig:longer-string}; these technical details are taken into account in Theorem~\ref{thm:optimal1D}.
In fact, we find that to obtain the code that minimizes the number of physical qubits needed, we must indeed delegate the entirety of all duplication to just the lattice sites, prior to concatenation.

\begin{theorem}[Optimal Gauss's Law Code in $1$D]\label{thm:optimal1D}
    Let $d \ge 3$ and $N \ge d$ be given, with $N$ even.
    Suppose that $C_\OUT\circ C_\IN = \llbracket n, k, d\rrbracket$ is a quantum Gauss's law code (see Def.~\ref{def:gaussquantum}) with distance $d$ for the $1$D $\mathbb Z_2$ LGT defined by Hamiltonian \eqref{eq:Hstandard} on $N$ lattice sites.
    Then the following two results hold:
    \begin{enumerate}[label=(\roman*)]
        \item The encoding rate is bounded by
        \begin{equation}\label{eq:bound1D}
            \frac{k}{n} \le \frac2{d(d+1)}
        \end{equation}
        \item There exists a quantum Gauss's law code $C_\OUT\circ C_\IN$ 
        (i.e., a specific choice of duplication variables $(p_{\mathbf{n}},q_{\mathbf{n}})$ and a classical code $C_\IN$)
        that saturates the bound \eqref{eq:bound1D}. 
    \end{enumerate}
\end{theorem}
In other words, this theorem states that the optimal Gauss's law code has $\llbracket  Nd(d+1)/2,N,d\rrbracket$.
The proof of theorem \ref{thm:optimal1D} can be found in \cref{app:proofs}.
In this section, we provide just the optimal code construction.
The concatenation step is straightforward---take $C_\OUT$ to be a classical distance-$d$ Gauss's law code for bit-flips using as few physical bits as possible, and take $C_\IN$ to be a classical repetition code $C_\IN = [d,1,d]$ for phase-flips.

The interesting part of the theorem comes from considering how to choose $C_\OUT$ with the minimum possible number of physical bits.
Recall that the classical Gauss's law code $C_\OUT$ for a $1$D lattice is defined entirely from a sequence of integers, $p_n$ and $q_n$.
Therefore, specifying $p_n$ and $q_n$ for each lattice site index $n \in \Lambda = \{1, 2, \dots, N\}$ is enough to define the classical code.
An optimal code is obtained by the following selection:
\begin{align}\label{eq:optimal1Dpq}
    p_n &=
    \begin{cases}
        \lceil\frac{d-1}{2}\rceil, & n \equiv 0 \pmod 2,\\
        \lceil\frac{d}{2}\rceil - 1, & n \equiv 1 \pmod 2,
    \end{cases}\\
    q_n &= 1.
\end{align}
The effect of this choice is to condition the duplication of lattice sites by the parity of $d$.
If $d$ is odd, then all sites are uniformly duplicated by the same amount.
If $d$ is even, then the odd-numbered sites are duplicated with one less copy than the even-numbered sites. 
These two cases are shown in \cref{fig:optimal-1D}.
It's easy to see that the classical code distance $d$ is recovered in both cases.
Upon concatenation with $C_\IN = [d,1,d]$, it's also easy to verify that the bound from theorem \ref{thm:optimal1D} is saturated, regardless of the parity of $d$.

\begin{figure}[htp]
    \centering
    \begin{minipage}{\linewidth}
        \begin{tikzpicture}
            \pgfmathsetmacro{\a}{2.0}
            \pgfmathsetmacro{\halfsep}{0.08}
            \pgfmathsetmacro{\pr}{0.075}
            \pgfmathsetmacro{\offset}{0.5}

            \draw (-1*\a,0*\a) -- (3*\a,0*\a);
            
            \SmallCircle[color=blue]{1*\a}{0*\a+\halfsep}{1}{\pr}
            \SmallCircle[color=blue]{1*\a}{0*\a-\halfsep}{1}{\pr}
            \draw[color=blue, line width=2pt] (0*\a,0*\a) -- (1*\a,0*\a);
            \SmallCircle[color=blue]{0*\a}{0*\a}{1}{\pr}
            \SmallCircle[color=black]{2*\a}{0*\a}{1}{\pr}
    
            \node at (0*\a,0*\a-\offset) {$S_{2i-1}$};
            \node at (0.5*\a,0*\a-\offset) {$L_{2i-1}$};
            \node at (1*\a,0*\a-\offset) {$S_{2i}$};
            \node at (1.5*\a,0*\a-\offset) {$L_{2i}$};
            \node at (2*\a,0*\a-\offset) {$S_{2i+1}$};
        \end{tikzpicture}
    \end{minipage}

    \vspace{3em}

    \begin{minipage}{\linewidth}
        \begin{tikzpicture}
            \pgfmathsetmacro{\a}{2.0}
            \pgfmathsetmacro{\halfsep}{0.08}
            \pgfmathsetmacro{\pr}{0.075}
            \pgfmathsetmacro{\offset}{0.5}

            \draw (-1*\a,0*\a) -- (3*\a,0*\a);
            
            \SmallCircle[color=blue]{0*\a}{0*\a+\halfsep}{1}{\pr}
            \SmallCircle[color=blue]{0*\a}{0*\a-\halfsep}{1}{\pr}
            \draw[color=blue, line width=2pt] (0*\a,0*\a) -- (1*\a,0*\a);
            \SmallCircle[color=blue]{1*\a}{0*\a+\halfsep}{1}{\pr}
            \SmallCircle[color=blue]{1*\a}{0*\a-\halfsep}{1}{\pr}
            \SmallCircle[color=black]{2*\a}{0*\a+\halfsep}{1}{\pr}
            \SmallCircle[color=black]{2*\a}{0*\a-\halfsep}{1}{\pr}
    
            \node at (0*\a,0*\a-\offset) {$S_{2i-1}$};
            \node at (0.5*\a,0*\a-\offset) {$L_{2i-1}$};
            \node at (1*\a,0*\a-\offset) {$S_{2i}$};
            \node at (1.5*\a,0*\a-\offset) {$L_{2i}$};
            \node at (2*\a,0*\a-\offset) {$S_{2i+1}$};
        \end{tikzpicture}
    \end{minipage}
    \caption{Duplicating lattice variables for an optimal code. Each link or site (and any duplicate) corresponds to a single bit in the outer bit-flip code. The minimum-weight logical error is shown in both cases (blue), which corresponds to a flux string excitation on the lattice. Top: distance $d=4$ (even case). Bottom: distance $d=5$ (odd case).}\label{fig:optimal-1D}
\end{figure}
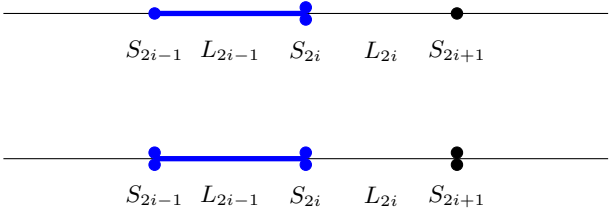

Roughly speaking, it works to choose $p_n$ with the split cases in \eqref{eq:optimal1Dpq} due to the fact that a length-$1$ flux string always uses exactly one endpoint from both the even and odd sub-lattices; and therefore, the code distance increases by $1$ upon incrementing the site multiplicity on just one of the sub-lattices.

\subsubsection{$2$D Case}
\label{sec:optimal2D}

Deducing the optimal code construction in the $2$D case is substantially more complicated than the $1$D case, due to the presence of closed strings.
Nevertheless, there is a generalization of the methodology used in the $1$D case that applies directly to the $2$D case and higher dimensions.
The idea is to separately keep track of the minimum weights $w_o$ and $w_c$ for open strings and closed strings, respectively.
By partitioning the lattice variables based on parity, we can control $w_o$ and $w_c$, and therefore also the code distance.

For instance, we partition the set of lattice sites $\mathcal S$ as
\begin{align}
    \mathcal S &= \mathcal S_1 \sqcup \mathcal S_2,\\
    \mathcal S_1 &\equiv \{S_{\mathbf{n}} \mid n_1 + n_2 \equiv 0\pmod{2}\},\\
    \mathcal S_2 &\equiv \{S_{\mathbf{n}} \mid n_1 + n_2 \equiv 1\pmod{2}\}.
\end{align}
This is the standard staggering into the even and odd sub-lattices.
The point is that when all multiplicities on a sub-lattice of staggered sites is incremented, $w_o$ is expected to increase by $1$ (exactly as in the $1$D case), while $w_c$ remains unchanged.

A similar property can be engineered for minimum-weight closed strings, by duplicating links.
Specifically, one can partition the set of links $\mathcal L$ into $4$ groups:
\begin{align}
    \mathcal L &= \mathcal L_1 \sqcup \mathcal L_2 \sqcup \mathcal L_3 \sqcup \mathcal L_4,\\
    \mathcal L_1 &\equiv \{L_{\mathbf{n},\mu} \mid n_2 \equiv 0\pmod 2,\quad \mu = x\},\\
    \mathcal L_2 &\equiv \{L_{\mathbf{n},\mu} \mid n_2 \equiv 1\pmod 2,\quad \mu = x\},\\
    \mathcal L_3 &\equiv \{L_{\mathbf{n},\mu} \mid n_1 \equiv 0\pmod 2,\quad \mu = y\},\\
    \mathcal L_4 &\equiv \{L_{\mathbf{n},\mu} \mid n_1 \equiv 1\pmod 2,\quad \mu = y\}.
\end{align}
To put it succinctly, the links are categorized by: (i) whether they are horizontal or vertical; and (ii) the parity of the transverse coordinate.
This partition has the property that every plaquette on the lattice uses exactly one link from each set ($\mathcal L_1$, $\mathcal L_2$, $\mathcal L_3$, \textit{and} $\mathcal L_4$) of this partition.
This ultimately has the effect that incrementing the multiplicities of every link in a fixed set $\mathcal L_i$ increases $w_c$ by $1$, while leaving $w_o$ unchanged.
See \cref{fig:partition-2D} for an illustration that summarizes the partitioning.

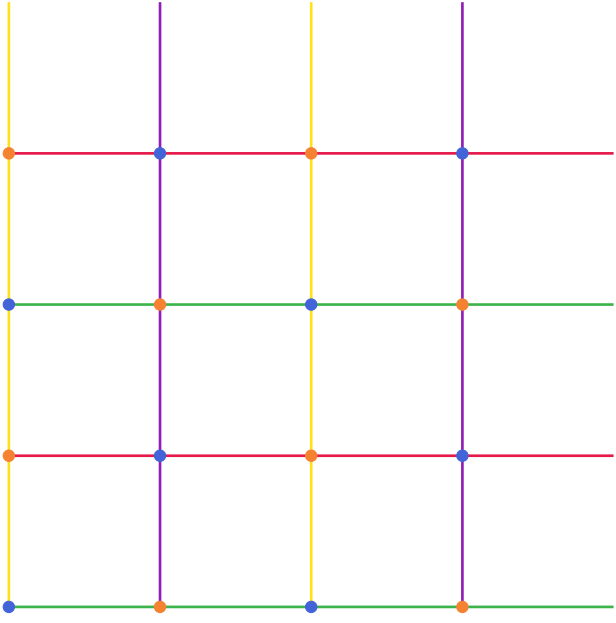
\begin{figure}[htp]
    \centering
    \begin{tikzpicture}
        \pgfmathsetmacro{\a}{2.0}
        \pgfmathsetmacro{\offsetx}{0.25}
        \pgfmathsetmacro{\offsety}{0.25}
        \pgfmathsetmacro{\pr}{0.075}
        
        \foreach \i in {0,2} {
            \draw[color=colorB, line width=1pt] (0,\i*\a) -- (4*\a,\i*\a);
            \draw[color=colorA, line width=1pt] (0,\i*\a+\a) -- (4*\a,\i*\a+\a);
            \draw[color=colorC, line width=1pt] (\i*\a,0) -- (\i*\a,4*\a);
            \draw[color=colorF, line width=1pt] (\i*\a+\a,0) -- (\i*\a+\a,4*\a);
        }
        
        \foreach \i in {0,2} {
            \foreach \j in {0,2} {
                \SmallCircle[color=colorD]{\i*\a}{\j*\a}{1}{\pr}
                \SmallCircle[color=colorE]{\i*\a+\a}{\j*\a}{1}{\pr}
                \SmallCircle[color=colorE]{\i*\a}{\j*\a+\a}{1}{\pr}
                \SmallCircle[color=colorD]{\i*\a+\a}{\j*\a+\a}{1}{\pr}
            }
        }
    \end{tikzpicture}
    \caption{Simultaneous partitioning of lattice sites and links. The two sub-lattices are colored blue ($\mathcal S_1$) and orange ($\mathcal S_2$). The partitioned links are colored as green ($\mathcal L_1$), red ($\mathcal L_2$), yellow ($\mathcal L_3$), and purple ($\mathcal L_4$). Note that every length-$1$ open string touches both sub-lattices (blue and orange), and every plaquette uses all the link colors (green, red, yellow, and purple).}\label{fig:partition-2D}
\end{figure}

As in the $1$D case, there are subtleties regarding open strings and closed strings with larger size, and these technicalities are taken into account to yield Theorem~\ref{thm:optimal2D}.
We find that the optimal code can be built from this incremental procedure, performed in parallel on both lattice sites and links, with minor modifications designed to achieve exact optimality.

\begin{theorem}[Optimal Gauss's Law Code in $2$D]\label{thm:optimal2D}
    Let $d \ge 3$ and $N_x,N_y \ge d$ be given, with $N_x,N_y$ even.
    Suppose that $C_\OUT\circ C_\IN = \llbracket n, k, d\rrbracket$ is a quantum Gauss's law code (see Def.~\ref{def:gaussquantum}) with distance $d$ for the $2$D $\mathbb Z_2$ LGT defined by Hamiltonian \eqref{eq:Hstandard} on an $N_x \times N_y$ lattice.
    Then the following two results hold:
    \begin{enumerate}[label=(\roman*)]
        \item The encoding rate is bounded by
        \begin{equation}\label{eq:bound2D}
            \frac{k}{n} \le
            \begin{cases}
                \frac 29, & d = 3,\\
                \frac{16}{7d^2}, & d > 3, \quad d:  {\rm even},\\
                \frac{16}{(7d+1)d} & d > 3, \quad d: {\rm odd}.
            \end{cases}
        \end{equation}
        \item There exists a quantum Gauss's law code $C_\OUT\circ C_\IN$ that saturates the bound \eqref{eq:bound2D}. 
    \end{enumerate}
\end{theorem}

In other words, the optimal Gauss's law code has a parameter
\begin{align}
    \begin{cases}
        \llbracket 9N, 2N, 3\rrbracket & d=3\,,
        \\
        \llbracket \frac{7d^2}{8}N, 2N, d\rrbracket & d:\text{even}\,,
        \\
        \llbracket \frac{(7d+1)d}{8}N, 2N, d\rrbracket & d:\text{odd}\,.
    \end{cases}
        \end{align}
The proof of this theorem is again provided in \cref{app:proofs}.
For now, we just provide the optimal code construction.
As in the $1$D case, it suffices to specify the values of all $p_{\mathbf{n}}$ and $q_{\mathbf{n},\mu}$.
It is easiest to provide these multiplicities in four separate cases, depending on the remainder of $d$ modulo~$4$.
An explicit example for the classical Gauss's law code with $d = 5$ is shown in \cref{fig:optimal-2D}.

\begin{figure}[htp]
    \centering
    \begin{tikzpicture}
        \pgfmathsetmacro{\a}{2.0}
        \pgfmathsetmacro{\offset}{0.3}
        \pgfmathsetmacro{\pr}{0.075}
        \pgfmathsetmacro{\prr}{0.08}
        \foreach \i in {0} {
            \LPlaquette{\i*\a}{0*\a}{\a}{\prr}{black}
            \draw[color=black] (\i*\a+\a,0) -- (\i*\a+2*\a,0);
            \draw[color=black] (\i*\a+\a,\a) -- (\i*\a+2*\a,\a);
            \LPlaquette{\i*\a+2*\a}{0*\a}{\a}{\prr}{black}
            \SmallCircle[draw=black,fill=black]{\i*\a-\prr}{0*\a}{1}{\pr}
            \SmallCircle[draw=black,fill=black]{\i*\a+\prr}{0*\a}{1}{\pr}
            \SmallCircle[draw=black,fill=black]{\i*\a+\a-\prr}{0*\a}{1}{\pr}
            \SmallCircle[draw=black,fill=black]{\i*\a+\a+\prr}{0*\a}{1}{\pr}
            \SmallCircle{\i*\a+2*\a-\prr}{0*\a}{1}{\pr}
            \SmallCircle{\i*\a+2*\a+\prr}{0*\a}{1}{\pr}
            \SmallCircle{\i*\a+3*\a-\prr}{0*\a}{1}{\pr}
            \SmallCircle{\i*\a+3*\a+\prr}{0*\a}{1}{\pr}
            
            \draw[color=black] (\i*\a-\prr,\a) -- (\i*\a-\prr,2*\a);
            \draw[color=black] (\i*\a+\prr,\a) -- (\i*\a+\prr,2*\a);
            \draw[color=black] (\i*\a+\a,\a) -- (\i*\a+\a,2*\a);
            \LPlaquette{\i*\a+2*\a}{1*\a}{\a}{\prr}{black}
            \SmallCircle[draw=black,fill=black]{\i*\a-\prr}{1*\a}{1}{\pr}
            \SmallCircle[draw=black,fill=black]{\i*\a+\prr}{1*\a}{1}{\pr}
            \SmallCircle[draw=black,fill=black]{\i*\a+\a-\prr}{1*\a}{1}{\pr}
            \SmallCircle[draw=black,fill=black]{\i*\a+\a+\prr}{1*\a}{1}{\pr}
            \SmallCircle{\i*\a+2*\a-\prr}{1*\a}{1}{\pr}
            \SmallCircle{\i*\a+2*\a+\prr}{1*\a}{1}{\pr}
            \SmallCircle{\i*\a+3*\a-\prr}{1*\a}{1}{\pr}
            \SmallCircle{\i*\a+3*\a+\prr}{1*\a}{1}{\pr}
            
            \LPlaquette{\i*\a}{2*\a}{\a}{\prr}{black}
            \RPlaquette{\i*\a+\a}{2*\a}{\a}{\prr}{black}
            \LPlaquette{\i*\a+2*\a}{2*\a}{\a}{\prr}{black}
            \SmallCircle{\i*\a-\prr}{2*\a}{1}{\pr}
            \SmallCircle{\i*\a+\prr}{2*\a}{1}{\pr}
            \SmallCircle{\i*\a+\a-\prr}{2*\a}{1}{\pr}
            \SmallCircle{\i*\a+\a+\prr}{2*\a}{1}{\pr}
            \SmallCircle{\i*\a+2*\a-\prr}{2*\a}{1}{\pr}
            \SmallCircle{\i*\a+2*\a+\prr}{2*\a}{1}{\pr}
            \SmallCircle{\i*\a+3*\a-\prr}{2*\a}{1}{\pr}
            \SmallCircle{\i*\a+3*\a+\prr}{2*\a}{1}{\pr}
        }
        \foreach \i in {0,2} {
            \draw (\i*\a-\prr,3*\a) -- (\i*\a-\prr,4*\a);
            \draw (\i*\a+\prr,3*\a) -- (\i*\a+\prr,4*\a);
            \draw (\i*\a+\a,3*\a) -- (\i*\a+\a,4*\a);
            \draw (\i*\a+\a,3*\a) -- (\i*\a+\a,4*\a);
            
            \draw (3*\a,\i*\a) -- (4*\a,\i*\a);
            \draw (3*\a,\i*\a+\a) -- (4*\a,\i*\a+\a);
        }
        \SmallCircle{-\prr}{3*\a}{1}{\pr}
        \SmallCircle{\prr}{3*\a}{1}{\pr}
        \SmallCircle{\a-\prr}{3*\a}{1}{\pr}
        \SmallCircle{\a+\prr}{3*\a}{1}{\pr}
        \SmallCircle{2*\a-\prr}{3*\a}{1}{\pr}
        \SmallCircle{2*\a+\prr}{3*\a}{1}{\pr}
        \SmallCircle{3*\a-\prr}{3*\a}{1}{\pr}
        \SmallCircle{3*\a+\prr}{3*\a}{1}{\pr}

        \pgfmathsetmacro{\lpr}{0.08}
        \pgfmathsetmacro{\lprr}{0.085}
        \SmallCircle[color=blue]{1*\a+\lprr}{3*\a}{1}{\lpr}
        \SmallCircle[color=blue]{1*\a-\lprr}{3*\a}{1}{\lpr}
        \draw[color=blue, line width=2pt] (1*\a,3*\a) -- (1*\a,2*\a);
        \SmallCircle[color=blue]{1*\a+\lprr}{2*\a}{1}{\lpr}
        \SmallCircle[color=blue]{1*\a-\lprr}{2*\a}{1}{\lpr}

        \pgfmathsetmacro{\offsethor}{0.155}
        \pgfmathsetmacro{\offsetver}{0.05}
        \pgfmathsetmacro{\offsetdoublever}{0.075}
        \draw[color=green!70!black, line width=2pt] (2*\a+\offsethor,2*\a) -- (3*\a-\offsethor,2*\a);
        \draw[color=green!70!black, line width=2pt] (3*\a,2*\a-\offsetver) -- (3*\a,1*\a+\offsetver);
        \draw[color=green!70!black, line width=2pt] (2*\a+\offsethor,1*\a) -- (3*\a-\offsethor,1*\a);
        \draw[color=green!70!black, line width=2pt] (2*\a-0.075,2*\a-\offsetdoublever) -- (2*\a-0.075,1*\a+\offsetdoublever);
        \draw[color=green!70!black, line width=2pt] (2*\a+0.075,2*\a-\offsetdoublever) -- (2*\a+0.075,1*\a+\offsetdoublever);
    \end{tikzpicture}
    \caption{Optimal code construction for distance $d=5$. Each link or site (and any duplicate) corresponds to a single bit in the outer bit-flip code. Blue: logical error with weight $5$ corresponding to an open string. Green: logical error with weight $5$ corresponding to a closed string.}\label{fig:optimal-2D}
\end{figure}
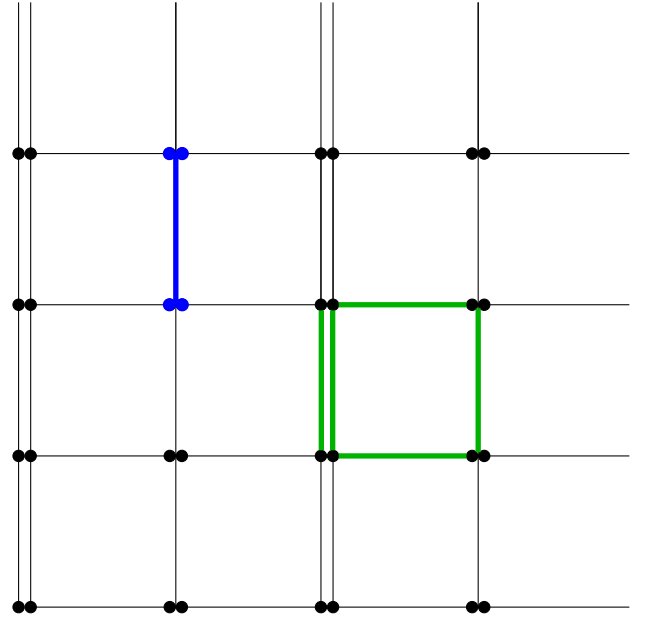

In what follows, we will assume $d > 3$, since the $d = 3$ case is trivial.
There are four separate cases to consider.
If $d \equiv 0\pmod{4}$, then let $d = 4a$, and take
\begin{align}
    p_{\mathbf{n}} &=
    \begin{cases}
        \lfloor\frac{3a}{2}\rfloor, & S_{\mathbf{n}} \in \mathcal S_1,\\
        \lceil\frac{3a}{2}\rceil, & S_{\mathbf{n}} \in \mathcal S_2,
    \end{cases}\\
    q_{\mathbf{n},\mu} &= a.
\end{align}
If $d \equiv 1\pmod{4}$, then let $d = 4a+1$, and take
\begin{align}
    p_{\mathbf{n}} &=
    \begin{cases}
        \lfloor\frac{3a+1}{2}\rfloor, & S_{\mathbf{n}} \in \mathcal S_1,\\
        \lceil\frac{3a+1}{2}\rceil, & S_{\mathbf{n}} \in \mathcal S_2,
    \end{cases}\\
    q_{\mathbf{n},\mu} &=
    \begin{cases}
        a, & L_{\mathbf{n},\mu} \in \mathcal L_1 \cup \mathcal L_2 \cup \mathcal{L}_4,\\
        a+1, & L_{\mathbf{n},\mu} \in \mathcal L_3.
    \end{cases}
\end{align}
If $d \equiv 2\pmod{4}$, then let $d = 4a+2$, and take
\begin{align}
    p_{\mathbf{n}} &=
    \begin{cases}
        2a, & n_1 \equiv n_2 \equiv 0\pmod{2},\\
        2a+1, & n_1 \equiv n_2 \equiv 1\pmod{2},\\
        a+1, & S_{\mathbf{n}} \in \mathcal S_2,
    \end{cases}\\
    q_{\mathbf{n},\mu} &=
    \begin{cases}
        a+1, & L_{\mathbf{n},\mu} \in \mathcal L_1 \cup \mathcal L_3,\\
        a, & L_{\mathbf{n},\mu} \in \mathcal L_2 \cup \mathcal L_4.
    \end{cases}
\end{align}
If $d \equiv 3\pmod{4}$, then let $d = 4a+3$, and take
\begin{align}
    p_{\mathbf{n}} &=
    \begin{cases}
        2a+1, & n_1 \equiv n_2 \equiv 0\pmod{2},\\
        2a, & n_1 \equiv n_2 \equiv 1\pmod{2},\\
        a+2, & S_{\mathbf{n}} \in \mathcal S_2,
    \end{cases}\\
    q_{\mathbf{n},\mu} &=
    \begin{cases}
        a+1, & L_{\mathbf{n},\mu} \in \mathcal L_1 \cup \mathcal L_2 \cup \mathcal L_4,\\
        a, & L_{\mathbf{n},\mu} \in \mathcal L_3.
    \end{cases}
\end{align}

Note that the $d \equiv 2\pmod{4}$ and $d \equiv 3\pmod{4}$ cases involve further splitting the sub-lattices $\mathcal S_1$ and $\mathcal S_2$ based on the exact parity of each coordinate.
The reason for doing this is a slight technicality that enables hitting exactly the optimal encoding rate.
That being said, a simple algebraic calculation leads to, after concatenating with $C_\IN = [d,1,d]$, that all of the above cases hit exactly
\begin{align}
    n =
    \begin{cases}
    \frac{7}{8} d^2 N   & d:\text{even}\,, 
    \\
    \frac{7}{8}d(d+\frac 17) N & d:\text{odd}\,.
    \end{cases}
\end{align}
Since there are exactly $k = 2N$ logical DOFs in the $2$D $\mathbb Z_2$ LGT, this immediately implies that the bound from Theorem~\ref{thm:optimal2D} is saturated.

\section{Performance Comparison for Gauss's Law Codes}\label{sec:compare}

In this section, we provide in-depth comparisons for the codes constructed in this paper.
This includes comparing the performance of Gauss's law codes as a function of code distance, as well as comparing Gauss's law codes with the basic alternatives.
To keep the discussion concise, we consider only the \textit{optimal} Gauss's law codes derived in~\cref{sec:gauss}, i.e., codes that saturate our optimality bounds for encoding rate from~\cref{thm:optimal1D,thm:optimal2D}.

\Cref{sec:parameter-comparison} showcases a high-level comparison of code properties such as the number of physical qubits used and the locality of the encoded lattice Hamiltonian for Gauss's law codes against baselines.
\Cref{sec:LER-comparison} performs code capacity demonstrations with two selected noise models: the depolarizing channel (simulated numerically) and the uncorrelated $X/Z$ noise model (analytically tractable).

\subsection{Code Properties}\label{sec:parameter-comparison}

In this section, we discuss and compare the basic properties of Gauss's law codes with distance $d \ge 3$ on an $N$-site lattice (in both $1$D and $2$D).
We will restrict to odd $d = 2t+1$, which is necessary to enable correction of arbitrary $t$-qubit errors.
A useful comparison to an alternative code is obtained by encoding all sites and links with an $\llbracket n, 1, d=2t+1\rrbracket$ QECC, which can correct both $X$ and $Z$ errors.
For a comparison, we consider the minimal $n$ for a given $d=2t+1$, which is known to satisfy Rains' bound~\cite{Rains:1996yc}.
For $k=1$, the bound is
\begin{equation}
    t \leq \left\lfloor\frac{n+1}{6}\right\rfloor\,.
\end{equation}
When $t=1,2,3$, QECCs saturating this bound are exactly known~\cite{laflamme1996perfect, Calderbank:1996aj, dutta2010quantumcycliccode, dutta2011quantum}.
We refer to these codes as \emph{minimal-length codes} and denote the number of physical qubits as 
\begin{align}
    n_{\text{min}}=6t-1\,.
\end{align}

In addition to counting the total number of physical qubits used by each code, it is also worth comparing the locality of the encoded lattice Hamiltonian, since this is an indication of the complexity of performing quantum simulations using the encoding.
In $1$D, we will consider both the hard-core boson Hamiltonian \eqref{eq:Hstandard} and the staggered fermion case \eqref{eq:HfermionJW} for illustrative purposes; in $2$D, we will only consider the hard-core boson case to keep the discussion concise.
The encoded Hamiltonian(s) are obtained simply by replacing each Pauli operator $X$, $Y$, and $Z$ in the original Hamiltonian(s) \eqref{eq:Hstandard} (and \eqref{eq:HfermionJW} as applicable) with a logical representative in the chosen code.
The hard-core boson case \eqref{eq:Hstandard} consists of single-body $Z_S$ and $Z_L$ terms on every site and link, respectively; three-body $X_SX_LX_S$ terms on every site-link-site triple; and four-body $X_LX_LX_LX_L$ terms on every plaquette (not applicable in $1$D).
The staggered fermion case \eqref{eq:HfermionJW} contains additional $Y_SX_LY_S$ terms on every site-link-site triple.
For a fair comparison, we choose the logical representatives to have the least possible weight in order to reduce the overhead to perform quantum simulation.
We leave as a future work to evaluate the effect of weight and locality on the cost under concrete fault-tolerant quantum computing architectures.

\subsubsection{1D Comparison}

In $1$D, the optimal Gauss's law code at distance $d = 2t+1 \ge 3$ for an $N$-site lattice uses $(2t+1)(t+1)N$ physical qubits~(see~\cref{thm:optimal1D}).
Encoding the same lattice with the minimal-length codes uses $2Nn_{\text{min}}$ physical qubits.
\Cref{tab:qubit-comparison-1D} shows the comparison between these values using Rains' bound for $t = 1,2,3$ ($d=3,5,7$).
One can see that the number of physical qubits used by the optimal Gauss's law code is less than that of the minimal-length encoding for $d\le 7$.
However, the former grows like $\mathcal{O}(d^2)$ while the lower bound $n_{\text{min}}$ grows like $\mathcal{O}(d)$, which suggests that at higher code distances, the optimal Gauss's law code will eventually require more resources than the minimal-length code, if similar bounds turn out to be saturated asymptotically.
Indeed, the Gauss's-law code uses fewer physical qubits up to $d=9$, while the minimal-length code that saturates the bound (if it exists) uses fewer physical qubits for $d\geq 11$.
\begin{table*}[htp]
    \centering
    \caption{Physical qubits used by the optimal Gauss's law code (GLC) and the minimal-length code to achieve distance $d = 2t+1$ on an $N$-site lattice in $1$D. It does not depend on whether we choose the hard-core boson theory \eqref{eq:Hstandard} or the staggered fermion theory \eqref{eq:HfermionJW}.}
    \label{tab:qubit-comparison-1D}
    \begin{tabular}{c||c@{\hspace{0.2em}}c|c@{\hspace{0.2em}}c}
         & \multicolumn{2}{c|}{Optimal GLC}
         & \multicolumn{2}{c}{Minimal-length}
         \\
         \hline \hline
         $d=3$
         & $6N$ & (RRW~\cite{Rajput:2021trn})
         & $10N$ & ($\llbracket5,1,3\rrbracket$~\cite{laflamme1996perfect})
         \\
         $d=5$
         & $15N$ & (ours)
         & $22N$ & ($\llbracket11,1,5\rrbracket$~\cite{Calderbank:1996aj})
         \\
         $d=7$
         & $28N$ & (ours)
         & $34N$ & ($\llbracket17,1,7\rrbracket$~\cite{dutta2010quantumcycliccode, dutta2011quantum})
         \\
         \hline
         general
         & $(2t+1)(t+1)N$ & (ours)
         & $2(6t-1)N$ & (bound~\cite{Rains:1996yc})
    \end{tabular}
\end{table*}

Now we come to the Hamiltonian locality.
We first study the logical encoding in the Gauss's law (GL) code.
The first (electric) and second (mass) terms in both~\cref{eq:Hstandard,eq:HfermionJW} can be encoded by using logical operators $\bar{Z}^{(\text{GL})}_{L_n},\bar{Z}^{(\text{GL})}_{S_n}$ obtained by applying $Z$ to every qubit on the inner phase-flip (PF) code in the corresponding block,
\begin{equation}
    \bar{Z}^{(\text{GL})}_{L_n} \equiv \bar{Z}^{(\text{PF})}_{L^{(1)}_n}
    \,,\quad
    \bar{Z}^{(\text{GL})}_{S_n} \equiv \bar{Z}^{(\text{PF})}_{S^{(1)}_n}
    \,,
\end{equation}
which requires weight $d = 2t+1$ since $\bar{Z}^{(\text{PF})}\sim Z^{\otimes d}$ at the physical level.
Note that in the case of using duplicated lattice variables for the outer bit-flip code, one only needs to apply a logical phase-flip for the inner block representing one of the duplicates, since this will already take the operator outside of the stabilizer group.

Regarding the third (hopping) term in \cref{eq:Hstandard}, the block $[X_SX_LX_S]$ can be considered as a logical operator in the outer Gauss's law code, but a single $X_S, X_L$ cannot.
The logical encoding of the inner block on any site or link is given as
\begin{align}
&\qquad[\overline{X_{S_n}X_{L_{n}}X_{S_{n+1}}}]^{(\text{GL})}
    \\&=
    \left(\prod_{i=1}^{t}\bar{X}^{(\text{PF})}_{S^{(i)}_n} \right)
    \bar{X}^{(\text{PF})}_{L_n}
    \left(\prod_{j=1}^{t}\bar{X}^{(\text{PF})}_{S^{(j)}_{n+1}}\right)\,.
\end{align}
Since each $\bar{X}^{(\text{PF})}$ can be taken as an $X$ operator acting on one physical qubit, the overall weight of this operator is given by $t+1+t=2t+1$.
Therefore, all terms in the $1$D Hamiltonian \eqref{eq:Hstandard} require weight exactly $2t+1$, matching the code distance.

On the other hand, the $[Y_SX_LY_S]$ term in the staggered fermion case \eqref{eq:HfermionJW} can be implemented as
\begin{align}
&\qquad
\bar{Z}^{(\text{GL})}_{S_n}
[\overline{X_{S_n}X_{L_{n}}X_{S_{n+1}}}]^{(\text{GL})}
\bar{Z}^{(\text{GL})}_{S_{n+1}}
    \\&=
    \bar{Z}^{(\text{PF})}_{S^{(1)}_n}
    \left(\prod_{i=1}^{t}\bar{X}^{(\text{PF})}_{S^{(i)}_n} \right)
    \bar{X}^{(\text{PF})}_{L_n}
    \left(\prod_{j=1}^{t}\bar{X}^{(\text{PF})}_{S^{(j)}_{n+1}}\right)
    \bar{Z}^{(\text{PF})}_{S^{(1)}_{n+1}}
    \\
    &=\left(\bar{Z}^{(\text{PF})}_{S^{(1)}_n}\bar{X}^{(\text{PF})}_{S^{(1)}_n}\right)
    \left(\prod_{i=2}^{t}\bar{X}^{(\text{PF})}_{S^{(i)}_n} \right)
    \bar{X}^{(\text{PF})}_{L_n}
    \\ \notag
    &\qquad\times \left(\prod_{j=2}^{t}\bar{X}^{(\text{PF})}_{S^{(j)}_{n+1}}\right)
    \left(\bar{Z}^{(\text{PF})}_{S^{(1)}_{n+1}}\bar{X}^{(\text{PF})}_{S^{(1)}_{n+1}}\right)\,.
\end{align}
The first and last $Z,X$ require weight $d$ physical operations, while the remaining $X$'s need one physical operation. In total, one needs $d+(t-1)+1+(t-1)+d=6t+1$, which is maximum among all terms in \cref{eq:HfermionJW}.

For the minimal-length code, we assume that the minimum-weight logical operators for $X$, $Y$, and $Z$ all have distance $d = 2t+1$.
(Since $\text{wt}(\bar{X}),\text{wt}(\bar{Y}),\text{wt}(\bar{Z})\geq d = \min\{\text{wt}(\bar{X}),\text{wt}(\bar{Y}),\text{wt}(\bar{Z})\}$, this assumption gives the best estimate for the minimal-length codes.)
Therefore, the maximum weight of any term in the Hamiltonian is (greater than or equal to) $3d = 6t+3$.

\Cref{tab:locality-comparison-1D} summarizes the Hamiltonian locality.
It is useful to note that $2t + 1 < 6t + 3$ and $6t + 1 < 6t + 3$ always, and therefore the optimal Gauss's law code always leads to a more local Hamiltonian than the minimal-length encoding, at arbitrary distance.
The hard-core boson case provides exactly an improvement by a factor of $3$, while the staggered fermion case has a more modest improvement.

\begin{table}[htp]
    \centering
    \caption{Hamiltonian locality using the optimal Gauss's law code (GLC) and the minimal-length code at distance $d = 2t+1$ on a $1$D lattice for the Hamiltonians \eqref{eq:Hstandard} and \eqref{eq:HfermionJW}.}
    \label{tab:locality-comparison-1D}
    \begin{tabular}{c||c|c|c}
         & GLC \eqref{eq:Hstandard} & GLC \eqref{eq:HfermionJW} & Min-length
         \\ \hline \hline
         $d=3$ & $3$ & $7$ & $9$
         \\
         $d=5$ & $5$ & $13$ & $15$
         \\
         $d=7$ & $7$ & $19$ & $21$
         \\
         \hline
         $d \ge 7$ (general) & $2t+1$ & $6t+1$ & $6t+3$
    \end{tabular}
\end{table}

\subsubsection{2D Comparison}

In $2$D, the optimal Gauss's law code at distance $d = 2t+1$ for an $N$-site lattice uses $\frac14 (7t+4)(2t+1)N$ physical qubits for $d > 3$, and $9N$ physical qubits for $d = 3$.
Encoding the same lattice with the minimal-length code uses $3Nn_{\text{min}}$ physical qubits.
\Cref{tab:qubit-comparison-2D} shows the comparison between these values using Rains bound for $t = 1,2,3$ ($d=3,5,7$).
As in the $1$D case, the number of physical qubits used by the optimal Gauss's law code is less than that of the minimal-length encoding for $d \le 7$.
But again, the former grows as $\mathcal{O}(d^2)$ while the latter grows as $\mathcal{O}(d)$, suggesting that the optimal Gauss's law code will eventually require more resources than the optimal perfect code.
The crossover happens between $d=7$ and $d=9$, if a bound-saturating code exists.

\begin{table*}[htp]
    \centering
    \caption{Physical qubits used by the optimal Gauss's law code and the minimal-length code to achieve distance $d = 2t+1$ on an $N$-site lattice in $2$D.}
    \label{tab:qubit-comparison-2D}
    \begin{tabular}{c||c@{\hspace{0.2em}}c|c@{\hspace{0.2em}}c}
         & \multicolumn{2}{c|}{Optimal GLC}
         & \multicolumn{2}{c}{Minimal-length}
         \\
         \hline \hline
         $d=3$
         & $9N$ & (RRW~\cite{Rajput:2021trn})
         & $15N$ & ($\llbracket5,1,3\rrbracket$~\cite{laflamme1996perfect})
         \\
         $d=5$
         & $22.5N$ & (ours)
         & $33N$ & ($\llbracket11,1,5\rrbracket$~\cite{Calderbank:1996aj})
         \\
         $d=7$
         & $43.75N$ & (ours)
         & $51N$ & ($\llbracket17,1,7\rrbracket$~\cite{dutta2010quantumcycliccode, dutta2011quantum})
         \\
         \hline
         general
         & $\frac14(7t+4)(2t+1)N$ & (ours)
         & $3(6t-1)N$ & (bound~\cite{Rains:1996yc})
    \end{tabular}
\end{table*}

For the Hamiltonian locality comparison in $2$D, we can parallel the discussion from before.
We will need to use the optimal multiplicities $p_{\mathbf{n}}$ and $q_{\mathbf{n},\mu}$ provided in~\cref{sec:optimal} for the $2$D case.

By similar reasoning to the $1$D case, the electric energy term, mass term, and hopping term from \eqref{eq:Hstandard} all have weight exactly $2t+1$ after encoding.
It turns out that the $X_LX_LX_LX_L$ plaquette terms for $t>1$ \textit{also} require weight exactly $2t+1$, matching the code distance. For $t=1$, this gives weight $4$.
This can be shown by casework on the remainder of $d$ modulo $4$, using the exact expressions for $p_{\mathbf{n}}$ and $q_{\mathbf{n},\mu}$ from~\cref{sec:optimal}.
The key point is that the optimization procedure used to minimize the logical encoding rate for the Gauss's law code simultaneously makes the weight of all terms in the Hamiltonian \eqref{eq:Hstandard} essentially identical.

For the minimal-length code, the plaquette operator $X_LX_LX_LX_L$ clearly has the highest weight, requiring an interaction of $4d$ qubits at once.
The weight is then~$4d = 8t+4$.
Therefore, the improvement in locality by the Gauss's law code is a factor of $4$ ($d\geq5$) in the $2$D case.
This is summarized in \Cref{tab:locality-comparison-2D}.

\begin{table}[htp]
    \centering
    \caption{Hamiltonian locality for the hard-core boson case \eqref{eq:Hstandard} using the optimal Gauss's law code (GLC) and the minimal-length code at distance $d = 2t+1$ on a $2$D lattice. Note that the $d = 3$ case does not follow the general rule.}
    \label{tab:locality-comparison-2D}
    \begin{tabular}{c||c|c}
         & GLC \eqref{eq:Hstandard} & Min-length
         \\ \hline \hline
         $d=3$ & $4$ & $12$
         \\
         $d=5$ & $5$ & $20$
         \\
         $d=7$ & $7$ & $28$
         \\
         \hline
         $d \ge 7$ (general) & $2t+1$ & $8t+4$
    \end{tabular}
\end{table}

\subsection{Code Capacity and Logical Error Rate}\label{sec:LER-comparison}

In this section, code capacity demonstrations are performed by exposing code states to selected noise models, followed by attempting recovery via the usual stabilizer error correction procedure, equipped with a chosen syndrome decoder.
A code capacity study assumes that errors only happen before syndrome extraction and only on data qubits, and all other parts (circuits and measurements for syndromes, initial state preparation, ancilla qubits) are done perfectly.
Although such a code capacity study is less realistic than full \emph{circuit-level} simulation, where all circuits undergo errors, it can assess how much logical information can be recovered from ideal syndrome outcomes, as was done in e.g.~\cite{Dennis:2001nw}.
Many technical details regarding the syndrome decoding strategy and analytical calculations are given in~\cref{app:syndrome-decoding,app:exact-probability}.
The following features are in common between all of our demonstrations:

\begin{itemize}
    \item Spatial dimensionality of underlying lattice: $1$D.
    \item Noise model characteristic: a single-qubit error channel applied independently with identical distribution (i.i.d.) to every physical qubit.
    \item Syndrome decoder used for Gauss's law codes: maximum-likelihood decoder (MLD) (along with posterior evaluation using concatenated structure~\cite{Poulin:2006lth}); see~\cref{app:syndrome-decoding} for details.
    \item Syndrome decoder used for minimal-length codes: maximum-likelihood decoder (MLD).
\end{itemize}

A convenient performance metric for code capacity demonstrations is the logical error rate (LER), which measures the likelihood that the initial noise is disruptive enough to cause a logical failure even under the assumption of perfectly executed recovery.
Additional details are provided under the heading for each demonstration.

\subsubsection{Depolarizing Noise Demonstration}\label{sec:numerical}

For our first demonstration, consider the $1$D spatial lattice with $N$ sites, and equip the optimal Gauss's law code $C_\OUT \circ C_\IN$ for odd distance $d \ge 3$.
In addition to the Gauss's law code, consider also the minimal-length code with the same distance $d$ and on the same $N$-site lattice.
In the numerical simulation, we impose open boundary conditions for both boundary links (setting $L_0=L_{N}=1$), resulting in an open lattice with $(N-1)$ links.

We use depolarizing noise, which acts on the $i^{\mathrm{th}}$ qubit by
\begin{equation}\label{eq:depolarizing}
    \mathcal E_i(\rho) = (1-p)\rho + \frac{p}{3}\left(X_i\rho X_i + Y_i\rho Y_i + Z_i\rho Z_i\right),
\end{equation}
where $\rho$ denotes the state of the entire $n$-qubit physical system, and $0\le p\le 1$ is the physical error rate~\footnote{It is different from $p_{\mathbf{n}}$, which is the duplication parameter.}.
Applying this as i.i.d. noise to every qubit results in the error channel $\mathcal E = \mathcal E_1 \circ \mathcal E_2 \circ \cdots \circ \mathcal E_n$.

The recovery operation $\mathcal R$ depends on whether the code used is the Gauss's law code or the minimal-length code.
In both cases, begin by measuring the full stabilizer syndrome.
For the Gauss's law code, decode the syndrome with MLD (see~\cref{app:syndrome-decoding}), where the posterior probability is evaluated making use of the concatenated structure, thereby obtaining a candidate Pauli error pattern.
For the minimal-length code, decode the syndrome with MLD (see~\cref{app:syndrome-decoding}).
In both cases, to conclude the recovery operation, apply the decoded pattern directly to the physical system.

As a metric, we define the logical error rate by using the individual logical operators. 
Specifically, for Gauss's law code (GL), we define the logical operators as
\begin{align}
    \bar{Z}^{(\text{GL})}_{n} &:= \bar{Z}^{(\text{PF})}_{L_n}\,,
    \\
    \bar{X}^{(\text{GL})}_{n} &:= \bar{X}^{(\text{PF})}_{S_{n}}\bar{X}^{(\text{PF})}_{L_{n}}\bar{X}^{(\text{PF})}_{S_{n+1}}\,,\label{eq:glc-logical-x}
\end{align}
where $\bar{X}^{(\text{PF})},\bar{Z}^{(\text{PF})}$ represent the logical operators in the encoding by the inner phase-flip code.
For the minimal-length codes, the logical operators are defined by
\begin{align}
 \bar{Z}^{(\text{min})}_{2m-1} &:=\bar{Z}^{(\text{min})}_{S_m}
 \,,&
 \bar{Z}^{(\text{min})}_{2m} &:=\bar{Z}^{(\text{min})}_{L_m}\,,
 \\
 \bar{X}^{(\text{min})}_{2m-1} &:=\bar{X}^{(\text{min})}_{S_m}
 \,,&
 \bar{X}^{(\text{min})}_{2m} &:=\bar{X}^{(\text{min})}_{L_m}\,.
\end{align}

Let $x^{(\texttt{code})},z^{(\texttt{code})}\in\mathbb F_2^{k^{(\texttt{code})}}$ denote the binary vectors specifying the residual logical Pauli error after recovery. That is to say,  the residual error is written as
\begin{equation}
\prod_{n=1}^{k^{(\texttt{code})}}
\left(\bar X_n^{(\texttt{code})}\right)^{x_n^{(\texttt{code})}}
\left(\bar Z_n^{(\texttt{code})}\right)^{z_n^{(\texttt{code})}}\,,
\end{equation}
up to a stabilizer and an overall phase.
Then two types of LER can be defined by~\footnote{The different metrics (LERs) are used for the respective codes, since the logical dimensions are different. 
We can instead evaluate $\LER^{(\texttt{GL})}$ for both codes, but that for the minimal-length code does not measure failure of its full logical $X/Z$ eigenvalues. On the other hand, imposing Gauss's law to match the logical dimension would increase the relevant $X$-distance. }
\begin{align}
    \LER^{(\texttt{code})}_Z & = \Pr\left[ \left(x^{(\texttt{code})}_{1},\cdots, x^{(\texttt{code})}_{k^{(\texttt{code})}}\right) \neq \bm{0}\right]\,,
    \\
    \LER^{(\texttt{code})}_X & = \Pr\left[ \left(z^{(\texttt{code})}_{1},\cdots, z^{(\texttt{code})}_{k^{(\texttt{code})}}\right) \neq \bm{0}\right]\,,
\end{align}
where $k^{(\texttt{code})}$ represents the number of logical qubits in each code, given as
\begin{align}
    k^{(\text{GL})} = N-1\,,\quad 
    k^{(\text{min})} = 2N-1\,,
\end{align}
under the assumption that the lattice has $N$ sites and $(N-1)$ links.
The value of $\LER^{(\texttt{code})}_{Z}$ ($\LER^{(\texttt{code})}_{X}$) evaluates the ability of this setup to protect a simultaneous eigenstate of the logical $\bar{Z}$ ($\bar{X}$) operators against errors.
\Cref{fig:memory-experiment-depolarizing} shows $\LER^{(\texttt{code})}_Z$ (a) and $\LER^{(\texttt{code})}_X$ (b) against the physical error rate $p$.
The former (latter) setting is called the $Z$ memory ($X$ memory) simulation.
We see that, for both code families, LER decreases with increasing $d$, for the value of $p$ below the points where several $d$ lines cross each other (called the \emph{threshold}).
For the $Z$-memory simulation, the Gauss's-law code has a lower LER than the minimal-length code for $d=3$, but a higher LER for $d=7$. In contrast, in the $X$-memory experiment, the Gauss's-law code has a lower LER than the minimal-length code for all three distances, $d=3,5,7$. Furthermore, the crossing point for the Gauss's-law code occurs at a higher value of $p$ than that for the minimal-length code in the $X$-memory experiment, whereas the opposite ordering is observed in the $Z$-memory experiment.
Evaluating threshold values in the \textit{circuit-level} simulation is important to quantify the performance in a more practical setting, which we leave as a future direction.

For the $Z$ memory experiment, the single-site flip (product over $i$ is only taken for Gauss's law code)
\begin{equation}
    \ket{0}_{L_n}\prod_{i}(\ket{0}_{S^{(i)}_{n+1}})\ket{0}_{L_{n+1}}
    \to 
    \ket{0}_{L_n}\prod_{i}(\ket{1}_{S^{(i)}_{n+1}})\ket{0}_{L_{n+1}} \,,
\end{equation}
can be detected and corrected for Gauss's law code, while it gives a logical error for minimal-length codes.
Gauss's law codes can correct such gauge-violating errors, which are \textit{also} ideally avoided in LGT simulations.
However, when the site flip is accompanied by a flip of either neighboring link, the resulting error can cause a logical failure of the Gauss's-law code. The simulation results are consistent with the interpretation that, as $d$ increases, the multiplicity of such failure patterns comes to outweigh the advantage of correcting a simultaneous flip of all duplicated copies of one site variable (which is not gauge invariant), leading to the reversal in the relative performance of the two codes.
Besides, the asymmetry in $X$ and $Z$ memory behavior that is observed only for the Gauss's law code may come from the concatenation structure (phase-flip and bit-flip errors are treated separately in inner and outer codes) and the weight of logical operators~\cref{eq:glc-logical-x}.

\begin{figure*}
  \centering
  \includegraphics[width=0.9\linewidth]{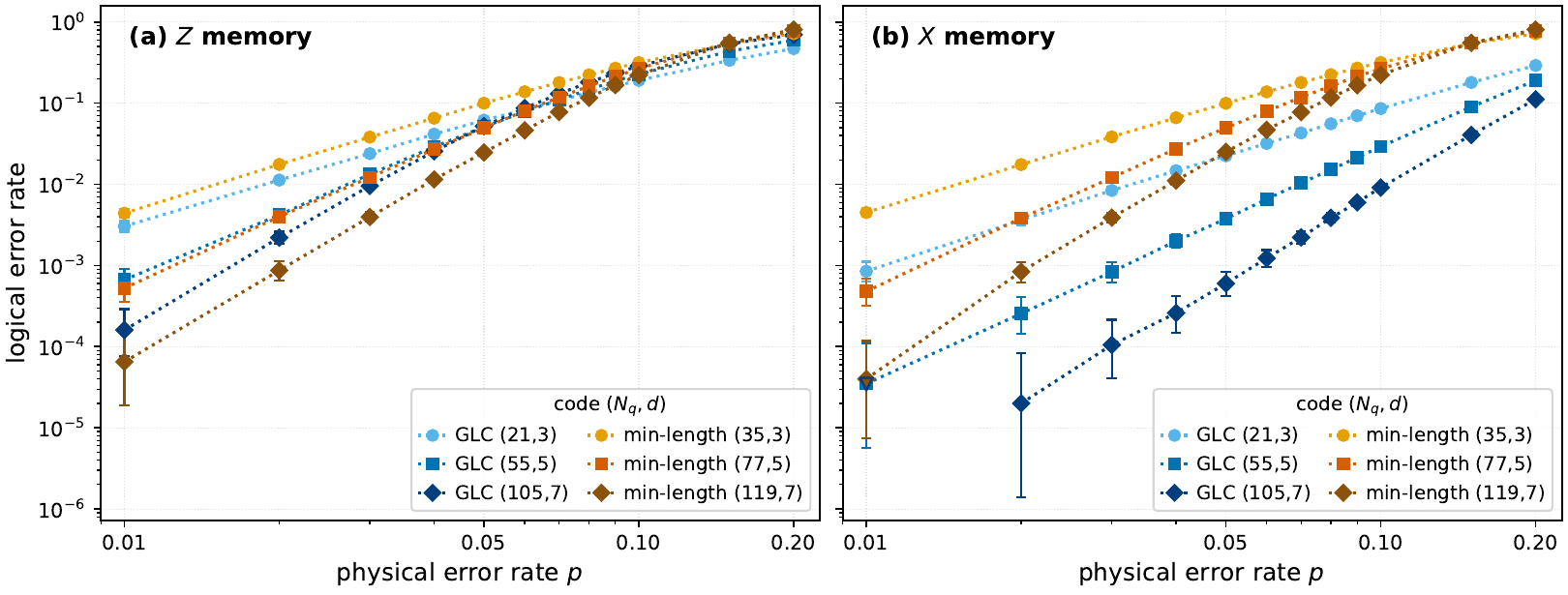}
  \caption{
  Memory experiment for a 1D $N=4$ lattice with the open boundary condition. The $x$-axis shows the physical error rate, and the $y$-axis is the logical error rate. 
  Blue and red curves correspond to the Gauss’s-law code (GLC) and the minimal-length code (min-length), respectively. Each legend entry is labeled by $(N_q,d)$, where $N_q$ is the total number of physical qubits and $d$ is the code distance. The left (right) panel focuses on the $Z$ ($X$) memory experiment, if they can keep the $Z$ ($X$) eigenstates. Monte-Carlo simulation is performed at each data point with $N_{\text{shot}} = 2\times 10^5$ shots, and LER is estimated as $N_{\text{fail}}/N_{\text{shot}}$. The error bar shows binomial likelihood-factor intervals with a maximum likelihood factor $10^3$.
  }
  \label{fig:memory-experiment-depolarizing}
\end{figure*}

\subsubsection{Uncorrelated X/Z Noise Demonstration}\label{sec:analytical}

For our second demonstration, consider again the $1$D spatial lattice with $N$ sites, and equip the optimal Gauss's law code $C_\OUT \circ C_\IN$ for odd distance $d \ge 3$.

Define uncorrelated $X/Z$ noise on the $i^{\mathrm{th}}$ qubit as the error channel
\begin{align}\label{eq:uncorrelatedX/Z}
    \mathcal E_i(\rho) = (1&-p)(1-p)\rho + p(1-p)X_i\rho X_i\\
    & + p^2Y_i\rho Y_i + p(1-p)Z_i\rho Z_i,
\end{align}
where $\rho$ denotes the state of the entire $n$-qubit physical system, and $0 < p < 1$ is used as the independent rate for physical $X$-errors and $Z$-errors.\footnote{Technically, we will need to assume $p < 1/2$ in order for the results to be exactly correct. This technicality is discussed in~\cref{app:syndrome-decoding}.}
When acting i.i.d. on every qubit, this results in the error channel $\mathcal E = \mathcal E_1 \circ \mathcal E_2 \circ \cdots \circ \mathcal E_n$.

For the recovery operation $\mathcal R$, measure all inner and outer stabilizer generators, and decode the resulting syndrome with MLD, thereby obtaining a candidate Pauli error pattern.
The explicit decoding method is discussed in~\cref{app:syndrome-decoding}; importantly, many details are simplified due to the fact that uncorrelated $X/Z$ noise is used.
To conclude the recovery operation, apply the decoded pattern directly onto the physical system.

\begin{figure}[htp]
    \centering
    \includegraphics[width=\linewidth]{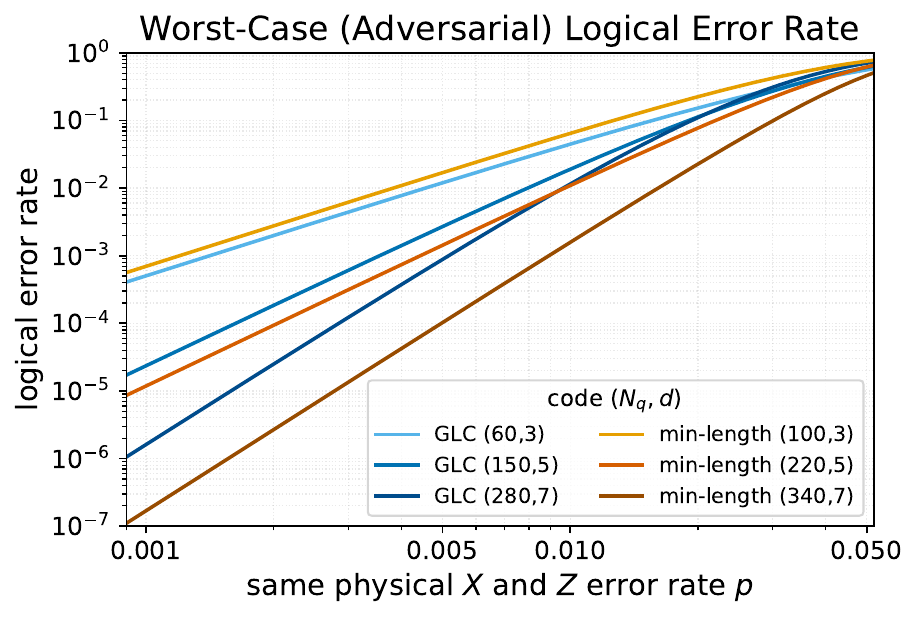}
    \caption{Logical error rates on a $1$D periodic lattice with $N_s = 10$ sites, on the worst-case (adversarial) initial code state. The uncorrelated $X/Z$ channel is used with the same physical $X$ and $Z$ error rates. Blue and red curves correspond to the Gauss's law code (GLC) and the minimal-length code (min-length), respectively. Each legend entry is labeled by $(N_q, d)$ where $N_q$ is the total number of physical qubits and $d$ is the code distance. These curves are obtained from the syndrome decoding and analytical calculations that are detailed in~\cref{app:syndrome-decoding,app:exact-probability}.
    }\label{fig:uncorrelated-XZ}
\end{figure}

In~\cref{fig:uncorrelated-XZ}, we plot the (strict) \textit{worst-case} LER, defined as
\begin{equation}\label{eq:strict-LER}
    \LER \equiv \max_{\ket{\psi}\in\mathcal H_C}\Pr\left[|\langle\psi_r|\psi\rangle|^2\neq 1\right],
\end{equation}
where $\ket{\psi_r} \equiv \mathcal R(\mathcal E(\ket{\psi}))$ denotes the (random) state obtained after recovery, and $\mathcal H_C$ denotes the code space of the Gauss's law code.
The LER defined by \eqref{eq:strict-LER} maximizes over all possible initial states to find the worst-case probability that the initial state is not recovered exactly.
For Gauss's law codes, which are handled by MLD, this is identical to the LER defined by
\begin{equation}
    \LER^{(\texttt{GL})} = \Pr\left[ x^{(\texttt{GL})} \neq \bm{0} \lor z^{(\texttt{GL})} \neq \bm{0}\right]\,,
\end{equation}
where we use the same notation as in~\cref{sec:numerical}.

As it turns out, there is a very high degree of analytical control over this question, and we can compute \eqref{eq:strict-LER} exactly.
This analytical calculation is carried out in~\cref{app:exact-probability}, and~\cref{fig:uncorrelated-XZ} contains the results.
One can see a similar trend as in the depolarizing noise simulation.
The logical error rate of the Gauss's law code is lower than that of the minimal-length code for $d=3$, but this ordering is reversed for $d=5$ and $d=7$.

\section{Summary and Conclusion}\label{sec:summary}
In this paper, we considered a quantum error-correction protocol incorporating Gauss’s law constraint. 
It allowed us to correct violations of Gauss’s law caused by hardware noise within a unified QEC framework. 
Specifically, we extended the code construction with code distance $d=3$~\cite{Rajput:2021trn} to an arbitrary distance. 
The sites and links are duplicated, and the classical Gauss’s law code is built from Gauss’s law and parity checks between pairs of duplicated links/sites, which is further concatenated with a phase-flip repetition code to be a full QEC code. 
We have shown that our construction with a specific duplication pattern achieves an arbitrary distance and has the optimal encoding rate among the considered constructions.

To quantify the performance of the resulting codes,
the number of physical qubits was compared with baseline minimal-length code constructions, and at low code distances ($d\leq 7)$, Gauss’s law code is smaller than that of the minimal-length code, generalizing similar findings for $d=3$ in RRW~\cite{Rajput:2021trn}.
We also performed the memory simulation in the code capacity setting for the depolarizing noise, and observed that the Gauss’s-law code exhibits lower $X$-memory logical error rates than the minimal-length code at all distances considered, whereas the $Z$-memory comparison depends on the distance. The exact analysis under uncorrelated $X/Z$ noise exhibits the same qualitative distance dependence as the depolarizing-noise $Z$-memory comparison.

There are many important future directions. First, it would be important to consider the fault-tolerant syndrome extraction circuit. 
It can be implemented using general ancilla-based or flag-based constructions~\cite{Shor:1996qc,DiVincenzo:1996xb, Steane:1996ic,Chao:2017two,Chamberland:2017flag} or a logical-to-physical CNOT gate as proposed in RRW~\cite{Rajput:2021trn}. 
Evaluating their performance using circuit-level simulation will allow us to compare the performance of Gauss’s law code with that of other possible codes, in particular in terms of the required error threshold. 
How to perform gate operations in a fault-tolerant way is another important direction toward fault-tolerant LGT simulation, as was discussed in~\cite{Spagnoli:2024mib} for the $d=3$ case. 
Finally, it would be interesting to extend Gauss’s law code construction to non-Abelian gauge theory.

\begin{acknowledgments}
This work was supported by the US DOE and other Agencies.
N.~S.~M. and C.~W.~B. were supported by the US Department of Energy, Office of Science, National Quantum Information Science Research Centers, Quantum Systems Accelerator (Award No. DE-SCL0000121).
C.~W.~B. also acknowledges support from the U.S. Department of Energy (DOE), Office of Science under contract DE-AC02-05CH11231, partially through Quantum Information Science Enabled Discovery (QuantISED) for High Energy Physics (KA2401032).
L.~N. is supported by JST PRESTO Grant Number JPMJPR25F5 and JST Grant Number JPMJPF2221.
M.~H.~is supported by JST CREST Grant Number JPMJCR24I3, JSPS KAKENHI Grant Number JP22H01222 and the Royal Society grants ICA/R2/242058 and IEC/R3/243103.
N.Y. is supported by JST Grant Number JPMJPF2221, JST CREST Grant Number JPMJCR23I4, IBM Quantum, Google Quantum AI, JST ASPIRE Grant Number JPMJAP2316, JST ERATO Grant Number JPMJER2302, JST [Moonshot R\&D] [Grant Number JPMJMS256J], and Institute of AI and Beyond of the University of Tokyo.
The authors acknowledge usage of Codex and ChatGPT for assistance with developing numerical simulations, as well as mathematical proofs in the appendix, which were subsequently human-checked.
\end{acknowledgments}

\bibliography{refs}

\appendix
\section{Proofs for Optimality Theorems}\label[appendix]{app:proofs}

Here, we provide the proofs of~\cref{thm:optimal1D,thm:optimal2D}.

\begin{proof}[Proof of~\cref{thm:optimal1D}]
    It suffices to prove the bound
    \begin{equation}
        \frac{k}{n} \le \frac{2}{d(d+1)},
    \end{equation}
    since it's easy to check that the code defined explicitly by the multiplicities \eqref{eq:optimal1Dpq} saturates the bound.

    Write $C_\OUT \circ C_\IN$ for the Gauss's law code, where $C_\OUT = [n_\OUT,k_\OUT,d_\OUT]$ is the classical Gauss's law code and $C_\IN = [n_\IN,1,d_\IN]$ is an arbitrary classical code encoding a single logical bit.
    
    Recall that the number of logical qubits is always $k = k_\OUT = N$ for any $1$D Gauss's law code, and therefore it suffices to find a tight lower bound on $n = n_\OUT n_\IN$.
    Clearly, $n_\IN \ge d_\IN \ge d$, with equality when $n_\IN = d_\IN = d$, e.g., by taking $C_\IN = [d,1,d]$ to be the classical repetition code.

    It remains to find the tight lower bound for $n_\OUT$.
    To this end, suppose $p_1,\dots,p_N \in \mathbb N$ and $q_1,\dots,q_N \in \mathbb N$ are the multiplicities defining the classical Gauss's law code $C_\OUT$.
    Define
    \begin{align}
        P &\equiv \sum_{i=1}^{N} p_i,\\
        Q &\equiv \sum_{i=1}^{N} q_i,
    \end{align}
    for brevity.

    The classical distance satisfies $d_\OUT \ge d$, which implies that every open string has weight at least $d$.
    In particular, for an open string between sites $j$ and $j+1$, we have
    \begin{equation}
        p_j + q_j + p_{j+1} \ge d_\OUT \ge d.
    \end{equation}
    Summing over $j$ and using periodicity yields
    \begin{equation}
        2P + Q \ge dN.
    \end{equation}
    Moreover, since every $q_j$ is a positive integer, we have $Q \ge N$.
    Therefore,
    \begin{equation}
        n_\OUT = P+Q = \frac{(2P+Q)+Q}{2} \ge \frac{dN+N}{2} = \frac{d+1}{2}N,
    \end{equation}
    which must be satisfied for every $d$, regardless of parity.

    Finally, using the fact that the number of logical qubits is $k = N$ for the concatenated code, we have that the encoding rate satisfies
    \begin{equation}
        \frac{k}{n} = \frac{N}{n_\OUT n_\IN} \le \frac{N}{\left(\frac{d+1}{2}N\right)d} = \frac{2}{d(d+1)},
    \end{equation}
    as desired.
\end{proof}

\begin{proof}[Proof of~\cref{thm:optimal2D}]
    As with the $1$D case, it suffices to prove the upper bound on the encoding rate $k/n$, since it's easy to check that the code defined explicitly by the multiplicities in~\cref{sec:optimal} saturates the bound.
    
    As in the $1$D case, write $C_\OUT \circ C_\IN$ for the Gauss's law code, where $C_\OUT = [n_\OUT,k_\OUT,d_\OUT]$ is the classical Gauss's law code and $C_\IN = [n_\IN,1,d_\IN]$ is an arbitrary classical code encoding a single logical bit.
    
    For the $2$D case, the number of logical qubits is always $k = k_\OUT = 2N$; therefore, again, it suffices to find a tight lower bound on $n = n_\OUT n_\IN$.
    For $C_\IN$, the same reasoning holds as in the $1$D case: $n_\IN \ge d_\IN \ge d$, with equality when $n_\IN = d_\IN = d$, e.g., by taking $C_\IN = [d,1,d]$ to be the classical repetition code.

    The interesting step again reduces to finding the tight lower bound for $n_\OUT$.
    This time, we denote $p_{\mathbf{n}} \in \mathbb N$ and $q_{\mathbf{n},\mu} \in \mathbb N$ as the multiplicities defining the classical Gauss's law code $C_\OUT$.
    Define
    \begin{align}
        P &\equiv \sum_{\mathbf{n}\in\Lambda} p_{\mathbf{n}},\\
        Q &\equiv \sum_{\mathbf{n}\in\Lambda}\sum_{\mu=1}^2 q_{\mathbf{n},\mu}.
    \end{align}

    The classical distance satisfies $d_\OUT \ge d$, which implies that every open string has weight at least $d$, \textit{and} every closed loop has weight at least $d$.
    We proceed by casework on $d$.
    If $d = 3$, then the result is straightforward, because $k = 2N$ and $n_\OUT \ge 3N$, so that
    \begin{equation}
        \frac{k}{n} = \frac{2N}{n_\OUT n_\IN} \le \frac{2N}{3N\cdot3} = \frac{2}{9},
    \end{equation}
    as desired.

    Now suppose $d > 3$, and consider the weight of the open string spanning just a single link $L_{\mathbf{n},\mu}$ and its endpoints $S_{\mathbf{n}}$ and $S_{\mathbf{n}+\hat{\mu}}$.
    This weight is given by
    \begin{equation}
        p_{\mathbf{n}} + q_{\mathbf{n},\mu} + p_{\mathbf{n}+\hat{\mu}} \ge d_\OUT \ge d.
    \end{equation}
    By summing over all $\mathbf{n}\in\Lambda$ and $\mu \in\{1,2\}$, we obtain
    \begin{equation}
        4P + Q \ge 2dN.
    \end{equation}
    Similarly, consider the weight of the plaquette with links $L_{\mathbf{n},x}$, $L_{\mathbf{n}+\hat{x},y}$, $L_{\mathbf{n}+\hat{y},x}$, and $L_{\mathbf{n},y}$.
    This weight is given by
    \begin{equation}
        q_{\mathbf{n},x} + q_{\mathbf{n}+\hat{x},y} + q_{\mathbf{n}+\hat{y},x} + q_{\mathbf{n},y} \ge d_\OUT \ge d.
    \end{equation}
    After summing over all $\mathbf{n} \in \Lambda$, we obtain
    \begin{equation}
        2Q \ge d N.
    \end{equation}
    Combining the inequalities together, we find
    \begin{align}
        n_\OUT &= P+Q = \frac14(4P+Q) + \frac38(2Q)\\
        & \ge \frac14 (2dN) + \frac38 (dN)\\
        & = \frac78 dN.
    \end{align}

    This immediately gives the bound on the encoding rate
    \begin{equation}
        \frac{k}{n} = \frac{2N}{n_\OUT n_\IN} \le \frac{2N}{\left(\frac78 dN\right)d} = \frac{16}{7d^2},
    \end{equation}
    which is the desired bound for even $d$ (but also holds for odd $d$).

    For odd $d$, there is a slight improvement.
    For any $\mathbf{n} \in \Lambda$, consider the sums of the site and link variables around a plaquette with corner at $\mathbf{n}$:
    \begin{align}
        v_{\mathbf{n}} &\equiv p_{\mathbf{n}} + p_{\mathbf{n}+\hat{x}} + p_{\mathbf{n}+\hat{y}} + p_{\mathbf{n}+\hat{x}+\hat{y}},\\
        \ell_{\mathbf{n}} &\equiv q_{\mathbf{n},x} + q_{\mathbf{n},y} + q_{\mathbf{n}+\hat{x},y} + q_{\mathbf{n}+\hat{y},x}.
    \end{align}
    Consider the $4$ open strings that can be formed from any single link and its endpoints around this plaquette.
    The weights of these strings satisfy
    \begin{align}
        p_{\mathbf{n}} + q_{\mathbf{n},x} + p_{\mathbf{n}+\hat{x}} &\ge d,\\
        p_{\mathbf{n}} + q_{\mathbf{n},y} + p_{\mathbf{n}+\hat{y}} &\ge d,\\
        p_{\mathbf{n}+\hat{x}} + q_{\mathbf{n}+\hat{x},y} + p_{\mathbf{n}+\hat{x}+\hat{y}} &\ge d,\\
        p_{\mathbf{n}+\hat{y}} + q_{\mathbf{n}+\hat{y},x} + p_{\mathbf{n}+\hat{x}+\hat{y}} &\ge d.
    \end{align}
    By adding together all these inequalities, we find
    \begin{equation}
        2v_{\mathbf{n}} + \ell_{\mathbf{n}} \ge 4d.
    \end{equation}
    Furthermore, the weight of the plaquette itself satisfies
    \begin{equation}
        \ell_{\mathbf{n}} \ge d.
    \end{equation}
    By adding thrice of the latter inequality to the former inequality, we find
    \begin{equation}
        2v_{\mathbf{n}} + 4\ell_{\mathbf{n}} \ge 7d.
    \end{equation}
    The LHS is always an even integer, but if $d$ is odd then the RHS is an odd integer.
    Therefore, in the case of odd $d$, we have the stronger inequality
    \begin{equation}
        2v_{\mathbf{n}} + 4\ell_{\mathbf{n}} \ge 7d + 1.
    \end{equation}
    By summing over all $\mathbf{n} \in \Lambda$, we find
    \begin{equation}
        8P + 8Q \ge (7d+1)N.
    \end{equation}
    Equivalently,
    \begin{equation}
        n_\OUT \ge \frac{7d+1}{8}N.
    \end{equation}

    This gives a stronger encoding rate bound for odd $d$, namely,
    \begin{equation}
        \frac{k}{n} = \frac{2N}{n_\OUT n_\IN} \le \frac{2N}{\left(\frac{7d+1}{8}N\right)d} = \frac{16}{(7d+1)d},
    \end{equation}
    as desired.
\end{proof}
\section{Syndrome Decoding}\label[appendix]{app:syndrome-decoding}

In this appendix, we discuss the concept of error syndromes and how they can be used to diagnose and correct errors that occur on physical qubits.
We begin with the definitions of error syndromes and error cosets in the classical and quantum setting, as well as stabilizer cosets and the related notion of ``logical equivalence classes," which are features specific to the quantum case, without a classical analogue.
After defining these general concepts, we discuss the most popular types of syndrome \textit{decoders}, which aim to deduce a candidate error pattern from knowledge only of its error syndrome.
Finally, we provide explicit syndrome decoding algorithms that are used to obtain the results in~\cref{sec:LER-comparison} in the main text.

We will assume that all noise models used in the quantum setting are Pauli noise models, since that is the case relevant for this work.

\subsection{Error Syndromes, Cosets, and Logical Equivalence Classes}

For a classical code $C = [n,k,d]$ with parity-check matrix $H$, the \textit{error syndrome} of any bit string $v \in \mathbb F_2^n$ is defined as
\begin{equation}
    \sigma(v) \equiv Hv \in \mathbb F_2^{n-k}.
\end{equation}
The syndrome of $v$ is zero if and only if $v$ is a codeword, and therefore, the syndrome can be used for the purpose of detecting errors.
If a codeword $c \in C$ is transmitted, and the received word $c' \in \mathbb F_2^n$ has nonzero syndrome $s \equiv \sigma(c') \neq 0$, then we can deduce that an error must have occurred during transmission.

Every binary error is a bit-flip, and therefore we can summarize the error on $c$ as a bit string $e \in \mathbb F_2^n$, called the \textit{error pattern}, with $c' = c + e$.
The syndrome function can be applied to the error pattern itself, and the result will always match the syndrome of the received word, $\sigma(e) = \sigma(c') = s$.
For any fixed error syndrome $s \in \mathbb F_2^{n-k}$, there are always exactly $2^k$ distinct patterns $e \in \mathbb F_2^n$ that would lead to the same syndrome, $\sigma(e) = s$.
In classical error correction, this collection is sometimes called the ``error coset" corresponding to the syndrome $s$.
The reason for this terminology is that the entire set of all possible error patterns having the same syndrome $s$,
\begin{equation}
    \Sigma_{s} \equiv \{e \in \mathbb F_2^n \mid \sigma(e) = s\},
\end{equation}
can be generated from a single representative $\widetilde{e} \in \mathbb F_2^n$ with the same syndrome $\sigma(\widetilde{e}) = s$ by arbitrary $C$-translations.
In mathematical terms, $\Sigma_{s} = \tilde{e}+C$, where
\begin{equation}
    \tilde{e}+C
    \equiv \{\widetilde{e} + c\mid c \in C\}.
\end{equation}
This is the easiest way to see that $|\Sigma_{s}| = 2^k$, because clearly $|\tilde{e}+C| = |C| = 2^k$.

At the quantum level, the situation is analogous to the classical case, but there are important changes due to the fact that, in addition to bit-flip errors ($X$ errors), there can also be phase-flip errors ($Z$ errors), or even simultaneous bit-flip and phase-flip errors ($Y$ errors).
For starters, it is not possible to assign an error syndrome directly to an arbitrary quantum state of physical qubits.
Instead, one must \textit{measure} all stabilizer generators to obtain the error syndrome, implying that syndromes are non-deterministic functions of the physical state.

For instance, consider an $\llbracket n,k,d\rrbracket$ stabilizer code with stabilizer group $\mathcal S \subseteq \mathcal P_n$ and code space $\mathcal H_C$.
An initial (pure) code state $\rho \equiv \ket{\psi}\bra{\psi}$ for $\ket{\psi} \in \mathcal H_C$ becomes a (possibly mixed) received state $\rho'$ after errors are incurred.
The error syndrome for $\rho'$ is then given by measuring the stabilizer generators $S_1, \dots, S_{n-k} \in \mathcal S$, yielding respective outcomes $m_1, \dots, m_{n-k} \in \{\pm 1\}$ drawn from the probability distribution induced by $\rho'$.

To sidestep this non-deterministic behavior, it is possible to associate the notion of error syndrome directly to the error pattern, instead of the physical quantum state, analogous to the classical case.
Given that our noise model is assumed to be Pauli, it makes sense to define the syndrome function $\sigma : \mathcal P_n \to \{\pm 1\}^{n-k}$ as follows:
\begin{equation}
    \sigma(E) \equiv (m_1, \dots, m_{n-k}) \iff \forall S_i \in \mathcal S: E S_i = m_i S_i E,
\end{equation}
where the target space $\{\pm 1\}^{n-k}$ is due to the fact that every pair of Pauli operators either commutes or anti-commutes.
This always matches exactly the syndrome that would have been measured on the received state $\rho'$ if the Pauli error $E$ was incurred, i.e., if $\rho' = E \rho E^\dagger$.
This pushes the randomness of the syndrome measurement entirely to the noise model, as in the classical case.\footnote{In other words, we can assume that the syndrome is randomly decided at the moment the error is applied, instead of when the state is measured.}

One can now also define a quantum analogue of error cosets by considering the set of all Pauli operators with a given error syndrome $m \in \{\pm 1\}^{n-k}$,
\begin{equation}
    \Sigma_{m} \equiv \{E \in \mathcal P_n \mid \sigma(E) = m\}.
\end{equation}
Similarly to the classical case, if $\widetilde E \in \mathcal P_n$ has syndrome $\sigma(\widetilde{E}) = m$, then $\Sigma_{m}$ exactly coincides with a genuine coset $[\widetilde{E}]_{C_{\mathcal{P}_n}(\mathcal S)}$ in the group theory sense.
Importantly, this coset is built from the \textit{centralizer} $C_{\mathcal{P}_n}(\mathcal S)$ of the stabilizer $\mathcal S$ within the Pauli group $\mathcal P_n$.
The coset is given by
\begin{equation}
    [\widetilde{E}]_{C_{\mathcal{P}_n}(\mathcal S)}
    \equiv
\widetilde{E}C_{\mathcal{P}_n}(\mathcal S)
    = \{\widetilde{E} P \mid P \in C_{\mathcal{P}_n}(\mathcal S)\}.
\end{equation}
This holds because $\widetilde{E}$ and $\widetilde{E} P$ have the same syndrome if and only if $P$ commutes with the entire stabilizer.

However, there is a crucial subtlety present here: unlike in the classical setting, the Pauli error pattern $E$ that connects an initial code state $\ket{\psi} \in \mathcal H_C$ to the received state $\ket{\psi'}$ is \textit{not} unique.
This is true for two separate reasons.
First, there is a technicality involving global phases: if $E,E'\in\mathcal P_n$ differ only by a global phase, then there is no physical difference between $E\ket{\psi}$ and $E'\ket{\psi}$.
This is resolved by replacing the Pauli group $\mathcal P_n$ with the Pauli group modulo phases $\overline{\mathcal P}_n$, which is defined as the quotient space $\mathcal P_n / \{I,iI,-I,-iI\}$.
After passing to this quotient, we identify $\mathcal S$ with its image under the quotient map whenever no confusion arises.
Consequently, all previous definitions (e.g., stabilizer group, error cosets, etc.) are updated in the obvious way.

However, even when working with the Pauli group modulo phases, there can still exist $E, E' \in \overline{\mathcal P}_n$ with $E \neq E'$ such that $E\ket{\psi} = E'\ket{\psi} = \ket{\psi'}$ holds for every $\ket{\psi} \in \mathcal H_C$.
The point is that if $E' = ES$ where $S \in \mathcal S$, then $E'\ket{\psi} = ES\ket{\psi} = E\ket{\psi}$, because $\ket{\psi} \in \mathcal H_C$ is held fixed by stabilizer group elements.
This observation motivates defining a different coset called the \textit{stabilizer coset} of an error pattern $E \in \overline{\mathcal P}_n$,
\begin{equation}
[E]_{\mathcal{S}}\equiv
    E\mathcal{S}
    = \{ES \mid S \in \mathcal S\}.
\end{equation}
Every Pauli error pattern in $E\mathcal S$ not only has the same error syndrome, but cannot be distinguished by any physical measurement---including even the final results of a quantum computation.
The converse is also true, but requires a technicality involving the global phase.
As long as $E$ itself is an element of $\overline{\mathcal P}_n$, then $E$ represents an equivalence class of Pauli strings that are identical up to global phase.
In that case, each element $ES$ of the stabilizer coset is also regarded as an element of the Pauli group modulo phases, and the entire collection of all possible stabilizer cosets is mathematically identical to the quotient space $\overline{\mathcal P}_n / \mathcal S$, where the Pauli group modulo phases $\overline{\mathcal P}_n$ is used, rather than the proper Pauli group $\mathcal P_n$.
It is this quotient space, $\overline{\mathcal P}_n/\mathcal S$, which represents physically distinguishable Pauli errors on the code space.

The $4^n$ Pauli operators (modulo global phase) therefore naturally split into $2^{n-k}$ error cosets $\{\Sigma_{m}\}$ based on the $2^{n-k}$ possible error syndromes $m \in \{\pm 1\}^{n-k}$.
Each error coset $\Sigma_{m}$ contains $4^n / 2^{n-k} = 2^{n+k}$ distinct Pauli strings, which form stabilizer cosets containing $|\mathcal S| = 2^{n-k}$ error patterns each.\footnote{The fact that the size of each stabilizer coset matches the total number of error cosets is merely a mathematical coincidence at this level.}
Once a syndrome is fixed, the number of stabilizer cosets consistent with that syndrome is therefore $2^{n+k} / 2^{n-k} = 4^k$.
These stabilizer cosets are called \textit{logical equivalence classes}.
They are physically distinguishable, not at the level of the \textit{encoding}, but rather at the level of \textit{numerical results} for logical computations.

In summary, the $4^n$ elements of the Pauli group modulo phases $\overline{\mathcal P}_n$ can be organized as
\begin{equation}
    4^n = \underbrace{2^{n-k}}_{\text{syndromes}} \times \underbrace{4^k}_{\text{logical classes}} \times \underbrace{2^{n-k}}_{\text{equivalent errors}}.
\end{equation}
It is useful to compare this to the situation in the classical setting.
For a classical $[n,k,d]$ code, there are $2^{n-k}$ distinct error cosets (one per syndrome pattern), and each syndrome pattern is consistent with $2^k$ different logical ``classes."
Each ``class" in the classical setting contains only one error pattern, so that every error pattern is uniquely specified by its syndrome and logical class.
This can be summarized as follows:
\begin{equation}
    2^n = \underbrace{2^{n-k}}_{\text{syndromes}} \times \underbrace{2^k}_{\text{logical ``classes"}} \times \underbrace{1}_{\text{no equivalence classically}}.
\end{equation}

\subsubsection{Binary vector representation}\label[appendix]{app:coordinates}

Pauli errors, stabilizer cosets, and logical equivalence classes all admit a natural coordinate description in the setting of stabilizer codes.
Such a description is useful for specifying probability measures on stabilizer cosets, expressing syndrome decoding algorithms, and analyzing the performance of quantum error correction.

This coordinate description is based on the fact that the Pauli group (up to phase) itself has the structure of a vector space over $\mathbb F_2$.
Specifically, if $t$ is a positive integer, then the multiplication on $\overline{\mathcal P}_t$ is group-isomorphic to the addition on $\mathbb F_2^{2t}$.
Explicitly, if $P \in \overline{\mathcal P}_t$ is an arbitrary Pauli string, then we define its \textit{binary sympectic vector}  (and \textit{$X$ and $Z$ components} of the binary vector),
\begin{equation}
    (x_1,\dots,x_t , z_1,\dots,z_t) \in \mathbb F_2^{2t},
\end{equation}
by selecting each pair $(x_i,z_i)$ based on the $i^{\mathrm{th}}$ tensor factor of $P$, via the correspondence provided in~\cref{tab:pauli-coordinates}.
The equivalence between multiplication on $\overline{\mathcal P}_t$ and addition on $\mathbb F_2^{2t}$ is then straightforward to check.
These coordinates are called the binary symplectic representation in the literature.
\begin{table}[htp]
    \centering
    \caption{Identification of Pauli matrices with $\mathbb F_2$-valued binary vector.}
    \label{tab:pauli-coordinates}
    \begin{tabular}{c|c}
         Pauli Matrix & Binary Vector\\
         \hline
         $I$ & $(0,0)$\\
         $X$ & $(1,0)$\\
         $Y$ & $(1,1)$\\
         $Z$ & $(0,1)$
    \end{tabular}
\end{table}

For an $\llbracket n,k,d\rrbracket$ stabilizer code with stabilizer $\mathcal S$, one can conversely specify a Pauli error $E \in \overline{\mathcal P}_n$ by explicitly providing a list of $X$ and $Z$ components for that error.
Importantly, however, Pauli errors with different coordinates can belong to the same stabilizer coset.
This is because the physically distinguishable errors are elements of the quotient space $\overline{\mathcal P}_n/\mathcal S$, \textit{not} simply the Pauli group modulo phases~$\overline{\mathcal P}_n$.

It turns out that a very natural representation of the quotient space $\overline{\mathcal P}_n/\mathcal S$ exists, thereby enabling easy access to defining and analyzing physically distinct errors.
This representation is based on the previously discussed structure of the Pauli group modulo phases, whose physically inequivalent elements can first be classified by their error syndrome, and then by their logical equivalence class.
The key point is that, after fixing the syndrome~$m \in \{\pm1\}^{n-k}$, the $4^k$ different logical equivalence classes are in bijective correspondence with the \textit{logical} Pauli group modulo phases, $\overline{\mathcal P}_k$.

The bijection is not unique, but it can be specified by two pieces of information: (i) a selection of Pauli operators $\{E_m\}\subseteq \overline{\mathcal P}_n$, one for each syndrome $m \in \{\pm1\}^{n-k}$, such that $\sigma(E_m) = m$, and (ii) a representation of logical operators within $\overline{\mathcal P}_n$.
The latter is obtained by specifying physical representatives $\overline X_i, \overline Z_i \in \overline{\mathcal P}_n$ for the logical Pauli basis, with $i\in\{1,\dots,k\}$.
For $(x_{\mathrm L},z_{\mathrm L})\in\mathbb F_2^{2k}$, define the abstract logical Pauli
\begin{equation}
    P(x_{\mathrm L},z_{\mathrm L})
    \equiv \prod_{j=1}^{k}X_j^{(x_{\mathrm L})_j}Z_j^{(z_{\mathrm L})_j}
    \in \overline{\mathcal P}_k,
\end{equation}
and its chosen physical representative
\begin{equation}
    \overline{P}(x_{\mathrm L},z_{\mathrm L})
    \equiv \prod_{j=1}^{k}\overline X_j^{(x_{\mathrm L})_j}\overline Z_j^{(z_{\mathrm L})_j}
    \in \overline{\mathcal P}_n.
\end{equation}
These two pieces of information together allow us to uniquely write every stabilizer coset $[E]_{\mathcal{S}} \in \overline{\mathcal P}_n/\mathcal S$ in the form
\begin{equation}
   [E]_{\mathcal{S}}
    = 
    E_m\overline{P}(x_{\mathrm L},z_{\mathrm L})\mathcal S,
    \qquad m=\sigma([E]_{\mathcal{S}})\in\{\pm1\}^{n-k}.
\end{equation}
Here, we write $\sigma([E]_{\mathcal{S}})$ to denote the syndrome of any representative of $E\mathcal{S}$; this is well defined because all representatives in a stabilizer coset have the same syndrome.
We can then define a $2^{n-k}$-to-$1$ map $\pi : \overline{\mathcal P}_n/\mathcal S \to \overline{\mathcal P}_k$ by
\begin{equation}
    \pi\left( E_m\overline{P}(x_{\mathrm L},z_{\mathrm L})\mathcal S\right) = P(x_{\mathrm L},z_{\mathrm L}).
\end{equation}
It's not difficult to see that this map is a bijection when restricted to stabilizer cosets having a fixed syndrome~$m \in \{\pm1\}^{n-k}$.

Therefore, the equivalence class of physical errors that any $E \in \overline{\mathcal P}_n$ belongs to is exactly specified by $m=\sigma(E)$ and the abstract logical Pauli $\pi(E\mathcal S)=P(x_{\mathrm L},z_{\mathrm L})$.
It follows that this equivalence class can be uniquely represented by a set of syndrome and logical labels,
\begin{equation}\label{eq:binary-syndrome-logical}
    \left(s_1,\dots,s_{n-k},x_{\mathrm L,1},\dots,x_{\mathrm L,k},z_{\mathrm L,1},\dots,z_{\mathrm L,k}\right) \in \mathbb F_2^{n+k},
\end{equation}
where the $s_1,\dots,s_{n-k} \in \mathbb F_2$ are called the \textit{syndrome vector}, related to the syndrome $m = \sigma(E)$ by $m_i = (-1)^{s_i}$, and the $x_{\mathrm L,1},\dots,x_{\mathrm L,k},z_{\mathrm L,1},\dots,z_{\mathrm L,k} \in \mathbb F_2$ are called \textit{logical labels}, defined to be the binary symplectic representation of $\pi(E\mathcal S) \in \overline{\mathcal P}_k$.
It's easy to see that the binary representations of two different Pauli strings $E,E'\in\overline{\mathcal P}_n$ are identical if and only if $E$ and $E'$ are in the same equivalence class of physically distinguishable errors.
Thus we have a well-defined coordinatization of the quotient space $\overline{\mathcal P}_n/\mathcal S$, whose elements are exactly those equivalence classes.
This labeling with the basis choice is explained in~\cite{Poulin:2006lth,Iyer:2013rry}

\subsection{General Syndrome Decoding}

Here, we discuss general decoding strategies that apply to every classical and quantum code.
We begin by defining the general concept of a syndrome decoder for both the classical and quantum case, and then proceed to discuss the two main types of syndrome decoders that are relevant for this work: \textit{minimum-weight} (MW) decoders and \textit{maximum-likelihood} (ML) decoders.

\subsubsection{Classical and Quantum Decoders}

Consider the usual setup for the classical case, where $C = [n,k,d]$ is a classical code, and $c \in C$ is a transmitted codeword.
A classical error channel applies an error $e \in \mathbb F_2^n$ to produce the received word $c' = c + e$.
The question we now ask is how the syndrome $s \equiv \sigma(c') = \sigma(e) \in \mathbb F_2^{n-k}$ can be used to reconstruct the error $e$, and therefore also the original codeword $c$.

This question is answered by a \textit{syndrome decoder} for $C$, which takes a syndrome $s \in \mathbb F_2^{n-k}$ as input, and returns a plausible error pattern from the error coset $\Sigma_{s}$, called the \textit{candidate}.
In mathematical terms, the syndrome decoder is a function $\mathscr D : \mathbb F_2^{n-k} \to \mathbb F_2^n$ which satisfies the property that $\mathscr D(s) \in \Sigma_{s}$ (or equivalently, that $\sigma(\mathscr D(s)) = s$) for every $s \in \mathbb F_2^{n-k}$.
Once a candidate $e' \equiv \mathscr D(s)$ has been determined by the decoder, a ``recovered word" $r$ can be defined by adding this pattern to the received word, $r = e' + c'$.
If the candidate $e'$ matches the actual error $e$ that occurred, then the correct initial codeword is exactly recovered, $r = e + c' = e + (e + c) = c$.
Conversely, if $e' \neq e$, then $r = e' + c'$ will still be a codeword, but \textit{not} the initial codeword transmitted.

Now consider the usual setup for the quantum case, where we have an $\llbracket n,k,d\rrbracket$ stabilizer code, with stabilizer group $\mathcal S$ and code space $\mathcal H_C$.
We subject this code to a quantum error channel, which we assume to be Pauli noise.
This noise applies a Pauli error string $E \in \overline{\mathcal P}_n$, which modifies a transmitted code state $\ket{\psi} \in \mathcal H_C$ to some received state $\ket{\psi'} \equiv E\ket{\psi}$.
By measuring the stabilizer generators on $\ket{\psi'}$, we obtain measurement outcomes $m \equiv \sigma(E)$, which tell us the error syndrome, $\sigma(E) = m$.
The question we now ask is how the syndrome $m \in \{\pm1\}^{n-k}$ can be used to reconstruct the error $E$, and therefore also the original code state $\ket{\psi}$.

The key starting point is the observation that, actually, $E$ \textit{cannot} be determined uniquely, because there are many equivalent Pauli errors that would have had the exact same effect as $E$ on the code state $\ket{\psi}$.
Therefore, the candidate error returned by a syndrome decoder in the quantum setting should really correspond to a \textit{stabilizer coset}, rather than any specific Pauli error pattern.
In mathematical terms, we therefore define a stabilizer syndrome decoder as a function $\mathscr D : \{\pm 1\}^{n-k} \to \overline{\mathcal P}_n / \mathcal S$, where $\mathscr D(m)$ always returns a stabilizer coset $[E]_{\mathcal{S}} \in \overline{\mathcal P}_n/\mathcal S$ such that $\sigma([E]_{\mathcal{S}}) = m$.

Despite this observation that stabilizer syndrome decoders should return stabilizer cosets rather than specific error patterns, it is sometimes useful, as an abuse of notation, to directly write $\mathscr D(m) = E$, where $E \in \overline{\mathcal P}_n$.
In such cases, the hidden assumption is that $\mathscr D(m)$ actually returns the entire stabilizer coset $E\mathcal S$, and $E$ just serves as a representative shorthand for that coset.
This becomes relevant downstream (at the level of implementation on a quantum computer), because applying a recovery requires choosing a single representative recovery operator from the returned stabilizer coset so that it can be enacted with physical quantum gates.

When working with syndrome and logical coordinates, there is another useful way to describe syndrome decoders for stabilizer codes. 
Recall that a stabilizer coset can be specified by the binary tuple in~\cref{eq:binary-syndrome-logical}.
The job of a syndrome decoder is then to provide candidate logical coordinates $x_{\mathrm L},z_{\mathrm L} \in \mathbb F_2^k$ from the syndrome coordinates $s \in \mathbb F_2^{n-k}$.
Therefore, we may also write the coordinate shorthand $\mathscr D(s) = (x_{\mathrm L},z_{\mathrm L})$ to mean that the decoder returns the stabilizer coset with labels $(s,x_{\mathrm L},z_{\mathrm L})$ on the input syndrome $m \in \{\pm1\}^{n-k}$ defined by $m_i=(-1)^{s_i}$.

\subsubsection{Minimum-Weight Decoding}

The representative-level idea underlying MW decoding, both in the classical and quantum setting, is to select an error pattern with the smallest possible weight among all patterns consistent with a given error syndrome.

For a classical code $C = [n,k,d]$ with parity-check matrix $H$, the MW decoder $\mathscr D^{\MW}$ is defined as
\begin{equation}
    \mathscr D^{\MW}(s) \in \argmin_{\substack{e\in\mathbb F_2^n\\ \sigma(e)=s}} \wt(e).
\end{equation}
If there are ties for the error pattern with least weight, the decoder only needs to select one of them to qualify as an MW decoder.

For a quantum $\llbracket n,k,d\rrbracket$ code with stabilizer $\mathcal S$, the MW decoder is defined as
\begin{align}
    \widehat E^{\MW}(m) &\in \argmin_{\substack{E\in\overline{\mathcal P}_n\\ \sigma(E)=m}} \wt(E),
    \\
    \mathscr D^{\MW}(m) &\equiv \widehat E^{\MW}(m)\mathcal S.
\end{align}
This is a representative-level MW rule: it minimizes the weight of an individual Pauli representative and returns its stabilizer coset, rather than optimizing a degeneracy-aware coset weight.
Similarly to the classical case, if there are ties for the Pauli error pattern with least weight, then the MW decoder can select any of them.

\subsubsection{Maximum-Likelihood Decoding}

While the MW decoder is a useful baseline, it is not always the best possible decoder available, in the sense of maximizing the probability of successful recovery.
In the presence of a specific error channel, the most ideal syndrome decoder is called the \textit{maximum-likelihood} (ML) decoder.
For a classical code, the candidate returned by an ML decoder is the most likely error pattern consistent with the observed syndrome, and therefore depends strongly on the distribution of errors generated by the error channel.

For a classical $[n,k,d]$ code with a prescribed parity-check matrix, if a syndrome $s \in \mathbb F_2^{n-k}$ is observed, then the ML decoder $\mathscr D^{\ML}$ returns the most likely error pattern $e \in \mathbb F_2^n$ with syndrome $\sigma(e) = s$.
In mathematical terms, if $\mathbf{e}$ denotes the (random) bit-flip error drawn from the classical error channel, then
\begin{equation}
    \mathscr D^{\ML}(s) \in \operatorname*{arg\,max}_{\substack{e\in\mathbb F_2^n\\\sigma(e)=s}} \Pr_{\mathbf{e}\sim P}\left[\mathbf{e}=e\mid \sigma(\mathbf{e})=s\right],
\end{equation}
where $P$ denotes the distribution of bit-flip errors generated by the error channel from which $\mathbf{e}$ is drawn.
If there are ties for the most likely error pattern, the decoder is still classified as ML as long as it picks one of the equally likely maximum-likelihood options.

A useful special case of classical ML decoding is when the error channel is a binary symmetric channel (BSC) with physical error rate $0 \le p \le 1$.
In this channel, the error pattern $\mathbf{e}$ is generated by sampling every bit i.i.d. from the distribution $\Ber(p)$ that assigns probability $p$ to the outcome $1$ and probability $1-p$ to the outcome $0$.
For $0<p<1$, if $s \in \mathbb F_2^{n-k}$ is a syndrome with positive probability, then a simple consequence of Bayes' law is the relation
\begin{equation}\label{eq:erbayes}
    \Pr\left[\mathbf{e}=e \mid \sigma(\mathbf{e})=s\right] = \frac{\Pr[\mathbf{e}=e]\delta[\sigma(e) = s]}{\Pr[\sigma(\mathbf{e})=s]},
\end{equation}
where the $\delta[\text{condition}]$ evaluates to the boolean value of the condition, i.e., $0$ if the condition is not satisfied, or $1$ if the condition is satisfied.
Note that the numerator can be expressed as
\begin{align}
    \Pr[\mathbf{e}=e] \delta[\sigma(e) = s] &= p^{\wt(e)}(1-p)^{n-\wt(e)} \delta[\sigma(e)=s]\\
    & \propto \left(\frac{p}{1-p}\right)^{\wt(e)} \delta[\sigma(e)=s],
\end{align}
and it provides the only $e$-dependence in \eqref{eq:erbayes}.
Therefore, the ML decoder returns
\begin{equation}
    \mathscr D^{\ML}(s) \in \operatorname*{arg\,max}_{\substack{e\in\mathbb F_2^n \\ \sigma(e)=s}} \left(\frac{p}{1-p}\right)^{\wt(e)}.
\end{equation}
If $p < 1/2$, then this is a strictly decreasing function in $\wt(e)$, and
\begin{equation}
    \mathscr D^{\ML}(s) \in \argmin_{\substack{e\in\mathbb F_2^n \\ \sigma(e)=s}} \wt(e).
\end{equation}
This implies that classical ML decoding on the BSC with error rate $p < 1/2$ is equivalent to classical MW decoding.

For a quantum $\llbracket n,k,d\rrbracket$ code with stabilizer $\mathcal S$, the ML decoder must return the most likely \textit{stabilizer coset} of Pauli error patterns consistent with the observed syndrome.
Write $\mathbf{E}$ for the random Pauli error generated by a given Pauli noise model, and write $\sigma(\mathbf{E})$ for the stabilizer syndrome obtained by the random Paulis.
With these notations, we introduce the following probabilities:
\begin{align}
    \Pr[E\mid m] &\equiv
    \Pr [\mathbf{E}=E\mid \sigma(\mathbf{E})=m]\,,
    \\
    \Pr[[E]_\mathcal{S}\mid m] &\equiv
    \sum_{E'\in [E]_{\mathcal{S}}} \Pr [\mathbf{E} = E' \mid \sigma(\mathbf{E})=m]\,.
    \label{eq:prob-coset}
\end{align}
Then the ML decoder returns
\begin{equation}
    \mathscr D^{\ML}(m) \in \operatorname*{arg\,max}_{\substack{[E]_\mathcal{S}\in\overline{\mathcal P}_n/\mathcal S\\\sigma(E)=m}} 
     \Pr[[E]_\mathcal{S}\mid m] \,.
\end{equation}
That is, the decoder returns the stabilizer coset most likely to be hit by a random Pauli error generated by the noise distribution, among all stabilizer cosets with the correct syndrome.
Importantly, this coset does \textit{not} necessarily contain the most likely individual Pauli string consistent with the observed syndrome.
Additionally, if there are ties, the ML decoder must select one stabilizer coset among all those tied for the highest likelihood.
Therefore, as a shorthand, we may use the convention that equality ``$=$" denotes the same thing as set membership ``$\in$" for $\argmax$ and $\argmin$.
 
Furthermore, using the fact that the coset $[E]_{\mathcal{S}} \in \overline{\mathcal{P}}_n/\mathcal{S}$ with $\sigma(E) = m $ can be expressed as $E_m\overline{P}(x_{\rm L},z_{\rm L})\mathcal{S}$, 
the probability~$\Pr[[E]_\mathcal{S}\mid m]$~\eqref{eq:prob-coset} is specified by the syndrome and logical labels,
\begin{align}
    \Pr [x_{\rm L},z_{\rm L}\mid s]
    &\equiv
    \Pr [[E_m\overline{P}(x_{\rm L},z_{\rm L})]_{\mathcal{S}} \mid m]
    \\
    &= \sum_{S\in\mathcal{S}}
    \Pr [E_m\overline{P}(x_{\rm L},z_{\rm L})S \mid m]\,,
\end{align}
where $m_i=(-1)^{s_i}$.
Then the ML decoding can be written using the logical label as
\begin{equation} \label{eq:joint-MLD}
    \left(x_{\mathrm L}^{*},z_{\mathrm L}^{*}\right)
    = \operatorname*{arg\,max}_{x_{\mathrm L},z_{\mathrm L}\in\mathbb F_2^k}
    \Pr [x_{\rm L}, z_{\rm L} \mid s]
   \,.
\end{equation}

\subsubsection{Marginalized ML Decoding}\label{subsubsec:marginal-MLD}

In a memory experiment, one may use the eigenstate of the $\overline{Z}$ ($\overline{X}$) operator $\ket{0}_{\mathrm L}$ ($\ket{+}_{\mathrm L}$) as an initial state, and observe if the QEC protocol keeps the initial information.
We call it a $Z$ memory ($X$ memory) experiment/task.
For these specific tasks, only the $X$ ($Z$) components of a logical label are required, and thus the other components can be marginalized.
More specifically, instead of using \emph{joint} ML decoding \eqref{eq:joint-MLD}, it suffices to use \emph{marginalized} ML decoding, defined by
\begin{align}
    x_{\mathrm L}^{*}
    = \operatorname*{arg\,max}_{x_{\mathrm L}\in\mathbb F_2^k}
    \sum_{z_{\mathrm L} \in \mathbb{F}_2^k}
    \Pr\left[x_{\mathrm L}, z_{\mathrm L} \mid s\right]\,,
\end{align}
for the $Z$ memory task, and
\begin{align}
    z_{\mathrm L}^{*}
    = \operatorname*{arg\,max}_{z_{\mathrm L}\in\mathbb F_2^k}
    \sum_{x_{\mathrm L} \in \mathbb{F}_2^k}
    \Pr\left[x_{\mathrm L}, z_{\mathrm L} \mid s\right]\,,
\end{align}
for the $X$ memory task.

\subsection{Maximum-Likelihood Decoding for Concatenated Codes}\label{sec:MLD-concatenated}
Maximum-likelihood decoding for the concatenated code means the ordinary joint logical-class optimization defined above.
Although evaluating this posterior is hard in general, Ref.~\cite{Poulin:2006lth} showed that concatenation permits an efficient message-passing evaluation of the same posterior.
We will focus on the special case used in the main text, where $k_\IN = 1$ so that $C_\OUT \circ C_\IN = \llbracket n,k,d\rrbracket$ with $n=n_\OUT n_\IN$, $k = k_\OUT$, and $d = \min(d_\OUT,d_\IN)$.
To keep the discussion as simple as possible, we will assume the case of i.i.d. single-qubit Pauli noise, since that is all we need for the main text.

First, we will introduce the binary vector representation for concatenated codes.
Namely, the concatenated structure $C_\OUT \circ C_\IN$ admits a hierarchical understanding of how the $X$ and $Z$ components of a Pauli error $E \in \mathcal P_n$ relate to the syndrome vector and logical labels, defined in~\cref{app:coordinates}.

\subsubsection{Binary vector representation with concatenated structure}\label[appendix]{app:hierarchical-coordinatization}

Let $E \in \overline{\mathcal P}_n$ be an arbitrary Pauli error, and write $(x,z) \in \mathbb F_2^{2n}$ for the corresponding binary vector of $E$, where $x,z\in\mathbb F_2^n$ respectively denote the $X$ and $Z$ components.
Let $x_\IN^{(b)}, z_\IN^{(b)} \in \mathbb F_2^{n_\IN}$ respectively denote the segments of the $X$ and $Z$ components supported on the $b^{\mathrm{th}}$ block, where $1\le b\le n_\OUT$.
Recall that the error $E$ can be specified by a syndrome vector and logical labels as $(s,x_{\mathrm L},z_{\mathrm L})$, where the bit string $s \in \mathbb F_2^{n-k}$ corresponds to the syndrome vector, and $x_{\mathrm L},z_{\mathrm L} \in \mathbb F_2^k$ respectively correspond to the $X$ and $Z$ components of the label of the logical equivalence class.

For the syndrome vector, write
\begin{equation}
    s = \left(s_\OUT, s_\IN^{(1)},\dots,s_\IN^{(n_\OUT)}\right),
\end{equation}
where $s_\OUT \in \mathbb F_2^{n_\OUT-k_\OUT}$ corresponds to the outer syndrome, and each $s_\IN^{(b)}\in\mathbb F_2^{n_\IN-k_\IN}$ corresponds to the inner syndrome on the $b^{\mathrm{th}}$ block with $1\le b\le n_\OUT$.
The inner syndrome vector can then be understood as
\begin{equation}
    s_\IN^{(b)} = H_\IN z_\IN^{(b)},
\end{equation}
where $H_\IN$ is the parity-check matrix for $C_\IN$.
The outer syndrome can be similarly understood as
\begin{equation}\label{eq:s_OUT}
    s_\OUT = \left(H_\OUT\otimes [1^{n_\IN}]\right) x,
\end{equation}
where $H_\OUT\otimes [1^{n_\IN}]$ is the Kronecker product of $H_\OUT$ and $[1^{n_\IN}]$.
Here, $H_\OUT$ denotes the parity-check matrix for $C_\OUT$, and $[1^{n_\IN}]$ denotes the $1\times n_\IN$ binary matrix containing only $1$s.\footnote{Technically, this assumes that the inner code $C_\IN$ is a repetition code, but we can assume this without loss of generality because every classical $[n_\IN,k_\IN=1,d_\IN]$ code is equivalent to a repetition code.}

For the $x_{\mathrm L}$ and $z_{\mathrm L}$ coordinates of $E$, recall that we must first choose a representation of $\overline{\mathcal P}_k$ within $\overline{\mathcal P}_n$, along with a collection of syndrome representatives $E_m \in \overline{\mathcal P}_n$ with $\sigma(E_m) = m$ for every syndrome $m \in \{\pm 1\}^{n-k}$.
These specifics cannot be provided explicitly without knowledge of the specific codes $C_\IN$ and $C_\OUT$.

However, some specifics relevant to our discussion can be determined from the following observation: the concatenated code $C_\OUT\circ C_\IN$ has the special property that every stabilizer generator is either a purely $X$-type or purely $Z$-type Pauli string.
Stabilizer codes with this defining property are called \textit{CSS codes} \cite{Calderbank:1995dw,Steane:1995vv}, and an important fact about CSS codes is that logical Pauli operators can be chosen so that every logical $X$ (resp. logical $Z$) operator is represented by a purely $X$-type (resp. $Z$-type) physical operator.

A mathematical consequence of the CSS nature of the code, together with the known concatenated structure, is the guaranteed existence of functions
\begin{align}\label{eq:pi}
    \pi^X_\IN : \mathbb F_2^{n_\IN} \to \mathbb F_2^{k_\IN},\\
    \pi^Z_\IN : \mathbb F_2^{n_\IN} \to \mathbb F_2^{k_\IN},\\
    \pi^X_\OUT : \mathbb F_2^{n_\OUT} \to \mathbb F_2^{k_\OUT},\\
    \pi^Z_\OUT : \mathbb F_2^{n_\OUT} \to \mathbb F_2^{k_\OUT},
\end{align}
such that if $E$ has Pauli coordinates $x_1,\dots,x_n$ and $z_1,\dots,z_n$, then its $x_{\mathrm L}$ and $z_{\mathrm L}$ coordinates can be reconstructed as
\begin{align}
    x_{\mathrm L} &= \pi^X_\OUT\left(\pi^X_\IN\left(x^{(1)}_\IN\right),\dots,\pi^X_\IN\left(x^{(n_\OUT)}_\IN\right)\right),\\
    z_{\mathrm L} &= \pi^Z_\OUT\left(\pi^Z_\IN\left(z^{(1)}_\IN\right),\dots,\pi^Z_\IN\left(z^{(n_\OUT)}_\IN\right)\right),
\end{align}
where again $x^{(b)}_\IN$ and $z^{(b)}_\IN$ respectively denote $b^{\mathrm{th}}$ block of $x$ and $z$.

While the choice of functions \eqref{eq:pi} is not unique, they have a useful interpretation in the context of coordinatizing the logical class for a Pauli error, which can be made precise when specific codes $C_\IN$ and $C_\OUT$ are given.
First, separate an arbitrary Pauli error $E \in \mathcal P_n$ into independent actions on each inner block of qubits.
The action on the $b^{\mathrm{th}}$ block has labels $x^{(b)}_\IN$ and $z^{(b)}_\IN$.
Next, for $1\le b\le n_\OUT$, we define 
$(x_\OUT,z_\OUT)$ with their elements given as
$(x_\OUT)_b \equiv \pi^X_\IN(x^{(b)}_\IN)$ and $(z_\OUT)_b\equiv \pi^Z_\IN(z^{(b)}_\IN)$ respectively.
Then, one can interpret $(x_\OUT,z_\OUT)$ as the $X$ and $Z$ components of the ``inner" logical equivalence class for the piece of $E$ supported on the $b^{\mathrm{th}}$ inner block.
Normally, this would require $\pi^X_\IN$ and $\pi^Z_\IN$ to both depend on the full binary Pauli vector for the portion of $E$ supported on the $b^{\mathrm{th}}$ block, $(x^{(b)}_\IN,z^{(b)}_\IN)$; the existence of an interpretation where $\pi^X_\IN$ only depends on the $X$ components and $\pi^Z_\IN$ only depends on the $Z$ components is a direct consequence of the fact that our inner stabilizer generators are purely $X$-type.
Now, the $X$ and $Z$ components on each inner block can themselves be viewed as the binary vector for a single-qubit Pauli operator (since $k_\IN = 1$) associated with that block.
By concatenating these binary vectors over all $n_\OUT$ blocks, we obtain an effective binary vector $(x_\OUT,z_\OUT) \in \mathbb F_2^{2n_\OUT}$ for a Pauli string on $n_\OUT$ qubits.
By construction, \eqref{eq:s_OUT} can be re-expressed as $s_\OUT=H_\OUT x_\OUT$.
Therefore, similarly to the inner case, we can interpret $\pi^X_\OUT(x_\OUT) \in \mathbb F_2^{k_\OUT}$ and $\pi^Z_\OUT(z_\OUT) \in \mathbb F_2^{k_{\OUT}}$ respectively as the $X$ and $Z$ components of an ``outer" logical class corresponding to the Pauli string $\hat{E} \in \overline{\mathcal P}_{n_\OUT}$ whose binary representation is $(x_\OUT,z_\OUT)$.
As before, this interpretation is possible because the outer stabilizer generators are purely $Z$-type.
Finally, by recognizing that the overall $X$ and $Z$ components for the logical class of $E$ are $x_{\mathrm L}, z_{\mathrm L} \in \mathbb F_2^{k}$ and $k = k_\OUT$, it follows that the outer logical classes are in bijective correspondence with the \textit{actual} logical classes for the concatenated code.
The simplest way to associate them is by simply equating the $X$ and $Z$ components together, i.e.,
\begin{align}
    x_{\mathrm L} &\equiv \pi^X_\OUT(x_\OUT),\\
    z_{\mathrm L} &\equiv \pi^Z_\OUT(z_\OUT).
\end{align}
When taken together with the syndrome vector $s = (s_\OUT,s^{(1)}_\IN,\dots,s^{(n_\OUT)}_\IN)$, this concludes the construction of the syndrome and logical labels of $E$.

When the classical codes $C_\IN$ and $C_\OUT$ are specified, this ``hierarchical coordinatization" of physically distinguishable errors as described above can be made precise by providing all of the following information:
\begin{itemize}
    \item The parity-check matrices, $H_\IN$ and $H_\OUT$;
    \item The coordinatization of inner logical equivalence classes by $\pi^X_\IN$ and $\pi^Z_\IN$;
    \item The coordinatization of outer logical equivalence classes by $\pi^X_\OUT$ and $\pi^Z_\OUT$.
\end{itemize}
Importantly, there is a compatibility requirement between these maps---namely, we need to make sure that the resulting $x_{\mathrm L}$ and $z_{\mathrm L}$ do indeed attribute a unique list of coordinates to each logical class (this will not be satisfied by any arbitrary set of maps).

Note that by providing the coordinate maps, we sidestep the need to provide any syndrome representatives $\{E_m\}$, since the conclusion of our preceding discussion is that the inner and outer maps are themselves sufficient to determine the logical class of every Pauli operator.
Therefore, the $\{E_m\}$ are only needed when it comes time to actually implement quantum gates, since a representative of the equivalence class must be chosen at that point.

\subsubsection{Decoding Strategy}\label[appendix]{app:decoding-strategy}
Having introduced the binary coordinates for concatenated codes, we can now describe maximum-likelihood decoding for the concatenated code.

The starting point, as introduced in the case of ML decoding, is a random stabilizer coset, which takes values in the quotient space $\overline{\mathcal P}_n/\mathcal S$ and is induced by a random physical Pauli error.
This random stabilizer coset carries a random syndrome vector and logical labels,
\begin{equation}
(\mathbf{s}_\OUT,\mathbf{s}_\IN^{(1)},\dots,\mathbf{s}_\IN^{(n_\OUT)},\mathbf{x}_{\mathrm L,1},\dots,\mathbf{x}_{\mathrm L,k},\mathbf{z}_{\mathrm L,1},\dots,\mathbf{z}_{\mathrm L,k}).
\end{equation}

In our setting, the measure on these random syndromes and logical labels is induced by a Pauli noise channel, which generates random Pauli errors $\mathbf{E}$ taking values in $\overline{\mathcal P}_n$, with the corresponding random binary symplectic vector,
\begin{equation}
    (\mathbf{x}_1,\dots,\mathbf{x}_n,\mathbf{z}_1,\dots,\mathbf{z}_n).
\end{equation}
As a useful notation, we write $\mathbf{x}_\IN^{(b)}$ and $\mathbf{z}_\IN^{(b)}$ for the $b^{\mathrm{th}}$ block of the random variables $\mathbf{x}$ and $\mathbf{z}$, respectively.

The goal of any syndrome decoder is to use syndrome data to produce a candidate stabilizer coset with that same syndrome.
Therefore, at the level of syndrome and logical labels, it suffices for a decoder to compute candidate $x_{\mathrm L}$ and $z_{\mathrm L}$ coordinates, given just the syndrome vector $s$.

The ordinary joint logical-class ML decoder achieves this by maximizing the conditional probability distribution~\eqref{eq:joint-MLD}.
The likelihood function can be expanded as
\begin{widetext}
\begin{align}\label{eq:joint}
    \Pr\left[x_{\mathrm L},z_{\mathrm L} \mid s\right] \notag
    &=
   \frac{1}{\Pr[\mathbf{s}_\OUT = s_\OUT \mid s_\IN^{(1)},\dots,s_\IN^{(n_\OUT)}]} \notag
\\ &\times
   \sum_{x_\OUT,z_\OUT\in\mathbb F_2^{n_\OUT}} \delta[x_{\mathrm L}=\pi^X_\OUT(x_\OUT)] \delta[z_{\mathrm L}=\pi^Z_\OUT(z_\OUT)]
   \delta[s_\OUT=H_\OUT x_\OUT] \notag
   \\ &
   \times  \prod_{b=1}^{n_\OUT}
   \Pr[(\mathbf{x}_\OUT)_b=(x_\OUT)_b\wedge(\mathbf{z}_\OUT)_b=(z_\OUT)_b\mid s_\IN^{(b)}],
\end{align}
\end{widetext}
where we defined random variables
\begin{align}
    (\mathbf{x}_\OUT)_b &\equiv \pi^X_\IN\left(\mathbf{x}_\IN^{(b)}\right),\\
    (\mathbf{z}_\OUT)_b &\equiv \pi^Z_\IN\left(\mathbf{z}_\IN^{(b)}\right).
\end{align}
The denominator is just a normalization factor that doesn't depend on $x_{\mathrm L}$ or $z_{\mathrm L}$, so it suffices to optimize the numerator for joint logical-class ML decoding.
This expansion uses the concatenated structure but remains an \emph{exact} evaluation of the posterior in \eqref{eq:joint-MLD}: the inner blocks supply the posterior factors, while the outer sum imposes $s_\OUT=H_\OUT x_\OUT$ and the logical label constraints. This efficient evaluation using the concatenated structure was introduced in~\cite{Poulin:2006lth}.

For memory-protection tasks, the relevant objective is instead task-specific marginal ML decoding.
The $Z$ memory task retains $x_{\mathrm L}$ and sums over $z_{\mathrm L}$, whereas the $X$ memory task retains $z_{\mathrm L}$ and sums over $x_{\mathrm L}$, as explained in \cref{subsubsec:marginal-MLD}.
Note that this choice of marginal is determined by the task, not by the concatenated structure.

For the first case, let us imagine that we marginalize over $z_{\mathrm L}$.
By using \eqref{eq:joint}, the marginal can be reduced to the form
\begin{widetext}
\begin{align}
    \Pr\left[x_{\mathrm L}\mid s\right] &= \sum_{z_{\mathrm L}\in\mathbb F_2^{k_\OUT}}  \Pr\left[x_{\mathrm L},z_{\mathrm L}\mid s\right]\\
    & \propto \sum_{z_{\mathrm L}\in\mathbb F_2^{k_\OUT}} \sum_{x_\OUT,z_\OUT\in\mathbb F_2^{n_\OUT}} \delta[x_{\mathrm L}=\pi^X_\OUT(x_\OUT)] \delta[z_{\mathrm L}=\pi^Z_\OUT(z_\OUT)] \times \delta[s_\OUT=H_\OUT x_\OUT] 
    \\&\times 
    \prod_{b=1}^{n_\OUT} \Pr[(\mathbf{x}_\OUT)_b=(x_\OUT)_b\wedge(\mathbf{z}_\OUT)_b=(z_\OUT)_b\mid s_\IN^{(b)}]\\
    & = \sum_{\substack{x_\OUT\in\mathbb F_2^{n_\OUT}\\s_\OUT=H_\OUT x_\OUT\\x_{\mathrm L}=\pi^X_\OUT(x_\OUT)}} \prod_{b=1}^{n_\OUT} \sum_{(z_\OUT)_b\in\mathbb F_2} \Pr[(\mathbf{x}_\OUT)_b=(x_\OUT)_b\wedge(\mathbf{z}_\OUT)_b=(z_\OUT)_b\mid s_\IN^{(b)}]\\
    & = \sum_{\substack{x_\OUT\in\mathbb F_2^{n_\OUT}\\s_\OUT=H_\OUT x_\OUT\\x_{\mathrm L}=\pi^X_\OUT(x_\OUT)}} \prod_{b=1}^{n_\OUT} \Pr[(\mathbf{x}_\OUT)_b=(x_\OUT)_b\mid s_\IN^{(b)}].
\end{align}
\end{widetext}
This turns out to be simple enough that it can often be computed explicitly for i.i.d. single-qubit Pauli noise.
The $X$ coordinate returned by task-specific marginal ML decoding for $Z$ memory is then~\footnote{The $Z$ coordinate is irrelevant for this task.  If a full logical label is desired, it may be supplied by the optional auxiliary completion
\begin{align}
    z_{\mathrm L}^{*} &\equiv \pi^Z_\OUT\left(z_{\OUT}^{*}\right),\\
    z_{\OUT}^{*} &\in \operatorname*{arg\,max}_{z_{\OUT}\in\mathbb F_2^{n_\OUT}} \prod_{b=1}^{n_\OUT} \Pr\left[\mathbf{z}_{\OUT, b} = z_{\OUT, b} \mid s_\IN^{(b)}\right].
\end{align}
For uncorrelated $X/Z$ noise, combining this hard decision of inner code with the syndrome-constrained outer decision above yields a solution of the joint logical-class ML problem~\eqref{eq:joint-MLD}.
}
\begin{equation}
    x_{\mathrm L}^{*} \in \operatorname*{arg\,max}_{x_{\mathrm L}\in\mathbb F_2^{k_\OUT}} \Pr\left[x_{\mathrm L}\mid s\right].
\end{equation}

Now we come to the alternative task, where marginalization over $x_{\mathrm L}$ is performed instead.
The marginal distribution for $z_{\mathrm L}$ is then given by
\begin{widetext}
\begin{align}
    \Pr\left[z_{\mathrm L}\mid s\right] &= \sum_{x_{\mathrm L}\in\mathbb F_2^{k_\OUT}}  \Pr\left[x_{\mathrm L},z_{\mathrm L}\mid s\right]\\
    & \propto \sum_{x_{\mathrm L}\in\mathbb F_2^{k_\OUT}} \sum_{x_\OUT,z_\OUT\in\mathbb F_2^{n_\OUT}} \delta[x_{\mathrm L}=\pi^X_\OUT(x_\OUT)] \delta[z_{\mathrm L}=\pi^Z_\OUT(z_\OUT)] \times \delta[s_\OUT=H_\OUT x_\OUT] \\&\times  \prod_{b=1}^{n_\OUT} \Pr[(\mathbf{x}_\OUT)_b=(x_\OUT)_b\wedge(\mathbf{z}_\OUT)_b=(z_\OUT)_b\mid s_\IN^{(b)}]\\
    & = \sum_{\substack{x_\OUT,z_\OUT\in\mathbb F_2^{n_\OUT}\\s_\OUT=H_\OUT x_\OUT\\z_{\mathrm L}=\pi^Z_\OUT(z_\OUT)}} 
    \prod_{b=1}^{n_\OUT} \Pr[(\mathbf{x}_\OUT)_b=(x_\OUT)_b\wedge(\mathbf{z}_\OUT)_b=(z_\OUT)_b\mid s_\IN^{(b)}]
\end{align}
\end{widetext}
which puts us in a more complicated setup than before, because the sum cannot be simplified further until a specific code and noise model is given.
Nevertheless, this can be computed when a specific i.i.d. single-qubit Pauli noise model is given.
The $Z$ coordinate returned by task-specific marginal ML decoding for $X$ memory is then
\begin{equation}
    z_{\mathrm L}^{*} \in \operatorname*{arg\,max}_{z_{\mathrm L}\in\mathbb F_2^{k_\OUT}} \Pr\left[z_{\mathrm L}\mid s\right].
\end{equation}

\subsection{Explicit Decoders for Gauss's Law Codes}\label[appendix]{app:explicit-decoders}

Now we come to the specific decoders used in our code capacity demonstrations in~\cref{sec:LER-comparison}.
For the depolarizing noise model (used in our $Z$ memory and $X$ memory experiments), we implement marginalized ML decoding, as discussed in~\cref{sec:MLD-concatenated}.
Recall that the binary vector representation for the concatenated code is defined by the following information:
\begin{itemize}
    \item The parity-check matrices, $H_\IN$ and $H_\OUT$;
    \item The coordinatization of inner logical equivalence classes by $\pi^X_\IN$ and $\pi^Z_\IN$;
    \item The coordinatization of outer logical equivalence classes by $\pi^X_\OUT$ and $\pi^Z_\OUT$.
\end{itemize}

The parity-check matrix $H_\OUT$ is defined directly by the parity checks used by classical Gauss's law code $C_\OUT$, as discussed in~\cref{sec:gauss}.
For $H_\IN$, we use the standard parity checks of the $d$-bit repetition code, namely,
\begin{equation}
    H_\IN \equiv
    \begin{pmatrix}
        1 & 1 & 0 & \cdots & 0 & 0\\
        0 & 1 & 1 & \cdots & 0 & 0\\
        \vdots & \vdots & \vdots & \ddots & \vdots & \vdots\\
        0 & 0 & 0 & \cdots & 1 & 1
    \end{pmatrix}.
\end{equation}
We also assume odd $d$, as in~\cref{sec:compare}.

For the inner logical equivalence classes, we use
\begin{align}
    \pi_\IN^X\left(x_\IN^{(b)}\right) &\equiv
    \begin{cases}
        0, & \wt\left(x_\IN^{(b)}\right) \equiv 0\pmod2,\\
        1, & \wt\left(x_\IN^{(b)}\right) \equiv 1\pmod2,
    \end{cases}\\
    \pi_\IN^Z\left(z_\IN^{(b)}\right) &\equiv
    \begin{cases}
        0, & \wt\left(z_\IN^{(b)}\right) < \frac 12 d,\\
        1, & \wt\left(z_\IN^{(b)}\right) > \frac 12 d,
    \end{cases}\label{eq:inner-logical-coordinates}
\end{align}
for every $x_\IN^{(b)},z_\IN^{(b)}\in\mathbb F_2^{n_\IN}$.
Another way to put this is as follows: define the $b^{\mathrm{th}}$ block ``inner" logical bit-flip and phase-flip operators by
\begin{align}
    \overline{X}^{(b)} &\cong X\otimes I^{\otimes (n_{\IN}-1)},\\
    \overline{Z}^{(b)} &\cong Z^{\otimes n_\IN},
\end{align}
where the RHS is assumed to act solely on the $b^{\mathrm{th}}$ block.
Then the inner logical $X$ coordinate simply tells us whether the bit-flips from a Pauli error can be generated by the inner stabilizer generators (which would imply it is \textit{not} a genuine error) or not---indeed, the two possible answers to this question are exactly separated by $\overline{X}^{(b)}$.
For the inner logical $Z$ coordinate of the error, we simply look at whether a majority or minority of the qubits on block $b$ are affected by a phase-flip, which is also exactly separated by the action of $\overline{Z}^{(b)}$.

For $\pi^X_\OUT$ and $\pi^Z_\OUT$, we defer to the link qubits to define the logical class, since the reduced Pauli group restricted to just the link qubit states is in 1-to-1 correspondence with the $4^N$ ``outer" logical classes.
Specifically, for any bit string $x_\OUT\in\mathbb F_2^{n_\OUT}$, we define $\pi^X_\OUT(x_\OUT)\in\mathbb F_2^{k_\OUT}$ to be the bit string consisting of just the link bits of $x_\OUT$.
For the $Z$ coordinates, let $G_\OUT\in\mathbb F_2^{k_\OUT\times n_\OUT}$ be a full-rank generator matrix for $C_\OUT$, and define
\begin{equation}
    \pi^Z_\OUT(z_\OUT)\equiv G_\OUT z_\OUT.
\end{equation}

Based on these constructions, it's not difficult to see that for any Pauli string with $X$ and $Z$ coordinates $(x,z) \in \mathbb F_2^{2n}$, the logical coordinates
\begin{align}
    x_{\mathrm L} &\equiv \pi^X_\OUT\left(\pi^X_\IN\left(x^{(1)}_\IN\right),\dots,\pi^X_\IN\left(x^{(n_\OUT)}_\IN\right)\right),\\
    z_{\mathrm L} &\equiv \pi^Z_\OUT\left(\pi^Z_\IN\left(z^{(1)}_\IN\right),\dots,\pi^Z_\IN\left(z^{(n_\OUT)}_\IN\right)\right),
\end{align}
exactly track the space of all possible stabilizer cosets, when equipped with arbitrary syndrome representatives $\{E_m\}$.
We thus define and analyze our decoders in the physical coordinate system defined by these logical coordinates (in addition to the syndrome coordinates).

All error models we consider are i.i.d. single-qubit Pauli noise, which takes the form $\mathcal E = \prod_{i=1}^n \mathcal E_i$, with
\begin{equation}
    \mathcal E_i(\rho) = (1-p_X-p_Y-p_Z)\rho + p_X X\rho X + p_Y Y\rho Y + p_Z Z\rho Z,
\end{equation}
where $0\le p_X,p_Y,p_Z\le 1$ are the physical rates of $X$, $Y$, and $Z$ errors, respectively, and $p_X + p_Y + p_Z \le 1$.
If $p_X$, $p_Y$, and $p_Z$ are collectively too large (e.g., with sum greater than $1/2$), then the most ideal decoder would always be forced to assume the unnatural hypothesis that some error has occurred, even in the absence of any nontrivial syndrome.
Therefore, in the following discussion, we will make the additional assumption that the physical error rates are constrained to the regime where errors at each step of decoding are less likely than no error.

\subsubsection{Depolarizing Noise: \texorpdfstring{$Z$}{Z} Memory Experiment}\label[appendix]{app:Z-memory-decoder}

For the case of depolarizing noise, we take $p_X = p_Y = p_Z = \frac13 p$, where $p$ is the overall physical error rate.
We assume $p$ is below the cutoff needed so that the default assumption on each qubit is that an error is less likely to occur than it is to \textit{not} occur.
As we will see, $p < 1/2$ is sufficient.

Importantly, in the case of $Z$ memory, the $Z$ coordinates of the logical class have no effect on the recovery, because the original state is already an eigenstate of every logical $Z$ operator.
Therefore, it suffices to provide the decoding strategy just for the $X$ coordinates of the logical class.
From~\cref{app:decoding-strategy}, recall that we need to find $x_{\mathrm L}$ that maximizes
\begin{equation}
    \Pr\left[x_{\mathrm L}\mid s\right] \propto \sum_{\substack{x_{\OUT}\in\mathbb F_2^{n_\OUT}\\s_\OUT=H_\OUT x_\OUT\\x_{\mathrm L}=\pi^X_\OUT(x_\OUT)}} \prod_{b=1}^{n_\OUT} \Pr[(\mathbf{x}_\OUT)_b=(x_\OUT)_b\mid s_\IN^{(b)}].
\end{equation}
In the case of Gauss's law codes, this simplifies substantially.
Note that the constraint $x_{\mathrm L} = \pi^X_\OUT(x_\OUT)$ fixes the link bits of $x_\OUT$, and the constraint $s_\OUT = H_\OUT x_\OUT$ subsequently fixes the site bits.
Therefore, only one term contributes in the sum, and the decoder objective simplifies to
\begin{align}
    x_{\mathrm L}^* &= \pi^X_\OUT\left(x_\OUT^*\right),\\
    x_\OUT^* &\equiv \argmax_{\substack{x_\OUT \in \mathbb F_2^{n_\OUT}\\s_\OUT = H_\OUT x_\OUT}} \prod_{b=1}^{n_\OUT} \Pr[(\mathbf{x}_\OUT)_b=(x_\OUT)_b\mid s_\IN^{(b)}].
\end{align}

This can be computed explicitly by the following observation:
\begin{equation}\label{eq:conditional-equal-likelihood}
    \Pr\left[(\mathbf{x}_\OUT)_b = 0 \mid \mathbf{s}_\IN^{(b)} \neq 0\right] = \Pr\left[(\mathbf{x}_\OUT)_b = 1 \mid \mathbf{s}_\IN^{(b)} \neq 0\right] = \frac 12.
\end{equation}
That is, if the inner syndrome $s_\IN^{(b)}$ on a given block is nontrivial, then the $X$ coordinate $(\mathbf{x}_\OUT)_b \equiv \pi^X_\IN\left(\mathbf{x}_\IN^{(b)}\right)$ for the logical class on that inner block becomes a fair coin toss.
The reason for this is that a nontrivial inner block syndrome $s_\IN^{(b)} = H_\IN z_\IN^{(b)} \neq 0$ implies that at least one of the $Z$ coordinates in the vector $z_\IN^{(b)}$ is nonzero.
If, say, the $j^{\mathrm{th}}$ coordinate is nonzero, then the Pauli substring acting on the $b^{\mathrm{th}}$ block \textit{cannot} have its $j^{\mathrm{th}}$ tensor factor as a pure bit-flip, $X_j$.
Instead, this factor must either be $Y_j$ or $Z_j$, both of which have exactly the same likelihood, since their difference is invisible to the inner block syndrome vector $s_\IN^{(b)}$.
But the two options $Y_j$ and $Z_j$ also correspond to the two different inner logical $X$ coordinates, because they contribute different parities of bit-flip weights, thereby implying \eqref{eq:conditional-equal-likelihood}.

To compute the probabilities conditioned on the trivial inner block syndrome,
\begin{align}
    \Pr\left[(\mathbf{x}_\OUT)_b = 0 \mid \mathbf{s}_\IN^{(b)} = 0\right],\\
    \Pr\left[(\mathbf{x}_\OUT)_b = 1 \mid \mathbf{s}_\IN^{(b)} = 0\right],
\end{align}
we can use the relation
\begin{align}
    \Pr\left[(\mathbf{x}_\OUT)_b = 1\right] = & \Pr\left[(\mathbf{x}_\OUT)_b = 1 \mid \mathbf{s}_\IN^{(b)} = 0\right] \times \Pr\left[\mathbf{s}_\IN^{(b)} = 0\right]\\
    & + \Pr\left[(\mathbf{x}_\OUT)_b = 1 \mid \mathbf{s}_\IN^{(b)} \neq 0\right] \times \Pr\left[\mathbf{s}_\IN^{(b)} \neq 0\right].
\end{align}
The LHS is computing the probability that the random bit string $\mathbf{x}_\IN^{(b)}$ has odd Hamming weight.
Under depolarizing noise, each $X$ coordinate is independently nonzero with probability $p_X + p_Y = \frac23p$, because both $X$ and $Y$ errors have nonzero $X$ coordinate.
Therefore, this is precisely the probability that a $d$-bit error generated by the BSC with physical bit-flip rate $p_{\BF} \equiv \frac 23 p$ has odd weight.
This is a well-known elementary combinatorial problem with answer
\begin{equation}
    \Pr\left[(\mathbf{x}_\OUT)_b = 1\right] = \frac 12\left[1 - f(p_{\BF})^d\right],
\end{equation}
where $f(x) \equiv 1 - 2x$.
Meanwhile, the inner syndrome is trivial if and only if all of the $Z$ coordinates on the block are equal.
Since each $Z$ coordinate of the error is independently nonzero with probability $p_{\PF} \equiv p_Y + p_Z = \frac 23p$, we have
\begin{equation}
    \Pr\left[\mathbf{s}_\IN^{(b)}=0\right] = g(p_{\PF}),
\end{equation}
where $g(x) \equiv (1-x)^d + x^d$.
Using~\eqref{eq:conditional-equal-likelihood} in the preceding relation and solving for the remaining conditional probability gives
\begin{widetext}
\begin{align}
    \Pr\left[(\mathbf{x}_\OUT)_b=1\mid\mathbf{s}_\IN^{(b)}=0\right] &= \frac{\Pr[(\mathbf{x}_\OUT)_b=1] - \frac12\Pr[\mathbf{s}_\IN^{(b)}\neq0]}{\Pr[\mathbf{s}_\IN^{(b)}=0]}\\
    &= \frac12\left[1- \frac{f(p_{\BF})^d}{g(p_{\PF})}\right].
\end{align}
\end{widetext}
Similarly,
\begin{equation}
    \Pr\left[(\mathbf{x}_\OUT)_b=0\mid\mathbf{s}_\IN^{(b)}=0\right]=\frac12\left[1+\frac{f(p_{\BF})^d}{g(p_{\PF})}\right].
\end{equation}
Thus, conditioned on a trivial inner syndrome, the inner logical $X$ coordinate $0$ is more likely than the inner logical $X$ coordinate $1$.
Conditioned on a nontrivial inner syndrome, the two coordinates are equally likely.

Now define $w_0(s_\IN)$ to be the number of blocks $b$ such that $s_\IN^{(b)} = 0$, and define $w_0(x_\OUT,s_\IN)$ to be the number of blocks $b$ such that $s_\IN^{(b)} = 0$ and $(x_\OUT)_b = 1$.
Then it follows that
\begin{widetext}
\begin{align}
    x_\OUT^* &= \argmax_{\substack{x_\OUT \in \mathbb F_2^{n_\OUT}\\s_\OUT = H_\OUT x_\OUT}} \frac 1{2^{n_\OUT}} \left[1- \frac{f(p_{\BF})^d}{g(p_{\PF})}\right]^{w_0(x_\OUT,s_\IN)} \left[1+\frac{f(p_{\BF})^d}{g(p_{\PF})}\right]^{w_0(s_\IN) - w_0(x_\OUT,s_\IN)}\\
    & = \argmax_{\substack{x_\OUT \in \mathbb F_2^{n_\OUT}\\s_\OUT = H_\OUT x_\OUT}} \left(\frac{1- f(p_{\BF})^d / g(p_{\PF})}{1+f(p_{\BF})^d / g(p_{\PF})}\right)^{w_0(x_\OUT,s_\IN)}\\
    & = \argmin_{\substack{x_\OUT \in \mathbb F_2^{n_\OUT}\\s_\OUT = H_\OUT x_\OUT}} w_0(x_\OUT,s_\IN),
\end{align}
\end{widetext}
where we made the assumption that $f(p_{\BF})^d / g(p_{\PF}) > 0$, which holds as long as $p < 1/2$.
In other words, an algorithm for the ML decoder with the concatenation structure is obtained by searching for the $x_\OUT \in \mathbb F_2^{n_\OUT}$ that minimizes the number of bits turned on \textit{at the block indices where a trivial inner syndrome is observed}, conditioned on having the correct outer syndrome.
This can be thought of as a modification to the classical MW decoder, which would have simply searched for $x_\OUT \in \mathbb F_2^{n_\OUT}$ to minimize the number of bits turned on \textit{at all}, conditioned on having the correct syndrome.

\subsubsection{Depolarizing Noise: \texorpdfstring{$X$}{X} Memory Experiment}\label[appendix]{app:X-memory-decoder}

In the case of $X$ memory, it is now $x_{\mathrm L}$ which does not need to be considered.
Apart from this change, we proceed by similar logic to the previous case, and continue using the same notation.
We need to find $z_{\mathrm L}$ that maximizes the marginal likelihood
\begin{widetext}
\begin{equation}
    \Pr\left[z_{\mathrm L}\mid s\right] \propto \sum_{\substack{x_\OUT,z_\OUT\in\mathbb F_2^{n_\OUT}\\s_\OUT=H_\OUT x_\OUT\\z_{\mathrm L}=\pi^Z_\OUT(z_\OUT)}} \prod_{b=1}^{n_\OUT} \Pr[(\mathbf{x}_\OUT)_b=(x_\OUT)_b\wedge(\mathbf{z}_\OUT)_b=(z_\OUT)_b\mid s_\IN^{(b)}].
\end{equation}
\end{widetext}
In fact, as we will show, it turns out that
\begin{equation}
    z_{\mathrm L}^* \equiv \argmax_{z_{\mathrm L}\in\mathbb F_2^{k}} \Pr\left[z_{\mathrm L}\mid s\right] = 0,
\end{equation}
holds exactly for $p < 1/2$.
In light of \eqref{eq:inner-logical-coordinates}, this implies that the ML decoding with the concatenated structure can be accomplished by simply using an MW decoder for the classical $d$-bit repetition code to correct phase-flip errors independently on each inner block.

The first observation to this end is that if we define
\begin{equation}
    q_b\left(t';t\right) \equiv \Pr\left[(\mathbf{z}_\OUT)_b = t' \mid (\mathbf{x}_\OUT)_b = t\wedge\mathbf{s}_\IN^{(b)} = s_\IN^{(b)}\right],
\end{equation}
then
\begin{equation}\label{eq:rel-cond-block-prob}
    q_b\left(0;t\right) > q_b\left(1;t\right),
\end{equation}
for all choices of $t \in \mathbb F_2$ and $s_\IN^{(b)} \in \mathbb F_2^{n_\IN-k_\IN}$.
Intuitively, this makes sense: $(\mathbf{z}_\OUT)_b = 0$ corresponds to a Pauli error with low weight in the $Z$ coordinates, which is more likely to occur than an error with high weight in the $Z$ coordinates.
The complication is that this holds true even after we condition on knowledge of the inner syndrome $\mathbf{s}_\IN^{(b)}$ \textit{and} the inner logical $X$ coordinate $(\mathbf{x}_\OUT)_b$, irrespective of their values.
Despite the complication, this can easily be verified (by simply checking all possible cases).

Now write
\begin{align}
    R(x_\OUT) &\equiv \prod_{b=1}^{n_\OUT} r_b((x_\OUT)_b),\\
    r_b(t) &\equiv \Pr\left[(\mathbf{x}_\OUT)_b=t\mid s_\IN^{(b)}\right],
\end{align}
so that
\begin{widetext}
\begin{align}
    \Pr\left[z_{\mathrm L}\mid s\right] = \sum_{\substack{x_\OUT,z_\OUT\in\mathbb F_2^{n_\OUT}\\s_\OUT=H_\OUT x_\OUT\\z_{\mathrm L}=\pi^Z_\OUT(z_\OUT)}} R(x_\OUT)\prod_{b=1}^{n_\OUT} q_b(1;(x_\OUT)_b)^{(z_\OUT)_b} q_b(0;(x_\OUT)_b)^{1-(z_\OUT)_b}.
\end{align}
\end{widetext}
Note that
\begin{widetext}
\begin{equation}
    \sum_{\substack{z_\OUT\in\mathbb F_2^{n_\OUT}\\z_{\mathrm L}=\pi^Z_\OUT(z_\OUT)}} \prod_{b=1}^{n_\OUT} q_b(1;(x_\OUT)_b)^{(z_\OUT)_b} q_b(0;(x_\OUT)_b)^{1-(z_\OUT)_b} = \Pr\left[\mathbf{z}_{\mathrm L}=z_{\mathrm L}\mid \mathbf{x}_\OUT=x_\OUT\wedge \mathbf{s}_\IN = s_\IN\right],
\end{equation}
\end{widetext}
which we simply denote as $p(x_{\mathrm L} \mid x_\OUT,s_\IN)$ for brevity.
Therefore,
\begin{equation}
    \Pr\left[z_{\mathrm L}\mid s\right] \propto \sum_{\substack{x_\OUT\in\mathbb F_2^{n_\OUT}\\s_\OUT=H_\OUT x_\OUT}} R(x_\OUT) p(z_{\mathrm L} \mid x_\OUT,s_\IN).
\end{equation}
We will show that
\begin{equation}
    p(0 \mid x_\OUT,s_\IN) \ge p(z_{\mathrm L} \mid x_\OUT,s_\IN)
\end{equation}
for every choice of $x_\OUT$ and $s_\IN$.
Since every $R(x_\OUT)$ is clearly non-negative, this implies that $\Pr[0\mid s] \ge \Pr[z_{\mathrm L}\mid s]$ for every $z_{\mathrm L}$, establishing the desired conclusion that $z_{\mathrm L}^* = 0$ suffices for the ML decoding with the concatenated structure.

To proceed further, let us define $C_\OUT^{\perp}$ to be the dual code of $C_\OUT$.
Then $C_\OUT^{\perp}$ is also the row span of $H_\OUT$.
More importantly, since $\mathbf{z}_{\mathrm L} = \pi_\OUT^{Z}(\mathbf{z}_\OUT) = G_\OUT \mathbf{z}_\OUT$, it follows that inequivalent samples of $\mathbf{z}_{\mathrm L}$ exactly correspond to samples of $\mathbf{z}_\OUT$ that are drawn from different $C_\OUT^{\perp}$-cosets.
At an intuitive level, all this is saying is that two $Z$-type Pauli strings are in the same logical equivalence class if their quotient is in the group of outer stabilizer generators.
Mathematically, the statement is that for any $z_{\mathrm L} \in \mathbb F_2^{k_\OUT}$, one can assign $v_{z_{\mathrm L}} \in \mathbb F_2^{n_\OUT}$ such that
\begin{equation}
    \pi^Z_\OUT\left(z_\OUT\right) = z_{\mathrm L} \iff z_\OUT\in v_{z_{\mathrm L}} + C_\OUT^{\perp},
\end{equation}
where $v_{z_{\mathrm L}} + C_\OUT^{\perp}$ denotes the coset of $C_\OUT^{\perp}$ given by translating with $v_{z_{\mathrm L}}$.
Since every $z_{\mathrm L}$ now uniquely corresponds to some coset defined by $v_{z_{\mathrm L}}$, we are thus reduced to an optimization problem over these cosets.

To simplify the notation, fix $x_\OUT \in \mathbb F_2^{n_\OUT}$ and let $\hat{\mathbf{Z}}$ denote the random variable taking values in $\mathbb F_2^{n_\OUT}$ with bits $\hat{\mathbf{Z}}_b$ drawn i.i.d. from the distribution
\begin{equation}
    \Pr\left[\hat{\mathbf{Z}}_b = \hat{Z}_b\right] = q_b(\hat{Z}_b;(x_\OUT)_b),
\end{equation}
for $1\le b\le n_\OUT$.
This immediately implies
\begin{equation}
    p\left(z_{\mathrm L}\mid x_\OUT,s_\IN\right) = \Pr\left[\hat{\mathbf{Z}} \in v_{z_{\mathrm L}} + C_\OUT^{\perp}\right].
\end{equation}
As mentioned, the question reduces to the relative likelihoods for the cosets in which $\hat{\mathbf{Z}}$ can live.
We will show that the highest probability mass resides in the zero coset, i.e., $C_\OUT^{\perp}$ itself.

For this, let $v \in \mathbb F_2^{n_\OUT}$ be arbitrary, and define the indicator function of the desired event,
\begin{equation}
    \IND_v\left(\hat{Z}\right) \equiv
    \begin{cases}
        1, & \hat{Z} \in v + C_\OUT^{\perp},\\
        0, & \hat{Z} \notin v + C_\OUT^{\perp}.
    \end{cases}
\end{equation}
We then employ the commonly used expansion of this indicator function,
\begin{equation}
    \IND_v\left(\hat{Z}\right) = \frac 1{|C_\OUT|} \sum_{u\in C_\OUT} (-1)^{u\cdot\left(v+\hat{Z}\right)},
\end{equation}
where we technically rely on the fact that the dual of the dual of $C_\OUT$ is again $C_\OUT$.
Then it follows
\begin{align}
    \Pr\left[\hat{\mathbf{Z}} \in v + C_\OUT^{\perp}\right] &= \mathbb E\left[\IND_v\left(\hat{\mathbf{Z}}\right)\right]\\
    & = \frac1{|C_\OUT|}\sum_{u\in C_\OUT} (-1)^{u\cdot v} \mathbb E\left[(-1)^{u\cdot\hat{\mathbf{Z}}}\right]\\
    & = \frac1{|C_\OUT|}\sum_{u\in C_\OUT} (-1)^{u\cdot v} \prod_{b=1}^{n_\OUT}\mathbb E\left[(-1)^{u_b\hat{\mathbf{Z}}_b}\right],
\end{align}
where we used independence of the bits of $\hat{\mathbf{Z}}$.
Conditioned on the value of the bit $u_b$, it's easy to see that
\begin{equation}
    \mathbb E\left[(-1)^{u_b\hat{\mathbf{Z}}_b}\right] =
    \begin{cases}
        1, & u_b = 0,\\
        q_b(0;(x_\OUT)_b) - q_b(1;(x_\OUT)_b), & u_b = 1.
    \end{cases}
\end{equation}

Now we come to the critical point of the argument.
Recall that there was a slight bias towards lower weight $Z$ coordinates on physical Pauli errors, even when conditioned on arbitrary observation of the syndrome and $X$ coordinates.
This culminated in the mathematical result \eqref{eq:rel-cond-block-prob}, which holds irrespective of the values of $(x_\OUT)_b$ and the measured inner syndrome.
But this implies that the above expectation value $\mathbb E[(-1)^{u_b\hat{\mathbf{Z}}_b}]$ is \textit{positive} (in fact, we only need it to be non-negative for the following argument to carry through).
Therefore,
\begin{widetext}
\begin{equation}
    \frac1{|C_\OUT|}\sum_{u\in C_\OUT} (-1)^{u\cdot v} \prod_{b=1}^{n_\OUT}\mathbb E\left[(-1)^{u_b\hat{\mathbf{Z}}_b}\right] \le \frac1{|C_\OUT|}\sum_{u\in C_\OUT} \prod_{b=1}^{n_\OUT}\mathbb E\left[(-1)^{u_b\hat{\mathbf{Z}}_b}\right],
\end{equation}
\end{widetext}
since the RHS just directly adds up all the positive contributions where the LHS might attach a minus sign.
But the RHS is precisely the $v=0$ case, i.e., probability that $\hat{\mathbf{Z}}$ lives in the zero coset of $C_\OUT^{\perp}$.
This establishes the desired result.

\subsubsection{Uncorrelated X/Z Noise: Arbitrary Memory}\label[appendix]{app:arbitrary-memory-decoder}

In the case of uncorrelated X/Z noise, there exist parameters $p'_X$ and $p'_Z$ with $0 \le p'_X, p'_Z \le 1$ such that\footnote{In the main text, we assume $p'_X = p'_Z = p$. Here, we consider the more general case.}
\begin{align}
    p_X &= p_X'(1-p_Z'),\\
    p_Y &= p_X'p_Z',\\
    p_Z &= p_Z'(1-p_X').
\end{align}
This implies that the single-qubit channel $\mathcal E_i$ can be factored as $\mathcal E_i = \mathcal E^X_i \mathcal E^Z_i$, where
\begin{align}
    \mathcal E^X_i(\rho) &= (1-p_X')\rho + p_X' X \rho X,\\
    \mathcal E^Z_i(\rho) &= (1-p_Z')\rho + p_Z' Z \rho Z,
\end{align}
and therefore that $X$ and $Z$ errors are applied \textit{independently} (i.e., in an uncorrelated fashion) for each qubit.

This uncorrelated nature makes ML decoding relatively straightforward.
The point is simply that the joint law factorizes,
\begin{equation}
    \Pr\left[x_{\mathrm L},z_{\mathrm L}\mid s\right] = \Pr\left[x_{\mathrm L}\mid s\right]\Pr\left[z_{\mathrm L}\mid s\right],
\end{equation}
and each factor is separately equivalent to the corresponding marginal distribution in its own right.
In fact, the absence of correlations implies further that
\begin{align}
    \Pr\left[x_{\mathrm L}\mid s\right] &= \Pr\left[x_{\mathrm L} \mid \mathbf{s}_\OUT = s_\OUT\right],\\
    \Pr\left[z_{\mathrm L}\mid s\right] &= \Pr\left[z_{\mathrm L} \mid \mathbf{s}_\IN = s_\IN\right].
\end{align}

For $x_{\mathrm L}$, we have
\begin{widetext}
\begin{align}
    \Pr\left[(\mathbf{x}_\OUT)_b=(x_\OUT)_b\mid s_\IN^{(b)}\right] &= \Pr\left[(\mathbf{x}_\OUT)_b=(x_\OUT)_b\right]\\
    & = \frac12 \times
    \begin{cases}
        1 + f(p_\BF)^d, & (x_\OUT)_b = 0,\\
        1 - f(p_\BF)^d, & (x_\OUT)_b = 1,
    \end{cases}
\end{align}
\end{widetext}
where $p_\BF \equiv p_X + p_Y = p_X'$ is the physical bit-flip error rate, and $f(x) = 1-2x$.
This follows from the elementary combinatorial result for the distribution of the parity of the weight of a random $d$-bit error string generated from the BSC with physical error rate $p_\BF$.
Now it follows, using similar logic to~\cref{app:Z-memory-decoder}, that
\begin{align}
    x_{\mathrm L}^* &= \argmax_{x_{\mathrm L} \in \mathbb F_2^{k_\OUT}} \sum_{\substack{x_\OUT\in\mathbb F_2^{n_\OUT}\\s_\OUT=H_\OUT x_\OUT\\x_{\mathrm L}=\pi^X_\OUT(x_\OUT)}} \prod_{b=1}^{n_\OUT} \Pr[(\mathbf{x}_\OUT)_b=(x_\OUT)_b\mid s_\IN^{(b)}]\\
    & = \pi^X_\OUT\left(x_\OUT^*\right),
\end{align}
where
\begin{widetext}
\begin{align}
    x_\OUT^* &\equiv \argmax_{\substack{x_\OUT\in\mathbb F_2^{n_\OUT}\\ s_\OUT = H_\OUT x_\OUT}} \frac1{2^{n_\OUT}} \left[1 - f(p_\BF)^d\right]^{\wt(x_\OUT)} \left[1 + f(p_\BF)^d\right]^{n_\OUT - \wt(x_\OUT)}\\
    & = \argmin_{\substack{x_\OUT\in\mathbb F_2^{n_\OUT}\\ s_\OUT = H_\OUT x_\OUT}} \wt(x_\OUT).
\end{align}
\end{widetext}
In other words, we have an \textit{exact} MW decoder for the logical $X$ coordinates, provided that the bit-flip rate $p_\BF < 1/2$.

For $z_{\mathrm L}$, the marginal simplifies as
\begin{equation}
    \Pr\left[z_{\mathrm L} \mid s\right] \propto \sum_{\substack{z_\OUT\in\mathbb F_2^{n_\OUT}\\z_{\mathrm L}=\pi^Z_\OUT(z_\OUT)}} \prod_{b=1}^{n_\OUT} \Pr[(\mathbf{z}_\OUT)_b=(z_\OUT)_b\mid s_\IN^{(b)}].
\end{equation}
Each factor counts the probability that a BSC with physical error rate $p_\PF \equiv p_Y + p_Z = p_Z'$ generates a $d$-bit error string with weight less than $d/2$ (if $(z_\OUT)_b = 0$) or more than $d/2$ (if $(z_\OUT)_b = 1$).
In analogy with~\cref{app:X-memory-decoder}, we define
\begin{equation}
    q_b(t') \equiv \Pr[(\mathbf{z}_\OUT)_b=t'\mid s_\IN^{(b)}],
\end{equation}
with the explicit form
\begin{align}
    q_b(0) &= \sum_{w=0}^{(d-1)/2} \binom{d}{w} p_\PF^w (1-p_\PF)^{d-w},\\
    q_b(1) &= 1 - q_b(0),
\end{align}
which arises from summing the likelihood of every allowed weight of the $Z$ coordinate at block $b$.
It's not too hard to see that $q_b(0) > q_b(1)$, provided that $p_\PF < 1/2$, which puts us in the same position as in~\cref{app:X-memory-decoder}.
We will use this fact to show that $z_{\mathrm L}^* = 0$, which implies that we have MW decoding for phase-flip errors, independently on every block.

Following the argument from~\cref{app:X-memory-decoder} closely, we define the random variable $\hat{\mathbf{Z}}$ taking values in $\mathbb F_2^{n_\OUT}$ with i.i.d. bits drawn from the distribution
\begin{equation}
    \Pr\left[\hat{\mathbf{Z}}_b = \hat{Z}_b\right] = q_b(\hat{Z}_b).
\end{equation}
Then it again holds that
\begin{equation}
    \Pr\left[z_{\mathrm L}\mid s\right] \propto \Pr\left[\hat{\mathbf{Z}} \in v_{z_{\mathrm L}} + C_\OUT^{\perp}\right],
\end{equation}
where each $v_{z_{\mathrm L}} \in \mathbb F_2^{n_\OUT}$ serves as a coset identifier corresponding to $z_{\mathrm L}$.
The exact same argument from~\cref{app:X-memory-decoder} goes through, and we find
\begin{equation}
    \Pr\left[\hat{\mathbf{Z}} \in v_{z_{\mathrm L}} + C_\OUT^{\perp}\right] \le \Pr\left[\hat{\mathbf{Z}} \in C_\OUT^{\perp}\right],
\end{equation}
establishing that $z_{\mathrm L} = 0$.
\section{Exact Calculation for Logical Error Rate}\label[appendix]{app:exact-probability}

Now we come to the exact calculation of the logical error rate (LER) used in~\cref{sec:analytical}, for an optimal $1$D Gauss's law code $C_\OUT \circ C_\IN$ where $C_\OUT = [n_\OUT,k_\OUT,d]$ is a classical Gauss's law code for bit-flips and $C_\IN = [d,1,d]$ is a classical repetition code for phase-flips.
We focus on the case of uncorrelated X/Z noise, and use the strict version of the LER \eqref{eq:strict-LER}.
Additionally, to reproduce the results in~\cref{fig:uncorrelated-XZ}, it suffices to assume $d$ is odd.
This calculation is essentially a continuation from~\cref{app:arbitrary-memory-decoder}, once we have selected the ML decoder based on marginal likelihoods.
Thus, we adopt the full notation conventions introduced in~\cref{app:syndrome-decoding}.

The strict LER is given by the probability that the logical coordinates $(\mathbf{x}_{\mathrm L}^*,\mathbf{z}_{\mathrm L}^*)$ determined from the ML decoding \textit{do not exactly match} the logical coordinates $(\mathbf{x}_{\mathrm L},\mathbf{z}_{\mathrm L})$ drawn randomly from the distribution induced by the random Pauli error $\mathbf{E}$.
Note that we embolden $\mathbf{x}_{\mathrm L}^*$ and $\mathbf{z}_{\mathrm L}^*$ because they are now themselves random variables, since the output of the decoder is a function of the random Pauli error.
In the uncorrelated X/Z channel, the $X$ and $Z$ components of $\mathbf{E}$ are independent, and therefore the events $\mathbf{x}_{\mathrm L}^* = \mathbf{x}_{\mathrm L}$ and $\mathbf{z}_{\mathrm L}^* = \mathbf{z}_{\mathrm L}$ are independent.
Therefore, the strict LER can be expressed as
\begin{equation}
    \LER = 1 - p_\OUT \cdot p_\IN,
\end{equation}
where
\begin{align}
    p_\OUT &\equiv \Pr\left[\mathbf{x}_{\mathrm L}^* = \mathbf{x}_{\mathrm L}\right],\\
    p_\IN &\equiv \Pr\left[\mathbf{z}_{\mathrm L}^* = \mathbf{z}_{\mathrm L}\right].
\end{align}

Let's start with the case of inner decoding.
Recall that for $p_\PF < 1/2$, the inner ML decoding step is exactly equivalent to MW decoding on the BSC with physical error rate $p_\PF$.
Define $P_{C_\IN}(p)$ to be precisely the probability that an MW decoder succeeds for the classical code $C_\IN$ (i.e., the $d$-bit repetition code) on the BSC with physical error rate $p$.
Then we have, for instance,
\begin{equation}
    P_{C_\IN}(p_\PF) \equiv \Pr\left[(\mathbf{z}_\OUT)_b = 0\right],
\end{equation}
because the inner decoder always prefers $(z_\OUT)_b = 0$ on every inner block.
By iterating over all possible weights $0\le w\le (d-1)/2$ on which the MW decoder would succeed, it's not hard to see that
\begin{equation}
    P_{C_\IN}(p) = \sum_{w=0}^{(d-1)/2} \binom{d}{w} p^w (1-p)^{d-w}.
\end{equation}

Now, to get to $p_\IN$, note that the event $\mathbf{z}_{\mathrm L}^* = \mathbf{z}_{\mathrm L}$ occurs if and only if the pattern of inner blocks on which the MW decoder succeeds perfectly corresponds to an element of the dual code, $C_\OUT$.
In other words,
\begin{align}\label{eq:pin-exact}
    p_\IN &= \Pr\left[\mathbf{z}_\OUT \in C_\OUT^{\perp}\right]\\
    & = \sum_{z_\OUT \in C_\OUT^{\perp}} P_{C_\IN}(p_\PF)^{n_\OUT-\wt(z_\OUT)} \cdot \left[1-P_{C_\IN}(p_\PF)\right]^{\wt(z_\OUT)},
\end{align}
by summing the respective likelihoods over all the ways this can happen.

Let us assume we are dealing with an $N$-site lattice, and write $d = 2t+1$ for the code distance.
Then the sum \eqref{eq:pin-exact} can be interpreted as the partition function for a $1$D system of binary variables with nearest-neighbor interactions, which can thus be evaluated exactly with a transfer matrix.

To understand this picture, we must first understand the structure of the dual code $C_\OUT^{\perp}$.
Recall that the classical Gauss's law code $C_\OUT$ used in the optimal construction requires duplicating site variables $t$ times into the classical binary variables $S_i^{(j)} \in \mathbb F_2$ for $1\le i\le N$ and $1\le j\le t$; links are not duplicated, but we still denote $L_i^{(1)} \in \mathbb F_2$ for the classical binary variable corresponding to the $i^{\mathrm{th}}$ link, for $1\le i\le N$.
Any proposed parity check $h \in \mathbb F_2^{n_\OUT}$ can be indexed by these duplicated (or not) site and link variables, where $h_{S_i^{(j)}} = 1$ indicates that $S_i^{(j)}$ participates in the parity check, and $h_{L_i^{(1)}} = 1$ indicates that $L_i^{(1)}$ participates in the parity check.
The corresponding linear constraint for a given $h \in \mathbb F_2^{n_\OUT}$ would be
\begin{equation}
    \sum_{i=1}^{N} \left(h_{L_i^{(1)}}L_i^{(1)} + \sum_{j=1}^t h_{S_i^{(j)}} S_i^{(j)}\right) = 0.
\end{equation}
Not all such checks are elements of the dual code $C_\OUT^{\perp}$.
A convenient basis for $C_\OUT^{\perp}$ is provided by
\begin{align}
    h^{(i)} &: L_{i-1}^{(1)} + S_{i}^{(1)} + L_{i}^{(1)} = 0,\\
    h^{(i,j)} &: S_{i}^{(1)} + S_{i}^{(j)} = 0,
\end{align}
where $1\le i\le N$ and $2\le j\le t$, and indices are identified periodically modulo $N$, so that $L_0^{(1)} \equiv L_N^{(1)}$.
The $h^{(i)}$ checks enforce Gauss's law constraints, while the $h^{(i,j)}$ checks enforce duplication of the site variables.
Write $\alpha_i$ and $\beta_{ij}$ as the binary coefficients of any $h \in C_\OUT^{\perp}$ when expanded in this basis, so that
\begin{equation}
    h = \sum_{a=1}^{N} \left(\alpha_i h^{(i)} + \sum_{b=2}^{t} \beta_{ij} h^{(i,j)}\right).
\end{equation}
The key observation from this basis expansion is that an arbitrary word $h \in \mathbb F_2^{n_\OUT}$ is an element of the dual code $C_\OUT^{\perp}$ if and only if there exist binary variables $\alpha_1,\dots,\alpha_N \in \mathbb F_2$ such that
\begin{align}
    h_{L_{i}^{(1)}} &= \alpha_{i} + \alpha_{i+1},\\
    \sum_{j=1}^t h_{S_{i}^{(j)}} &= \alpha_i,
\end{align}
for every $1\le i\le N$, where indices are again identified periodically modulo $N$, so that $\alpha_{N+1} \equiv \alpha_1$.

With these points noted, the transfer-matrix procedure for computing \eqref{eq:pin-exact} can now be provided.
For any $z_\OUT\in C_\OUT^{\perp}$, assign a weight contribution to each bit of $z_\OUT$ as follows: if the bit is $0$, then assign weight $W(0) \equiv P_{C_\IN}(p_\BF)$; if the bit is $1$, then assign weight $W(1) \equiv 1 - W(0)$.
In other words, define the function
\begin{equation}\label{eq:W(zb)}
    W((z_\OUT)_b) \equiv
    \begin{cases}
        P_{C_\IN}(p_\BF), & (z_\OUT)_b = 0,\\
        1 - P_{C_\IN}(p_\BF), & (z_\OUT)_b = 1.
    \end{cases}
\end{equation}
The full weight of $z_\OUT$ is then the product of the weight contribution from each bit, and the inner success probability $p_{\IN}$ can be understood as a ``partition function" given by summing these weights over every $z_\OUT$,
\begin{align}
    p_{\IN} &\equiv \sum_{z_\OUT \in C_\OUT^{\perp}} P_{C_\IN}(p_\BF)^{n_\OUT-\wt(z_\OUT)} \cdot \left[1-P_{C_\IN}(p_\BF)\right]^{\wt(z_\OUT)}\\
    & = \sum_{z_\OUT\in C_\OUT^{\perp}} \prod_{b=1}^{n_\OUT} W((z_\OUT)_b)\\
    & = \sum_{z_\OUT\in C_\OUT^{\perp}} \prod_{i=1}^N \left[W\left((z_\OUT)_{L_i^{(1)}}\right) \prod_{j=1}^t W\left((z_\OUT)_{S_i^{(j)}}\right)\right]
\end{align}
In order to evaluate this, it benefits us to use the coordinatization of $z_\OUT\in C_\OUT^{\perp}$ in terms of the components $\alpha_i$ and $\beta_{ij}$, as introduced previously.
This change of coordinates, valid on the dual code, is given explicitly by
\begin{align}
    (z_\OUT)_{L_i^{(1)}} &\equiv \alpha_i + \alpha_{i+1},\\
    (z_\OUT)_{S_i^{(1)}} &\equiv \alpha_i + \sum_{j=2}^t \beta_{ij},\\
    (z_\OUT)_{S_i^{(j)}} &\equiv \beta_{ij},
\end{align}
where $1\le i\le N$ and $2\le j\le t$.
Then we can write
\begin{widetext}
    \begin{align}
        p_{\IN} &= \sum_{\alpha} \sum_{\beta} \prod_{i=1}^N \left[W(\alpha_i+\alpha_{i+1}) W\left(\alpha_i+\sum_{j=2}^t \beta_{ij}\right) \prod_{j=2}^t W(\beta_{ij})\right]\\
        & = \sum_{\alpha} \prod_{i=1}^N \sum_{\beta_{i 2},\dots,\beta_{i t}} \left[W(\alpha_i+\alpha_{i+1}) W\left(\alpha_i+\sum_{j=2}^t \beta_{ij}\right) \prod_{j=2}^t W(\beta_{ij})\right].
    \end{align}
\end{widetext}
In the first line, the sum runs over all binary coefficient tensors $\alpha \in \mathbb F_2^{N}$ and $\beta \in \mathbb F_2^{N\times(t-1)}$, where the second index on $\beta$ begins at $2$.
The second line follows from the distributive property, where $\beta_{i 2},\dots,\beta_{i t}$ are dummy variables that are independently summed over $\mathbb F_2$.

Now fix an index $1\le i\le N$, and consider just the $\beta$-dependent part of the sum in the $i^{\mathrm{th}}$ term of the product,
\begin{equation}
    F_i \equiv \sum_{\beta_{i 2},\dots,\beta_{i t}} W\left(\alpha_i + \sum_{j=2}^t \beta_{i j}\right) \prod_{j=2}^t W(\beta_{i j}). 
\end{equation}
To evaluate $F_i$, define $\gamma_1, \dots, \gamma_t \in \mathbb F_2$ by
\begin{equation}
    \gamma_a \equiv
    \begin{cases}
        \alpha_i + \sum_{j=2}^t \beta_{i j}, & a = 1,\\
        \beta_{i a}, & 2 \le a \le t.
    \end{cases}
\end{equation}
Then the choice of $\beta_{i 2},\dots,\beta_{i t} \in \mathbb F_2$ uniquely corresponds to a choice of $\gamma_1,\dots,\gamma_t \in \mathbb F_2$ subject to the constraint $\sum_{j=1}^t \gamma_a = \alpha_i$.
In other words,
\begin{equation}
    F_i \equiv \sum_{\substack{\gamma_1,\dots,\gamma_t\\ \gamma_1 + \cdots + \gamma_t = \alpha_i}} \prod_{a=1}^t W(\gamma_a).
\end{equation}
From the definition of the weight function \eqref{eq:W(zb)}, it is clear that $F_i$ is computing precisely the probability that $t$ random bits $\gamma_1,\dots,\gamma_t$ drawn i.i.d. from the distribution $\mathrm{Bernoulli}(W(1))$ have total parity $\gamma_1+\cdots+\gamma_t = \alpha_i$.
Therefore, we can invoke the well-known answer to this problem as used repeatedly in~\cref{app:explicit-decoders}; namely,
\begin{equation}
    F_i \equiv \frac 12 \times
    \begin{cases}
        1 + f(W(1))^t, & \alpha_i = 0,\\
        1 - f(W(1))^t, & \alpha_i = 1,
    \end{cases}
\end{equation}
where $f(x) = 1 - 2x$.
For brevity, let us denote this by $F(\alpha_i) \equiv F_i$, since it depends entirely on the value of the bit $\alpha_i$.

Now the partition function simplifies to
\begin{equation}
    p_{\IN} = \sum_{\alpha} \prod_{i=1}^N W(\alpha_i + \alpha_{i+1}) F(\alpha_i),
\end{equation}
which corresponds to a periodic, translation-invariant $1$D system with only nearest-neighbor interactions, where configurations are specified by the bit string $\alpha$.
This can immediately be computed by defining the transfer matrix
\begin{equation}
    M_{uv} \equiv W(u+v) F(v),
\end{equation}
for indices $u,v\in \mathbb F_2$.
The result is
\begin{equation}
    p_{\IN} = \tr(M^N) = \lambda_+^N + \lambda_-^N,
\end{equation}
where $\lambda_{\pm}$ are the eigenvalues of the transfer matrix
\begin{equation}
    M =
    \begin{pmatrix}
        W(0)F(0) & W(1)F(1)\\
        W(1)F(0) & W(0)F(1)
    \end{pmatrix}.
\end{equation}
Explicitly, these eigenvalues are given by
\begin{equation}
    \lambda_{\pm} = \frac12\left[W(0)\pm\sqrt{W(1)^2+(W(0)-W(1))^{2t+1}}\right],
\end{equation}
which now fully specifies the exact calculation of $p_{\IN}$.

Now we move onto the calculation of $p_\OUT$.
This requires understanding when exactly the event $\mathbf{x}_{\mathrm L}^* = \mathbf{x}_{\mathrm L}$ occurs.
Based on the construction in~\cref{app:arbitrary-memory-decoder}, it's not difficult to see that $\mathbf{x}_{\mathrm L}^* = \pi_\OUT^X(\mathbf{x}_\OUT^*)$, where $\mathbf{x}_\OUT^*$ is the minimum-weight solution to
\begin{equation}\label{eq:Hx*=Hx}
    H_\OUT\mathbf{x}_\OUT^* = H_\OUT \mathbf{x}_\OUT,
\end{equation}
or equivalently, the minimum weight solution to $\mathbf{x}_\OUT^* + \mathbf{x}_\OUT \in C_\OUT$.
Note that the bit string $\mathbf{x}_\OUT^*$ is a random variable because $\mathbf{x}_\OUT$ is a random variable.

On the condition that \eqref{eq:Hx*=Hx} is satisfied, the event $\mathbf{x}_{\mathrm L}^* = \mathbf{x}_{\mathrm L}$ occurs if and only if $\mathbf{x}_\OUT^* = \mathbf{x}_\OUT$.
This is because the equivalence $\mathbf{x}_{\mathrm L}^* = \mathbf{x}_{\mathrm L}$ forces the link bits of $\mathbf{x}_\OUT^*$ and $\mathbf{x}_\OUT$ to agree, and \eqref{eq:Hx*=Hx} then ensures the syndromes agree, thereby forcing the site bits to agree as well.
In other words,
\begin{equation}
    p_\OUT = \Pr\left[\mathbf{x}_\OUT^* = \mathbf{x}_\OUT\right].
\end{equation}
Importantly, this is \textit{not} simply calculating the probability that $\mathbf{x}_\OUT$ is the minimum-weight representative of its syndrome class.
This is because there can be \textit{several} minimum-weight representatives, of which $\mathbf{x}_\OUT^*$ will only be one of them.
An important concern that can be raised, then, is whether the success probability depends on the tie-breaking mechanism chosen by the MW decoder for $C_\OUT$.
As we will see, there is no such dependence---the success probability is well-defined without choosing any specific MW decoder explicitly.

We can view $p_\OUT$ as the success probability $P_{C_\OUT}(p)$ of the classical MW decoder on the classical Gauss's law code $C_\OUT$, subject to a BSC with physical error rate $p$ given by
\begin{align}
    p &\equiv \Pr\left[(\mathbf{x}_\OUT)_b = 1\right]\\
    & = \frac 12\left[1 - f(p_\BF)^d\right],
\end{align}
which is identical for every block index $1\le b\le n_\OUT$.
We have thus reduced the problem to a calculation in the setting of classical error correction.

For any syndrome $s_\OUT \in \mathbb F_2^{n_\OUT-k_\OUT}$, denote $A(s_\OUT)$ to be the minimum weight of any error pattern having the syndrome $s_\OUT$, and denote $B(A)$ to be the number of distinct syndromes $s_\OUT$ such that $A(s_\OUT) = A$.
Without loss of generality, assume that for each $s_\OUT$, the classical MW decoder deterministically returns the candidate error pattern $e(s_\OUT) \in \mathbb F_2^{n_\OUT}$.
Then the classical MW decoder succeeds if and only if the true error pattern $e \in \mathbb F_2^{n_\OUT}$ happens to be drawn from the set
\begin{equation}
    S \equiv \left\{e(s_\OUT) \mid s_\OUT \in \mathbb F_2^{n_\OUT}\right\}.
\end{equation}
If $\mathbf{e}$ denotes the random error pattern sampled from the BSC, then
\begin{align}
    P_{C_\OUT}(p) &= \Pr\left[\mathbf{e} \in S\right]\\
    & = \sum_{s_\OUT \in \mathbb F_2^{n_\OUT-k_\OUT}} \Pr\left[\mathbf{e} = e(s_\OUT)\right]\\
    & = \sum_{s_\OUT \in \mathbb F_2^{n_\OUT-k_\OUT}} p^{A(s_\OUT)} (1-p)^{n_\OUT - A(s_\OUT)}\\
    & = \sum_{A=0}^{A_\MAX} B(A) p^A (1-p)^{n_\OUT-A},
\end{align}
where
\begin{equation}
    A_\MAX \equiv \max_{s_\OUT \in \mathbb F_2^{n_\OUT-k_\OUT}} A(s_\OUT).
\end{equation}

For general classical codes, evaluating $A_\MAX$ and $B(A)$ is believed to be computationally difficult~\cite{McLoughlin1984, GuruswamiMicciancioRegev2005}.
Nevertheless, for the classical Gauss's law code $C_\OUT$ on an $N$-site lattice with distance $d = 2t+1$, the exact values can be computed systematically.

\end{document}